\documentclass[12pt]{article}

\usepackage{amsthm}

\usepackage[colorlinks=true,linkcolor=blue, anchorcolor=blue,
            citecolor=blue, urlcolor=blue]{hyperref}

\usepackage{adjustbox,amssymb,amsmath,amsfonts,array, bbm, booktabs, eurosym,geometry,ulem,graphicx,caption,color,setspace,comment,footmisc,caption,pdflscape,subfigure,subcaption,array,mathtools,bm,paralist,cleveref,float,titling,fancyhdr,tabularx,tikz,outlines,xcolor}

\usepackage[longnamesfirst]{natbib}
\newtheorem{assumption}{Assumption}[section]
\newtheorem{theorem}{Theorem}[section]
\newtheorem{lemma}{Lemma}[section]

\newtheorem{proposition}[theorem]{Proposition}

\theoremstyle{definition}

\newtheorem{remark}{Remark}[section]

\newcommand{\ind}{\perp\!\!\!\!\perp} 
\newcommand{\E}{\mathbb{E}}
\newcommand{\R}{\mathbb{R}}
\newcommand{\Var}{\mathrm{Var}}
\newcommand{\Cov}{\mathrm{Cov}}

\newcommand{\indicator}{\mathbbm{1}}

\DeclareMathAlphabet{\pazocal}{OMS}{zplm}{m}{n}

\newcommand{\abs}[1]{\left\lvert #1 \right\rvert}

\newcolumntype{L}[1]{>{\raggedright\let\newline\\arraybackslash\hspace{0pt}}m{#1}}
\newcolumntype{C}[1]{>{\centering\let\newline\\arraybackslash\hspace{0pt}}m{#1}}
\newcolumntype{R}[1]{>{\raggedleft\let\newline\\arraybackslash\hspace{0pt}}m{#1}}

\usepackage[colorinlistoftodos]{todonotes}
\usepackage{mdframed}
\newmdenv[
  backgroundcolor=yellow!30,
  linecolor=yellow!80!black,
  skipabove=\topsep,
  skipbelow=\topsep
]{fixmeblock}

\newcommand{\HE}{\mathrm{HE}}
\newcommand{\HO}{\mathrm{HO}}
\newcommand{\CW}{\mathrm{CW}}
\newcommand{\FE}{\mathrm{FE}}

\definecolor{myBlue}{HTML}{2B5DA5} 
\definecolor{myRed}{HTML}{C23B22}    
\definecolor{myGreen}{HTML}{1F8F5F}  
\definecolor{myBrown}{HTML}{8C564B}
\definecolor{myPurple}{HTML}{9467BD}
\definecolor{myOrange}{HTML}{FF7F0E}
\definecolor{myGray}{HTML}{7F7F7F}

\newcolumntype{L}[1]{>{\footnotesize\raggedright\arraybackslash}p{#1}} 
\newcolumntype{Y}{>{\raggedright\arraybackslash}X}   
\newcolumntype{s}{>{\scriptsize\raggedright\arraybackslash}X}
\newcolumntype{S}{>{\tiny\raggedright\arraybackslash}X} 
\newcolumntype{Z}{>{\hsize=2\hsize \tiny\raggedright\arraybackslash}X}

\newcommand{\keywords}[1]{%
  \noindent\texttt{Key words:}\ \hangindent=1.8em #1
}

\begin{document}

\setlength{\droptitle}{-5em} 

\author{%
{ Tian Xie}\\
\\
{\normalsize Department of Economics, University College London }
}

\title{%
Automatic Inference for Value-Added Regressions\thanks{ \hyphenpenalty=10000\exhyphenpenalty=10000
  I thank Tim Christensen, Raffaella Giacomini, and Liyang Sun for continuous guidance and generous support throughout this project. I would also like to thank Luis Alvarez, Debopam Bhattacharya, Kirill Borusyak, Pedro Carneiro, Jiafeng Chen, Xiaohong Chen, Andrew Chesher, Max Cytrynbaum, Ben Deaner, Aureo de Paula, Hugo Freeman, Christophe Gaillac, Joao Granja, Sukjin Han, Stephen Hansen, Toru Kitagawa, Yuichi Kitamura, Roger Koenker, Dennis Kristensen, Soonwoo Kwon, Simon Lee, Attila Lindner, Oliver Linton, Eric Mbakop, Lars Nesheim, Alexei Onatski, Matthias Parey, Kirill Ponomarev, Franck Portier, Silvia Sarpietro, Michela Tincani, Gabriel Ulyssea, Edward Vytlacil, Guanyi Wang, Weining Wang, Ming Yang, Andrei Zeleneev, and participants of the IAAE Conference (2025), Bristol Econometric Study Group Conference (2025), 2nd UCL-CeMMAP Ph.D. Econometrics Research Day, Midwest Econometrics Group Conference (2025), 4th International Econometrics PhD Conference Erasmus University Rotterdam, $(\text{EC})^2$ Conference (2025), ES European Winter Meeting (2025), and seminars at UCL, Yale University, and University of Cambridge. All errors are my own. Email address: \href{mailto:tian.xie.20@ucl.ac.uk}{tian.xie.20@ucl.ac.uk}.
}
}

\date{\normalsize \today}

\maketitle

\begin{center}
  \vspace{-1em} 
  \large \href{https://tian-xie.com/papers/jmp.pdf}{\textnormal{Click here for the latest version}}
\end{center}

\begin{abstract}
\begin{singlespace}
    A large empirical literature regresses outcomes on empirical Bayes shrinkage estimates of value-added, yet little is known about whether this approach leads to unbiased estimates and valid inference for the downstream regression coefficients. We study a general class of empirical Bayes estimators and the properties of the resulting regression coefficients. We show that estimators can be asymptotically biased and inference can be invalid if the shrinkage estimator does not account for heteroskedasticity in the noise when estimating value added. By contrast, shrinkage estimators properly constructed to model this heteroskedasticity perform an automatic bias correction: the associated regression estimator is asymptotically unbiased, asymptotically normal, and efficient in the sense that it is asymptotically equivalent to regressing on the true (latent) value-added. Further, OLS standard errors from regressing on shrinkage estimates are consistent in this case. As such, efficient inference is easy for practitioners to implement: simply regress outcomes on shrinkage estimates of value-added that account for noise heteroskedasticity.
\end{singlespace}
\end{abstract}
\smallskip

\keywords{empirical Bayes, shrinkage estimators, teacher value-added, \\errors in variables.}

\smallskip

\noindent\texttt{JEL classification codes:} C12.

\newpage

\section{Introduction}Empirical Bayes shrinkage estimators are widely applied in settings with a large number of latent individual effects observed via noisy measurements. {Economic applications include the study of teachers' 
value-added to student test scores \citep[e.g.][]{chetty2014measuring,chetty2014measuring2}, location effects \citep[e.g.][]{chetty2018impacts2}, hospital quality \citep[e.g.][]{hull2018estimating}, 
nursing home quality \citep[e.g.][]{einav2025producing}, firm-level discrimination 
\citep[e.g.][]{kline2022systemic}, among others.} Shrinkage estimates are typically not the final object, but instead serve as inputs to downstream analyses. A typical downstream use is to regress other variables of economic interest on the estimated individual effects.

It is common for researchers to treat the shrinkage estimates as if they were the true individual effects when performing conventional inference in the downstream regression. This widespread empirical approach implicitly assumes the property of automatic inference, i.e., no adjustment is required to achieve nominal coverage rates of confidence intervals. However, empirical Bayes shrinkage estimates generally differ from the latent true effects, introducing measurement error. Moreover, there is also the potential for a generated regressor problem \citep{pagan1984econometric}, because the regressors are themselves generated in a first-stage estimation step. Thus, there may be threats to the validity of both the estimator (via measurement error bias), and inference (via incorrect standard errors).

In empirical applications, various linear shrinkage estimators have been used as regressors. A leading example is the literature studying how teachers' value-added for test scores impacts students' future outcomes. Data on students' test scores are aggregated into shrinkage estimates, which are constructed as a weighted average of each teacher's individual mean score and all teachers' pooled mean score. Most studies employ individualized shrinkage, where noisier estimates (e.g., from smaller classes) are shrunk more aggressively toward the pooled mean. {Empirical work using this approach includes \citet{jacob2007parents}, \citet{jacob2008can}, \citet{kane2008estimating}, \citet{chandra2016health}, \citet{jackson2018test}, \citet{abdulkadirouglu2020parents}, \citet{bau2020teacher}, \citet{biasi2022flexible}, \citet{warnick2024instructor}, \citet{andrabi2025heterogeneity}, and \citet{angelova2025algorithmic}.} Individualized shrinkage differs from equal shrinkage approaches, such as \citet{chetty2014measuring,chetty2014measuring2}, which apply a uniform shrinkage rule to all units. However, the theoretical impact of these different shrinkage schemes on downstream estimates and inference has not been studied to date.

This paper is the first to formally study the impact of different shrinkage schemes on the properties of downstream regressions, and to characterize when automatic inference does or does not obtain. We focus on linear shrinkage estimators, which are widely used in practice due to their simplicity and ease of implementation. We develop a general econometric framework for analyzing a broad class of individualized shrinkage estimators as regressors.  These estimators are weighted averages whose weights reflect signal-to-noise ratios estimated from the data. They can be classified as individual-weight (individualized shrinkage) or common-weight (equal shrinkage). Our central finding shows that automatic inference hinges on the treatment of heteroskedastic measurement error. In particular, failing to account for heteroskedasticity in the shrinkage weights leads to \emph{invalid} downstream inference, with asymptotically biased estimators and coverage rates from conventional confidence intervals below nominal coverage. By contrast, automatic inference obtains if the weights properly account for this heteroskedasticity. In this case, the conventional confidence intervals achieve nominal coverage without further correction, and the resulting downstream coefficient estimator is asymptotically equivalent to the infeasible OLS regression on the true latent effects. Thus, there is no efficiency loss from regressing on the estimated individual effects instead of the true individual effects.

These findings have important practical implications for current practice. Many empirical implementations model heteroskedastic measurement error solely by the number of measurements per individual, e.g., class sizes. Yet idiosyncratic noise variances may also differ across individuals. We show that accounting only for heterogeneity in the number of measurements can yield invalid inference if noise variances are correlated with the number of measurements. This is because the individualized adjustments systematically under- or over-correct depending on the degree of heterogeneity. This, in turn, introduces a non-classical measurement error, where the bias can lead to amplification rather than attenuation. Inference remains valid when the two sources of heterogeneity are independent. Properly accounting for heteroskedasticity in the shrinkage weights ensures robustness to such dependence and allows automatic inference for the downstream regression coefficient. 

Accordingly, our analysis provides straightforward practical guidance for implementation. In the teacher value-added example, the weights may vary across teachers and are determined by the estimated signal-to-noise ratio. Components of the ratio are estimated by within- and across-teacher variances, which properly account for heteroskedasticity in the measurement error. In predicting long-run outcomes as a function of latent value-added, we simply regress the outcomes on these shrinkage estimates. From the reported results, conventional OLS standard errors and confidence intervals are valid for inference.

Our results do not impose distributional assumptions and have a clear intuition. Linear shrinkage estimators are the empirical Bayes posterior mean under normality of both the measurement error and the latent individual effect \citep{efron1973stein, morris1983parametric}. However, while the functional form of shrinkage is derived under normality, we study regressions in a semiparametric context and none of our results rely on normality. The overall intuition is as follows. In the presence of noise, it is well known that regressions on unshrunk individual means suffer from a classical measurement error problem. Shrinkage estimators reduce the variability of the estimated individual effects and thereby offset the variance inflation that comes from noise. What is not clear is whether they offset it by just the right amount, so that downstream regression estimators are asymptotically unbiased and downstream inference is valid. We show that for some scenarios this is the case, while for others it is not.

To formalize our analysis, we first restate the established baseline result that the unshrunk individual mean suffers from classical errors-in-variables, resulting in attenuation bias and invalid inference when used as a downstream regressor. We then show our main results on individual-weight shrinkage. First, we outline the precise conditions under which inference fails, even when heterogeneity in the number of measurements is accounted for. Second, we establish how a properly constructed, fully individualized shrinkage estimator achieves asymptotic unbiasedness, inferential validity, and asymptotic efficiency in the downstream regression. Finally, we analyze the common-weight estimator as a point of comparison. Because of its close connection to conventional bias-correction methods, this estimator also addresses the inference problem. Nevertheless, we focus our attention on the individual-weight approach because it is more used in contemporary work and no theoretical results currently exist to justify its use in downstream regressions. While our focus has been on individual-weight shrinkage that assumes (as is standard) independence between individual effects and variance of measurement error, our results can be generalized to relax this assumption. Recent work has derived empirical Bayes estimators allowing such dependence \citep{chen2023mean}, but has not used them in downstream regressions. In on-going work, we extend our analysis to this new class of estimators and show that they also yield automatic inference for downstream regressions.

We conduct Monte Carlo simulations to demonstrate the finite-sample performance of the different shrinkage estimators when used as regressors. The results confirm our asymptotic theory. The unshrunk individual means yield substantial attenuation bias and under-coverage of confidence intervals. Individual-weight shrinkage that accounts only for heterogeneity in the number of measurements can also lead to bias and under-coverage when noise variances are correlated with the number of measurements. In contrast, heteroskedastic individual-weight shrinkage that accounts for both sources of heterogeneity yields valid coverage, and achieves very close performance to regressing on the true latent individual effects. Common-weight shrinkage also performs well, but is outperformed by heteroskedastic individual-weight shrinkage in finite samples when the sample size is moderate. 

We illustrate our results using two empirical applications. The first revisits the firm discrimination data of \citet{kline2022systemic}, while the second is based on the school value-added data of \citet{andrabi2025heterogeneity}. In the first application, we consider a regression of future discrimination levels on estimated discrimination at the firm level. We find that individual-weight heteroskedastic shrinkage yields confidence intervals more strongly supportive of a positive predictive effect. In the second application, we regress private school fees on school value-added estimates. We find that individual-weight heteroskedastic shrinkage produces slightly different confidence intervals from those reported in the original study, and it reaffirms the positive effect of school value-added on private school fees.

\paragraph{Related Literature} 
This paper relates to an older literature that studies the use of common-weight shrinkage estimators as downstream regressors. Regressing on common-weight shrinkage estimators can be traced back to \citet{whittemore1989errors}, who demonstrates via simulations that common-weight shrinkage yields consistent downstream coefficients. That paper also points out informally its equivalence to bias-correction methods. Building on this insight, \citet{guo2012mean} provide theoretical justifications regarding the quadratic risk reduction of such estimators for downstream coefficients. \citet{chetty2014measuring,chetty2014measuring2} construct common-weight shrinkage estimators based on the best linear predictor, emphasizing that the estimator is bias-free when used as the regressor. Their estimator can be viewed as equivalent to IV methods, and we unify this IV perspective along with the bias-correction perspective of common-weight shrinkage in \Cref{sec:CW}. \citet{deeb2021framework} further exploits the equivalence to IV methods and develops inference results for \citet{chetty2014measuring,chetty2014measuring2}'s estimator, highlighting the need to adjust conventional OLS standard errors to account for errors in estimation of individual effects and nuisance parameters. While this line of work clarifies inference in the common-weight shrinkage setting, it essentially relies on the equivalence to IV and does not extend to individual-weight shrinkage. These results are not applicable to the analysis of individual-weight shrinkage, which we study.

We formalize and extend prior work on individual-weight shrinkage. Consistency of regression estimators based on individual-weight shrinkage has been briefly discussed in \citet{jacob2007parents}, \citet{kane2008estimating}, \cite{abdulkadirouglu2020parents}, \citet{walters2024empirical}, and \citet{andrabi2025heterogeneity}. In particular, \citet{walters2024empirical} highlights independence between individual effects and the variance of the noise as a necessary condition for using widespread individualized shrinkage strategies. Beyond such consistency results, there has been no work studying the consequences for downstream inference of using individual-weight shrinkage estimators as regressors. We show that valid inference depends critically on how heteroskedasticity is treated in the shrinkage weights. We establish asymptotic unbiasedness, asymptotic normality, and consistency of standard errors for the downstream regression estimator based on properly constructed individual-weight shrinkage.

This paper is also related to non-shrinkage options for various downstream models. \citet{Chang2024PostEB} leverage information from the estimated prior to develop correction methods for nonlinear downstream models. For regressions, this literature strand is largely concerned with bias-correction. As noted from the equivalence, discussions from this category also apply to common-weight shrinkage. \citet{kline2020leave} and \citet{bonhomme2024estimating} emphasize correct estimation of moments for value-added that solves the errors-in-variables problem. \citet{chen2025regression} show that standard bias-correction methods for the downstream regression remain consistent even if the individual effects and measurement error are correlated. We contribute to this literature by establishing conditions under which conventional shrinkage estimators yield valid inference for downstream regressions when individual effects and noise are independent. In ongoing work, we extend the analysis to this correlated case.

\paragraph{Outline}

This paper proceeds as follows. 
\Cref{sec:setup_J} introduces the framework and its implications. 
\Cref{sec:setup} discusses a broad class of shrinkage estimators, establishes their asymptotic normality and the consistency of standard errors, and discusses the validity of standard inference approaches for the downstream regression. It also provides some practical implementation guidelines.
\Cref{sec:sim} reports Monte Carlo simulations illustrating the finite-sample properties 
of the estimators. These results reinforce our main findings that properly constructed individualized shrinkage estimators yield valid inference, while improperly constructed ones do not.
\Cref{sec:empirical} revisits the study of \citet{kline2022systemic} on labor market discrimination. \Cref{sec:school} revisits the study of \citet{andrabi2025heterogeneity} on school value-added. Finally, \Cref{sec:conclusion} concludes. An appendix contains the additional lemmas, proofs, and extensions.

\section{Setup} \label{sec:setup_J}

Each unit \(i\) is associated with a latent individual effect $\theta_i$, which we will refer to in what follows as value-added. 
Let  $\bm\theta \coloneqq (\theta_i)_{i=1}^n$ denote the vector of these latent effects for the \(n\) sampled units. A large body of research studies the relationship between value-added $\theta_i$ and some outcome of economic interest $Y_i$. To fix ideas, consider the example of \citet{kane2008estimating}. In that paper, $\theta_i$ is teacher $i$'s latent value-added for students' test scores, and $Y_i$ is the average long-term test score outcome of students taught by teacher $i$. The relationship is studied using the regression model
\begin{align}
  Y_i = \alpha + \beta \theta_i + u_i, \label{eq:downstream}
\end{align}
where the coefficient $\beta$ captures the downstream effect. It is assumed that the error term $u_i$ has mean zero and is uncorrelated with $\theta_i$, $\E\left[ u_i \theta_i\right] = 0$. Research questions about how teacher value-added relates to long term outcomes such as college attendance or earnings can be addressed by performing inference on $\beta$.

Inference on $\beta$ faces the challenge that $\bm \theta$ is unobserved, so the parameter of interest $\beta$ cannot be estimated by regressing $Y_i$ on $\theta_i$. Instead, for each unit $i$ we only have data $\left(Y_i, X_i\right) \coloneqq \left( Y_i, X_{i,1},...,X_{i,J_i}\right)$, $i=1,...,n$, where $Y_i$ is the outcome, $X_{i,1}, \ldots, X_{i, J_i}$ are repeated measurements of $\theta_i$, and $J_i$ is the number of observations available for estimating $\theta_i$. Specifically, $X_i \in \R^{J_i}$ are noisy repeated measurements for $\theta_i$ from the model
\begin{align}
  X_{i,j} = \theta_i + \epsilon_{i,j}, \label{eq:measure}
\end{align}
where $\epsilon_{i,j}$ is the noise term. In the example, $X_{i,j}$ is the test score of student $j$ taught by teacher $i$. The individual mean score for teacher $i$ is denoted by $\bar X_i$, representing the average test score within teacher $i$'s students. The pooled mean score is denoted by $\bar X$, representing the grand average across all teachers and students. In this paper, we primarily focus on the case where the noise is independent of the value-added, i.e., $\epsilon_{i,j} \ind \theta_i$.

\section{Shrinkage Estimators as Regressors} \label{sec:setup} 

In this section, we analyze the inferential properties of different shrinkage estimators when used as downstream regressors. We follow the standard two-step workflow in empirical practice: Step~1: Estimate value-added $\theta_i$ with a shrinkage estimator $\hat \theta_i$; Step~2: Regress $Y_i$ on $\hat\theta_i$ and report conventional OLS estimates, standard errors and confidence intervals (i.e., treating the $\hat \theta_i$ as regular data rather than first-stage estimates). Our interest is in the validity of inference in Step~2 when different shrinkage estimators are used in Step~1. 

The primary implication of our asymptotic framework is that the measurement error in $\bar{X}_i$ does not vanish relative to the sampling error for $\beta$, mimicking the finite-sample problem faced in practice, where both play a role. This persistence of measurement error is the central econometric challenge we address. As we will show in \Cref{sec:FE}, this asymptotic setup confirms that the naive unshrunk estimator (simply regressing $Y_i$ on $\bar{X}_i$) suffers from classical errors-in-variables bias and yields invalid inference. This provides a useful baseline and motivation for the other estimators.

We consider four primary classes of estimators for $\hat{\bm \theta}$, which we analyze in the order of our main results. For each class, we derive the asymptotic distribution of the OLS estimator of $\beta$ in the downstream regression. In \Cref{sec:FE}, we begin with the {unshrunk fixed-effect (FE)} estimator $\bar X_i$ as a baseline. We then move on to {individual-weight shrinkage}. We distinguish between two forms: {homoskedastic individual-weight shrinkage (HO)}, which models heterogeneity in measurement precision {solely by the number of measurements} ($J_i$), and {heteroskedastic individual-weight shrinkage (HE)}, which additionally models idiosyncratic noise variances. We study homoskedastic individual-weight shrinkage in \Cref{sec:HO} and heteroskedastic individual-weight shrinkage in \Cref{sec:HE}. Finally, to connect with other approaches in the literature, we analyze {common-weight shrinkage (CW)}, which applies a single, uniform weight across units.

\subsection{Fixed Effects (FE)}\label{sec:FE}
As a baseline, we first discuss using the fixed-effect estimator (FE) with no shrinkage. In this estimator, value-added for unit $i$ is simply estimated by the individual average: $\hat\theta_{i,\FE} = \bar X_i$, for $i=1,...,n$. 

As we show formally below, using $\bar X_i$ as a regressor suffers from the classical errors-in-variables (EIV) problem, leading to attenuation bias. This bias shifts the location of standard OLS confidence intervals for $\beta$ away from the truth, rendering them invalid for inference. In the following, we first state and discuss assumptions that are needed for the asymptotic properties.

\begin{assumption}
  \label{assum:nonnormal}
  \begin{enumerate}
    \item $J_i$ is independent of $\theta_i$, $u_i$ and $\epsilon_{i,j}$, and $J_i \geq 3, a.s.$ 
    \item $\E \left[u_i\right] = 0$, $\E(u_i  \theta_i) = 0$. $\E\left[Y_i^4\right] < \infty$. 
    \item $\E\left[\epsilon_{i,j} \mid  \sigma^2_i\right] =0$, $\E\left[ \epsilon_{i,j}^2 \mid \sigma_i^2  \right] = \sigma_i^2$, $\E\left[\abs{\epsilon_{i,j}}^L \mid \sigma_i^2\right] \leq K\sigma_i^L$, for some $L\geq3$. Also, $\epsilon_{i,j} \ind \theta_i$.
    \item $ u_i \ind \epsilon_{i,j} \big| \theta_i$. 
    \item $\E\left[\theta_i^4\right] <\infty$.
    \item $\E\left[\sigma^{16}_i \right] < \infty$.
    \item $\sqrt{n} \E\left[ J_i^{-1} \right] \to \kappa$, $\kappa \in [0,+\infty)$. \label{item:J_nonnormal}
    \item $\sqrt{n} \E\left[ J_i^{-2} \right] \to 0$. \label{item:J2_nonnormal}
  \end{enumerate}
\end{assumption}

\Cref{assum:nonnormal}.1 assumes independent numbers of measurements $J_i$. In \Cref{sec:appendix-lemmas-correlated_J}, the condition will be relaxed to allow dependence between $J_i$ and $\sigma_i^2$, accompanied by slightly stronger assumptions than \Cref{assum:nonnormal}.7 and \Cref{assum:nonnormal}.8. We keep the independence for ease of exposition. \Cref{assum:nonnormal}.2 imposes unconditional exogeneity of the true $\theta_i$ and $u_i$, and moment conditions on $Y_i$. 

\Cref{assum:nonnormal}.3 places moment conditions on the noise $\epsilon_{i,j}$. The moment conditions are general and cover a wide range of distributions for $\epsilon_{i,j}/\sigma_i$, including normal, bounded, and even some asymmetric distributions.  Thus, while normality of $\bar X_i$ is assumed to derive the functional form of parametric empirical Bayes estimators, the inference results we derive do not require any normality assumption. Independence $\epsilon_{i,j}\ind \theta_i$ is standard in the empirical Bayes literature to justify shrinkage estimators without covariates. \Cref{assum:nonnormal}.4 excludes further effects of measurement errors on $Y_i$ conditional on $\theta_i$. \Cref{assum:nonnormal}.5 and \Cref{assum:nonnormal}.6 are standard moment conditions for technical arguments. 

Finally, \Cref{assum:nonnormal}.7 and \Cref{assum:nonnormal}.8 specify the asymptotic framework. Note that if we have the same number of measurements, $J_i = J$, then \Cref{assum:nonnormal}.8 is implied by \Cref{assum:nonnormal}.7. Thus \Cref{assum:nonnormal}.7 is the essential condition. \Cref{prop:mean} below holds without \Cref{assum:nonnormal}.8. The asymptotic framework has three motivations. First, it aligns with the standard justification for normality of $\bar X_i$ via the central limit theorem with $J_i$ reasonably large \citep{walters2024empirical}. Second, it reflects the data structure in common applications, such as the teacher value-added setting, where $J_i$ often has a similar magnitude to $\sqrt{n}$. In these settings, the number of teachers is in the thousands and the number of students per teacher is in the tens. For instance, in the study of North Carolina data in \citet{deeb2021framework}, the total number of teachers $n = 5266$, and the total number of students is $388{,}191$, so we would expect $J_i$ to be on average about $74$, which is comparable to $\sqrt{n} \approx 73$. In another example from \citet{bau2020teacher}, $\sqrt{n}\approx39$ and $J_i$ is on average about $15$. Third, it provides useful approximations to the finite-sample problem faced by researchers. The product $\sqrt{n}\E \left[J_i^{-1}\right]$ quantifies the relative scale of measurement error in the estimates of $\theta_i$ to sampling error in $\beta$. When $\kappa = 0$, $J_i$ grows asymptotically much larger than $\sqrt{n}$, implying that measurement error in $\theta_i$ is negligible for the purpose of inference in the regression. When $\kappa > 0$, measurement error in $\theta_i$ is of a similar order to the sampling error. In the previous two application examples, we have $\hat \kappa_1 \approx 73/74\approx1$, and $\hat \kappa_2 \approx 39/15\approx 2.6$,\footnote{These approximations are lower bounds, because Jensen's inequality produces a larger value if $J_i$ varies across $i$} as sample counterparts of $\kappa$. A finite positive $\kappa$ is thus appropriate for the context in which tens of students per teacher and thousands of teachers coexist.

Our asymptotic framework is related to the small-variance approximation in \citet{chesher1991effect} and \citet{evdokimov2019errors,evdokimov2023simple}. This literature does not put a structure on the form of the measurement error, but assumes its variance is proportional to $n^{-1/2}$. Here we know that the measurement error from estimation of $\theta_i$ has a variance proportional to $J_i^{-1}$, which leads to \Cref{assum:nonnormal}.7. It is also standard in settings with both cross-sectional and individual-specific datasets, such as large $N$, $T$ panels \citep{pesaran2006estimation}, factor augmented regressions \citep{gonccalves2014bootstrapping}, and unstructured data \citep{battaglia2024inference}.

Denote by $\hat\beta_{\text{FE}}$ the OLS estimator from the regression of $Y_i$ on $\hat\theta_{i,\FE}$ using observations $i= 1,...,n $. The following result implies that conventional OLS inference on $\beta$ using $\hat \beta_{\FE}$ is invalid.
  \begin{proposition}
    \label{prop:mean}
    Suppose \Cref{assum:nonnormal} holds. Then we have
    \begin{align*}
      \sqrt{n} \left(\hat \beta_\FE - \beta\right) \to_d N \left(-\kappa\beta\frac{\E \left[\sigma_i^2\right]}{\Var\left(\theta_i\right)}, \frac{\E\left[ u_i^2 \left(\theta_i - \E \left[\theta_i\right]\right)^2 \right]}{\left(\Var\left(\theta_i\right)\right)^2}\right).
    \end{align*}
  \end{proposition}

  In the special case where $\kappa=0$ (meaning measurement error is asymptotically negligible), the bias term disappears and the asymptotic distribution of $\sqrt{n} \left(\hat\beta_{\text{FE}} - \beta\right)$ is centered at zero. In our primary framework with $\kappa > 0$, $\hat\beta_{\mathrm{FE}}$ is consistent and asymptotically normal, with the efficient asymptotic variance. Thus, there is no generated regressor problem: OLS standard errors are consistent, following similarly to \Cref{prop:inference} in \Cref{sec:HE}. However, its asymptotic distribution is centered away from zero due to attenuation bias. This has important practical consequences. Consider the conventional confidence interval for $\beta$, formed as $\hat\beta_{\mathrm{FE}} \pm 1.96 \times \mathrm{SE}(\hat\beta_{\mathrm{FE}})$ where $\mathrm{SE}(\hat\beta_{\mathrm{FE}})$ represents the usual OLS standard error. This confidence interval will be centered away from $\beta$, and consequently will not have valid coverage for $\beta$. To give a sense of the magnitude of under-coverage, the simulations in \Cref{sec:sim} show that the coverage of a $95\%$ confidence interval is only about $70\%$.

\subsection{Homoskedastic Individual-Weight Shrinkage (HO)} \label{sec:HO}

In this subsection, we discuss the individual-weight shrinkage when the variance of $\bar X_i$ is assumed to only depend on the number of measurements $J_i$. That is, $\Var\left(\bar X_i \mid \theta_i\right) = \sigma^2/J_i$. Estimators constructed under this assumption are referred to in this paper as homoskedastic individual-weight shrinkage (HO), since the noise $\epsilon_{i,j}$ is homoskedastic across $i$ (note that we still allow for heteroskedasticity in the downstream regression, however). Here the primary source of heteroskedasticity is believed to arise from variation 
in the number of measurements across units $i$. The functional form of the shrinkage weights takes this heterogeneity into account.

Following the estimators proposed by \citet{jacob2007parents} and \citet{kane2008estimating}, the homoskedastic individual-weight shrinkage estimator $\bm{\theta}$ is constructed as
\begin{align*}
  \hat \theta_{i, \text{HO}} \coloneqq \frac{\hat\sigma_\theta^2}{\frac{1}{J_i}\hat\sigma^2 + \hat\sigma_\theta^2}\bar X_i +  \frac{\frac{1}{J_i}\hat\sigma^2}{\frac{1}{J_i}\hat\sigma^2 + \hat\sigma^2_\theta}\bar X.
\end{align*}
The estimator is the empirical Bayes posterior mean under normality of 
$\theta_i$ and $\epsilon_{i,j}$. 
The weight for $\bar X_i$ is $w_{i,\HO}$, obtained by replacing the oracle shrinkage weight  
\(
\frac{\Var(\theta_i)}{\sigma^2/J_i + \Var(\theta_i)}
\)
with its empirical counterparts. These variances are estimated as
\begin{align*}
  &\hat\sigma^2 \coloneqq \widehat{\Var} \left( X_{i,j} - \bar X_i   \right), \\
  &\hat\sigma_\theta^2 \coloneqq \widehat{\Cov}\left(\bar X_{i,t} ,\bar X_{i,t-1}\right),
\end{align*}
where $\bar X_{i,t}$, $\bar X_{i,t-1}$ are subset averages from splitting the measurements for each $i$ into two subsets. In a context with time periods, they can also be averages for two time periods. Note that the same $\hat\sigma^2$ is applied to units, with variation in $J_i$ reflecting differences in measurement precision.

We now derive the asymptotic distribution of the OLS estimator $\hat \beta_{\HO}$ of regressing $Y_i$ on $\hat\theta_{i,\HO}$ in two cases. In the first case, we assume the homoskedasticity in noise holds (i.e., $\sigma_i^2 = \sigma^2$ for all $i$). In the second case, we allow for heteroskedasticity in noise (i.e., $\sigma_i^2$ varies across $i$).

If homoskedasticity holds, this method is expected to correctly model the signal-to-noise ratio. Indeed, as we show below in \Cref{prop:HO_in_text}, regressing $Y_i$ on $\hat \theta_{i, \text{HO}}$ yields asymptotic unbiasedness and valid inference on $\beta$ under homoskedasticity.

If the true DGP is heteroskedastic, we demonstrate problems in inference when homoskedastic weights are applied. Here, denote the variance of $\bar X_i \mid \theta_i$ by $\sigma_i^2/J_i$. Regressing $Y_i$ on $\hat \theta_{i, \text{HO}}$ can lead to invalid inference on $\beta$ under noise heteroskedasticity. A leading case is when the number of measurements $J_i$ is correlated with the variance $\sigma_i^2$. In that case, it is natural to have $J_i$ endogenously increased for observations with large $\sigma_i^2$. But in this case the downstream regression estimator has an asymptotic bias which invalidates the standard approach to inference. The following result substantiates this claim, with details and proofs in \Cref{sec:appendix-lemmas-correlated_J}.
\begin{proposition}\label{prop:HO_in_text}
  Suppose \Cref{assum:corr_J} and \ref{assum:corr_J_large_J} hold, which allow for correlation between $J_i$ and $\sigma_i^2$. Then there exist DGPs with $\gamma >0$ in which
  \begin{align*}
    \sqrt{n} \left( \hat\beta_{\HO} - \beta \right) \to_d N\left(\gamma \beta,\frac{\E\left[u_i^2 \left(\theta_i - \E\left[\theta_i\right]\right)^2\right]}{\left(\Var\left(\theta_i\right)\right)^2}\right).
  \end{align*}
  If instead $J_i$ and $\sigma_i^2$ are independent, then
  \begin{align*}
    \sqrt{n} \left( \hat\beta_{\HO} - \beta \right) \to_d N\left(0,\frac{\E\left[u_i^2 \left(\theta_i - \E\left[\theta_i\right]\right)^2\right]}{\left(\Var\left(\theta_i\right)\right)^2}\right).
  \end{align*}
\end{proposition}
\Cref{prop:HO_in_text} clarifies the conditions required for regressing $Y_i$ on $\hat\theta_{i,\HO}$ to yield valid inference. While inference is valid if the homoskedasticity assumption holds or if $J_i$ and $\sigma_i^2$ are independent, the estimator's performance is sensitive to violations of homoskedasticity. As $\gamma >0$ in the bias term, the asymptotic bias is \emph{amplification} instead of classical EIV attenuation bias. Generally, the sign and magnitude of $\gamma$ depend on the dependence of $J_i$ and $\sigma_i^2$, and the bias can thus be positive or negative. It motivates the heteroskedastic estimator we consider next, which is designed to be robust to such dependence.

\subsection{Heteroskedastic Individual-Weight Shrinkage (HE)} \label{sec:HE}
In this subsection, we discuss the {heteroskedastic individual-weight (HE) estimator} under noise heteroskedasticity. This estimator is designed to accommodate heterogeneity from both the number of measurements $J_i$ and the idiosyncratic noise variance $\sigma_i^2$. It improves upon $\hat\theta_{i,\HO}$.

The heteroskedastic individual-weight shrinkage estimator is defined as 
\begin{align}
  &\hat \theta_{i,\HE} \coloneqq w_{i,\HE} \bar X_i + (1-w_{i,\HE})\bar X, \; \text{where}  \\
  &w_{i,\HE} \coloneqq \frac{\hat V}{\frac{1}{J_i}\hat\sigma_i^2 + \hat V}. \nonumber
\end{align}
Here, $\hat V$ is the estimator of $\Var\left(\theta_i\right)$, and $\hat \sigma_i^2$ is the estimator of $\Var\left(X_{i,j} \mid \theta_i\right)$. The variance estimators are constructed following \citet{kline2020leave} and \citet{kline2022systemic}:
\begin{align*}
  &\hat \sigma^2_i \coloneqq \frac{1}{J_i-1} \sum_{j=1}^{J_i} \left( X_{i,j} - \bar X_i \right)^2,\\
  &\hat V \coloneqq \frac{1}{n}\sum_{k=1}^{n}\left( \bar X_k - \bar{X} \right)^2 - \frac{n-1}{n^2}\sum_{k=1}^{n} \frac{1}{J_k}\hat \sigma^2_k.
\end{align*}
Unlike $\hat\theta_{i,\HO}$, the weights $w_{i,\mathrm{HE}}$ allow $\Var(X_{i,j}\mid \theta_i)$ to differ across units, thereby ensuring robustness to heteroskedasticity.

Let $\hat\beta_{\HE}$ denote the downstream OLS estimator from regressing $Y_i$ on $\hat\theta_{i,\HE}$. We take an intermediate step in deriving the properties of $\hat\beta_\HE$. The estimator $\hat\theta_{i,\HE}$ is an empirical Bayes estimator, with the estimated prior variance $\hat{V}$ and prior mean $\bar X$ for a normal prior on $\theta_i$. We start by studying the property of OLS estimator $\hat\beta_{c,\HE}$ from regressing $Y_i$ on the semi-oracle estimator $\hat\theta_{i,c,\HE}$ (with prior mean still estimated from the data),
  \begin{align*}
    &\hat{\theta}_{i,c,\HE} \coloneqq c_i \bar{X}_i + \left(1- c_i \right) \bar{X}, \;\text{where}\\
    &c_i \coloneqq \frac{\Var\left(\theta_i\right)}{\frac{1}{J_i}\hat \sigma^2_i + \Var\left(\theta_i\right)}.
  \end{align*}
\Cref{lemma:V_nonnormal} shows that the estimated variance $\hat V$ is a consistent estimator for the true variance $\Var\left(\theta_i\right)$. Thus, we build the properties of $\hat\beta_\HE$ based on those of $\hat\beta_{c,\HE}$. The following result shows that the OLS estimator $\hat \beta_{c,\HE}$ from regression $Y_i$ on $\hat\theta_{i,c,\HE}$ is asymptotically normal, correctly centered, with the efficient asymptotic variance from the infeasible regression of $Y_i$ on $\theta_i$.

\begin{lemma}
  \label{lemma:beta_c}
  Suppose \Cref{assum:nonnormal} holds. Then we have
  \begin{align*}
    \sqrt{n} \left(\hat{\beta}_{c,\HE} - \beta\right) \to_d N \left(0, \frac{\E\left[ u_i^2 \left(\theta_i - \E \left[\theta_i\right]\right)^2 \right]}{\left(\Var\left(\theta_i\right)\right)^2}\right).
  \end{align*}
\end{lemma}

\Cref{lemma:beta_c} shows that regressing $Y_i$ on $\hat\theta_{i,c,\HE}$ leads to an asymptotically unbiased estimator of $\beta$ in infeasible cases where $\Var\left(\theta_i\right)$ is known. The estimator performs asymptotically equivalently to that from regression of $Y_i$ on $\theta_i$. Next, we move on to the feasible estimator $\hat \theta_{i,\HE}$. 

The following result builds on \Cref{lemma:beta_c} and establishes the asymptotic distribution if we use the feasible shrinkage estimator $\hat\theta_{i,\HE}$. Just like $\hat\beta_{c,\HE}$, the OLS estimator $\hat\beta_\HE$ from regressing $Y_i$ on $\hat\theta_{i,\HE}$ is also asymptotically normal, correctly centered, with the efficient asymptotic variance.
\begin{theorem}
  \label{prop:beta}
  Suppose \Cref{assum:nonnormal} holds. Then we have that $\hat \beta_\HE$ and $\hat \beta_{c,\HE}$ are first-order asymptotically equivalent:
  \begin{align*}
    \sqrt{n}\left(\hat\beta_\HE - \beta\right) = \sqrt{n}\left(\hat\beta_{c,\HE} - \beta\right) + o_p(1),
  \end{align*}
  and so
\begin{align*}
  \sqrt{n} \left(\hat{\beta}_\HE - \beta\right) \to_d N \left(0, \frac{\E\left[ u_i^2 \left(\theta_i - \E \left[\theta_i\right]\right)^2 \right]}{\left(\Var\left(\theta_i\right)\right)^2}\right).
\end{align*}
\end{theorem}

\Cref{prop:beta} establishes that the OLS estimator based on the feasible regressor $\hat \theta_{i,\HE}$ is first-order asymptotically equivalent to the estimator based on the semi-oracle $\hat\theta_{i,c,\HE}$. Moreover, there is no efficiency loss relative to the infeasible regression of $Y_i$ on the latent $\theta_i$. The asymptotic variance is the efficiency bound for estimating $\beta$ in models where $\theta_i$ is observed \citep[See the unconditional moment case of][]{chamberlain1987asymptotic}. The asymptotic distribution doesn't depend on the value of $\kappa$. Thus, the OLS regression estimator $\hat \beta_{\HE}$ is robust to the amount of measurement error (under the conditions of \Cref{prop:beta}).   

\begin{remark}
  Recall from \Cref{prop:mean} that the asymptotic distribution of the estimator $\hat \beta_{\FE}$ from regressing on unshrunk $\hat\theta_{i,\FE}$ was biased by the $\kappa$ term, leading to invalid inference. Here, even though we remain in the same $\kappa > 0$ regime where measurement error is persistent, that bias term is absent. The HE shrinkage procedure perfectly corrects the EIV problem, thereby resolving the inferential issue in the baseline FE case.
\end{remark}

\begin{remark}
  In \Cref{sec:HO}, we showed that under heteroskedasticity, HO yields invalid inference for $\beta$. By contrast, as we show in \Cref{prop:beta_corr_J}, HE delivers an asymptotically unbiased estimator of $\beta$ under the same conditions. In this sense, HE provides a more robust approach. Details and proofs are given in \Cref{sec:appendix-lemmas-correlated_J}.
\end{remark}

We further establish consistency of standard errors from the regression of $Y_i$ on $\hat \theta_{i,\HE}$. Together with \Cref{prop:beta}, this ensures that regression of $Y_i$ on  $\hat \theta_{i,\HE}$ delivers automatic inference in practice: one can proceed as if $\hat \theta_{i,\HE}$ are the true latent $\theta_i$ for the purpose of estimation and inference on $\beta$.

Building on \Cref{assum:nonnormal}, additional regularity conditions are required for consistency of the variance estimator.
\begin{assumption}
  \label{assum:inference}
  \begin{enumerate}
    \item $\E\left(u_i \right)=0$, $\E(u_i  \theta_i) = 0$. $\E \left[\abs{ \theta_i^{k_1} u_i^{k_2}} \right] < \infty$ for any $1\leq k_1 \leq 8$, $1 \leq k_2 \leq 4$.
    \item $\E\left[\theta_i^8\right] <\infty$.
  \end{enumerate}
\end{assumption}
Here, \Cref{assum:inference} strengthens the moment conditions to ensure that the law of large numbers applies to higher order terms in the variance estimator. 

\begin{theorem}
\label{prop:inference}
Suppose \Cref{assum:nonnormal} and \Cref{assum:inference} hold. Then the Eicker--Huber--White estimator of standard errors in the regression of $Y_i$ on $\hat \theta_{i,\HE}$ is consistent:
\begin{align*}
  \hat \Omega \coloneqq \frac{\frac{1}{n}\sum_{i=1}^{n} \left(\hat\theta_{i,\HE} - \bar {\hat \theta}_\HE\right)^2 \hat u_i^2 }{\left( \frac{1}{n} \sum_{i=1}^{n} \left(\hat \theta_{i,\HE} -\bar{\hat\theta}_\HE \right)^2  \right)^2} \to_p \frac{\E\left[u_i^2\left(\theta_i - \E \left[\theta_i\right]\right)^2\right]}{\left(\Var\left(\theta_i\right)\right)^2},
\end{align*}
where $\hat u_i$ is the OLS residual:
\begin{align*}
  \hat u_i \coloneqq Y_i - \hat \alpha - \hat \beta  \hat\theta_{i,\HE}.
\end{align*}
\end{theorem}

Taken together, \Cref{prop:beta} and \Cref{prop:inference} show that asymptotically valid confidence intervals for $\beta$ can be constructed as follows. If we regress $Y_i$ on $\hat\theta_{i,\HE}$, the reported standard error is $\sqrt{\hat{\Omega}/n}$. The asymptotically valid confidence intervals at level $1-\alpha$ are given by
  \begin{align*}
    \hat\beta_{\HE} \pm z_{1-\alpha/2} \sqrt{\hat{\Omega}/n},
  \end{align*}
  where \( z_{1-\alpha/2} \) denotes the \( 1-\alpha/2 \) quantile of the standard normal distribution. \Cref{prop:beta} and \Cref{prop:inference} ensure that the confidence interval has asymptotically nominal coverage. Specifically:
  \begin{align*}
    \Pr \left( \beta \in [ \hat\beta_\HE - z_{1-\alpha/2} \sqrt{\hat{\Omega}/n}, \hat\beta_\HE + z_{1-\alpha/2}\sqrt{\hat{\Omega}/n} ]  \right) \to 1-\alpha.
  \end{align*}
In practice, efficient inference is easy to implement. One regresses $Y_i$ on $\hat{\theta}_{i,\mathrm{HE}}$ to obtain the OLS estimator $\hat{\beta}_{\mathrm{HE}}$ and its standard error. The resulting confidence interval for $\beta$ is asymptotically valid.

The HE estimator also serves another purpose. While the focus of this paper is on downstream inference, it is worth noting that the HE estimator has favorable properties for estimating the full vector of individual effects $\bm\theta$. As an estimator motivated by empirical Bayes principles, its construction improves estimation accuracy \citep[see, e.g.,][]{efron1973stein}. By correctly modeling all sources of heterogeneity, $\hat{\bm\theta}_\HE \coloneqq \left(\hat\theta_{i,\HE}\right)_{i=1}^n$ is known to achieve a lower mean squared error (MSE) for $\bm\theta$ than the unshrunk FE, the HO, and the common-weight estimators, a standard result in that literature.

\subsection{Common-Weight Shrinkage (CW)} \label{sec:CW}
As a final point of comparison, we analyze the {common-weight (CW) shrinkage estimator}, which serves as a benchmark. This estimator, motivated by the James-Stein estimator, shrinks each unit mean $\bar X_i$ towards the grand mean $\bar X$ by a {common} factor $w$:
\begin{align*}
  \hat \theta_{i, \text{CW}} \coloneqq w \bar X_i + (1-w)\bar X,
\end{align*}
where $\bar X$ denotes the grand mean of the sample. The weight $w$ can be data-dependent as in the James-Stein estimator. For instance, common-weight shrinkage is applied in \citet{chetty2014measuring,chetty2014measuring2}, where the weight $w$ is chosen from the best linear predictor of $\theta_i$ given $\bar X_i$. Inference of $\beta$ is then performed by regressing $Y_i$ on $\hat \theta_{i, \text{CW}}$.

We now give two examples of common weight $w$ for shrinkage.

\paragraph{(i) Bias Correction Shrinkage}

Regressing $Y_i$ on $\hat \theta_{i, \text{CW}}$ leads to an OLS estimator of $\beta$ given by
\begin{align*}
  \hat\beta_{\text{CW}} \coloneqq \frac{ \widehat{\Cov} \left(Y_i, \hat\theta_{i,\text{CW}}\right) }{\widehat{\Var}\left(\hat\theta_{i,\text{CW}}\right)} = w^{-1} \underbrace{\frac{ \widehat{\Cov} \left(Y_i, \bar X_i\right) }{\widehat{\Var}\left(\bar X_i\right)}}_{\hat\beta_{\text{FE}}}.
\end{align*}
In effect, regressing $Y_i$ on $\hat \theta_{i,\text{CW}}$ adjusts $\hat\beta_{\text{FE}}$ of regressing $Y_i$ on $\bar X_i$ by a factor of $w^{-1}$. By a proper choice of shrinkage weight $w \approx \widehat \Var\left(\theta_i\right)/\widehat \Var \left(\bar X_i\right)$, the adjustment overlaps with the bias correction method for the EIV problem. Heuristically, in that case
\begin{align*}
  \hat\beta_{\text{CW}} = \frac{\widehat \Var\left(\bar X_i\right)}{\widehat\Var\left(\theta_i\right)} \cdot \frac{\widehat{\Cov}\left(Y_i, \bar X_i\right)}{\widehat{\Var}\left(\bar X_i\right)} = \frac{\widehat{\Cov}\left(Y_i, \bar X_i\right)}{\widehat{\Var}\left(\theta_i\right)} \approx \frac{\widehat{\Cov}\left(Y_i, \theta_i\right)}{\widehat{\Var}\left(\theta_i\right)}.
\end{align*}
The bias correction shrinkage method with the above $w$ thus has the equivalent property as classical bias correction in EIV problems for downstream regressions.

\paragraph{(ii) IV Shrinkage}

Another choice of $w$ would also result in a shrinkage option that, when used as the regressor, refines the IV estimator for $\beta$. To see this, if we split the measurements for each unit $i$ into two subsets, and treat the subset averages $\bar X_{i,1}$, $\bar X_{i,2}$ as the instrument and endogenous variable respectively, and then we construct an IV estimator:
\begin{align*}
  \hat\beta_{\text{IV}} =\left( \frac{\widehat{\Cov}\left(\bar X_{i,1}, \bar X_{i,2}\right)}{\widehat{\Var}\left(\bar X_{i,1}\right)}  \right)^{-1}  \frac{\widehat{\Cov}\left(Y_i, \bar X_{i,1}\right)}{\widehat{\Var}\left(\bar X_{i,1}\right)}  = \frac{\widehat{\Cov}\left(Y_i, \bar X_{i,1}\right)}{\widehat{\Cov}\left(\bar X_{i,1}, \bar X_{i,2}\right)} \approx \frac{\widehat{\Cov}\left(Y_i, \theta_i\right)}{\widehat{\Var}\left(\theta_i\right)}.
\end{align*}
The IV estimator adjusts the estimator from regressing $Y_i$ on $\bar X_{i,1}$ by a factor. With a proper modification of the factor, it can also adjust the estimator from regressing $Y_i$ on $\bar X_{i}$, which exploits more data.

Because of its connection to them, common-weight shrinkage can perform equivalently to the bias correction method and the IV estimator. It also requires as weak assumptions as theirs. Even though it ignores the heterogeneous signal-to-noise ratio and pools every individual at a uniform ratio, common-weight shrinkage provides valid inference on $\beta$. The following result formalizes this point for the bias correction shrinkage, with the variance estimator for $\theta_i$ the same as in $\hat\beta_\HE$. It is proved in \Cref{sec:appendix-lemmas-nonnormal}.

\begin{proposition}
  \label{prop:CW}
  Suppose \Cref{assum:nonnormal} holds. Then we have
  \begin{align*}
    \sqrt{n} \left(\hat \beta_\CW - \beta\right) \to_d N \left(0, \frac{\E\left[ u_i^2 \left(\theta_i - \E \left[\theta_i\right]\right)^2 \right]}{\left(\Var\left(\theta_i\right)\right)^2}\right).
  \end{align*}
\end{proposition}

Therefore, the CW estimator also addresses the inferential problem from the FE baseline. Its primary distinction from the HE estimator is {how} it addresses the problem: the CW approach applies a uniform correction factor (equivalent to bias-correction or IV), while the HE estimator models unit-level heterogeneity. Our focus on individual-weight shrinkage is motivated by its prevalence in empirical applications. The common-weight approach characterizes influential work such as \citet{chetty2014measuring2}. By contrast, a distinct and extensive literature has focused on individualized shrinkage \citep[e.g.,][]{jacob2007parents,kane2008estimating,andrabi2025heterogeneity}. Our analysis thus provides the necessary theoretical foundation for this widely used empirical strategy.

\section{Simulations} \label{sec:sim}

\subsection{Simulation Design}
In the first set of simulations, we focus on the comparison between regressing on the shrinkage estimator $\hat\theta_{i,\HE}$, the sample mean $\bar X_i$, i.e. $\hat\theta_{i,\FE}$, and the true latent $\theta_i$. In each of the $S = 3000$ simulations, we generate the data, run the regression of $Y_i$ on shrunk or unshrunk regressors, and then report the $95\%$ confidence intervals of $\beta$. Finally, for each value of $\beta$ in the grid, we compute the coverage rate across all simulations, i.e., the proportion of simulations in which the value of $\beta$ falls within the $95\%$ confidence interval.

The number of measurements is fixed at $J_i = J = 20$ and the sample size is $n = 1000$. Here the ratio $\hat \kappa \approx 1.58$. We set the true $\theta_i$ to be drawn from the standard normal distribution, and the variance $\sigma_i^2$ drawn from $\chi^2(1)$. In the regression, we set $\alpha = 0$, $\beta = 1$, and draw $u_i$ from a standard normal---thus homoskedastic---distribution.

 In the first setting, we generate the measurement error $\epsilon_{i,j}$ from a normal distribution, with all conditions in \Cref{assum:nonnormal} satisfied. In the second setting, we set the measurement error $\epsilon_{i,j}$ still having the variance $\sigma_i^2$, but following a centered linear transform $\sigma_i\frac{\chi^2(2)-2}{2}$. Note that the ratio $\epsilon_{i,j}/\sigma_i$ satisfies the moment conditions in \Cref{assum:nonnormal}. As $\epsilon_{i,j}$ follows a shifted Gamma distribution, which is non-normal and asymmetric, the purpose is to show that the proposed methods perform well under non-normal measurement error (provided the regularity conditions are satisfied).

In the second set of simulations, we compare bias and coverage rates across different methods, especially on the comparison between the individual-weight method HE and CW. We report the square root scaled MSE of $\hat\beta$, coverage rate, bias, and MSE of $\hat{\bm \theta}$. The design we consider lets the heteroskedastic $\sigma_i^2$ be from a uniform distribution on $\left\{1,10\right\}$, with other configurations the same as the previous case. 

In the third set of simulations,  we again compute the coverage rate for each method. We check the performance with random $J_i$ drawn from $\text{Poisson}(20)$, and with the sample size $n = 1000$. Here $\hat \kappa$ is approximately $1.66$, larger than $1.58$ in the first set of simulations due to convexity. With other parameters unchanged, now we can introduce correlation between $J_i$ and $\sigma_i^2$. In the first setting, we generate $J_i$ and $\sigma_i^2$ independently, ensuring that all conditions in \Cref{assum:nonnormal} are satisfied. In the second setting, we introduce a positive correlation between $J_i$ and $\sigma_i^2$, reflecting possible endogenous choice of more measurements for lower precision, keeping the conditions in \Cref{assum:corr_J} satisfied. 

\subsection{Results}
The first set of coverage results is presented in \Cref{fig:normal_chi}. Coverage for regressing $Y_i$ on $\hat \theta_{i, \HE}$ (the blue line) performs well in both normal and non-normal settings. The coverage rates are close to $95\%$ at the true $\beta$, and the curve is always very close to the infeasible case (the dashed red line) of regressing on the true $\theta_i$. In both cases, regressing on the sample mean $\bar X_i$ suffers from attenuation bias, reflected by a shift to the left in the coverage curve (the orange line), though the spread is approximately correct, which confirms the theory developed in \Cref{prop:mean}.

\begin{figure}[t]
  \centering
  \subfigure[Normal Noise]{
      \includegraphics[width=0.45\textwidth]{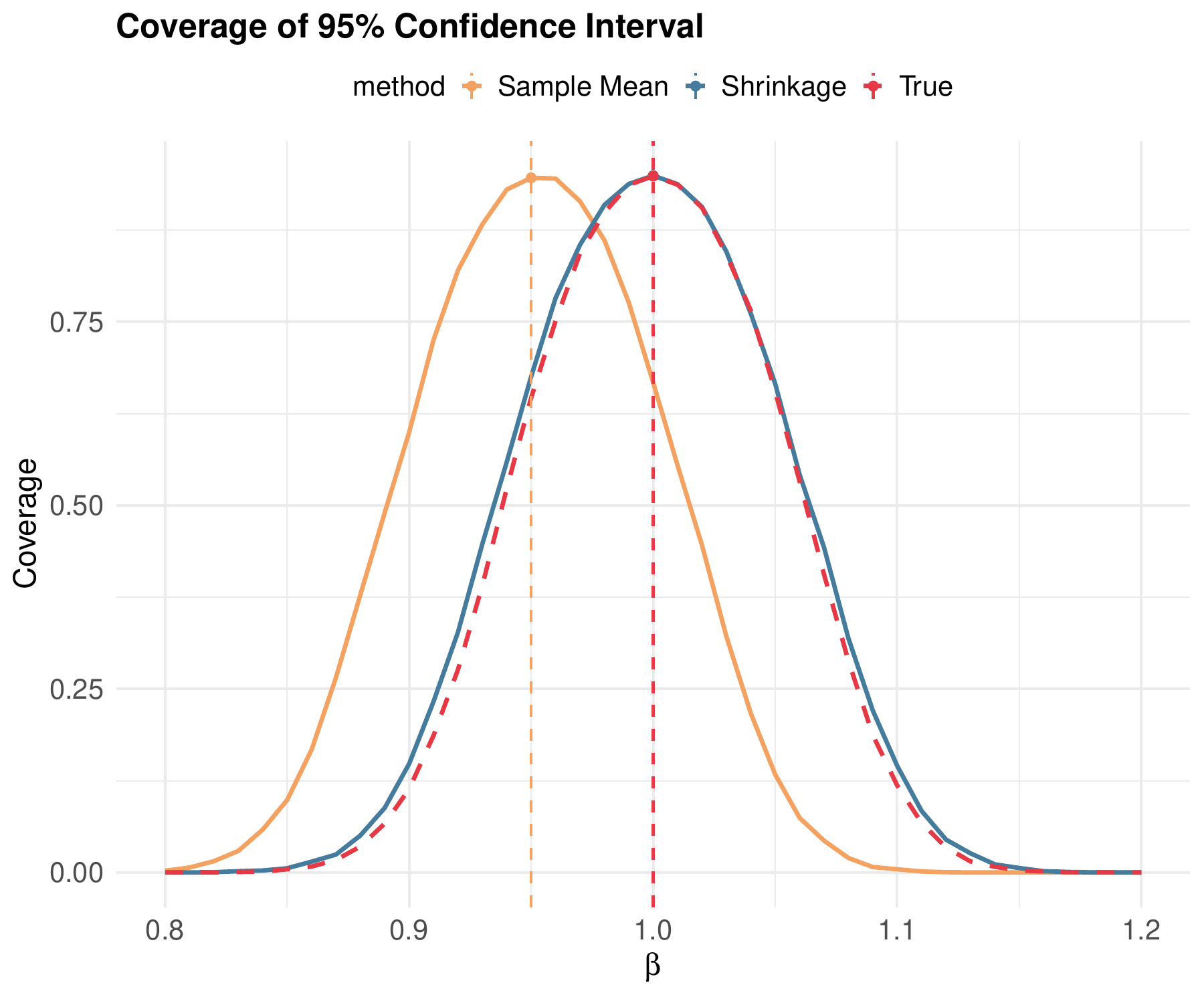}
  }
  \hfill
  \subfigure[Gamma Noise]{
      \includegraphics[width=0.45\textwidth]{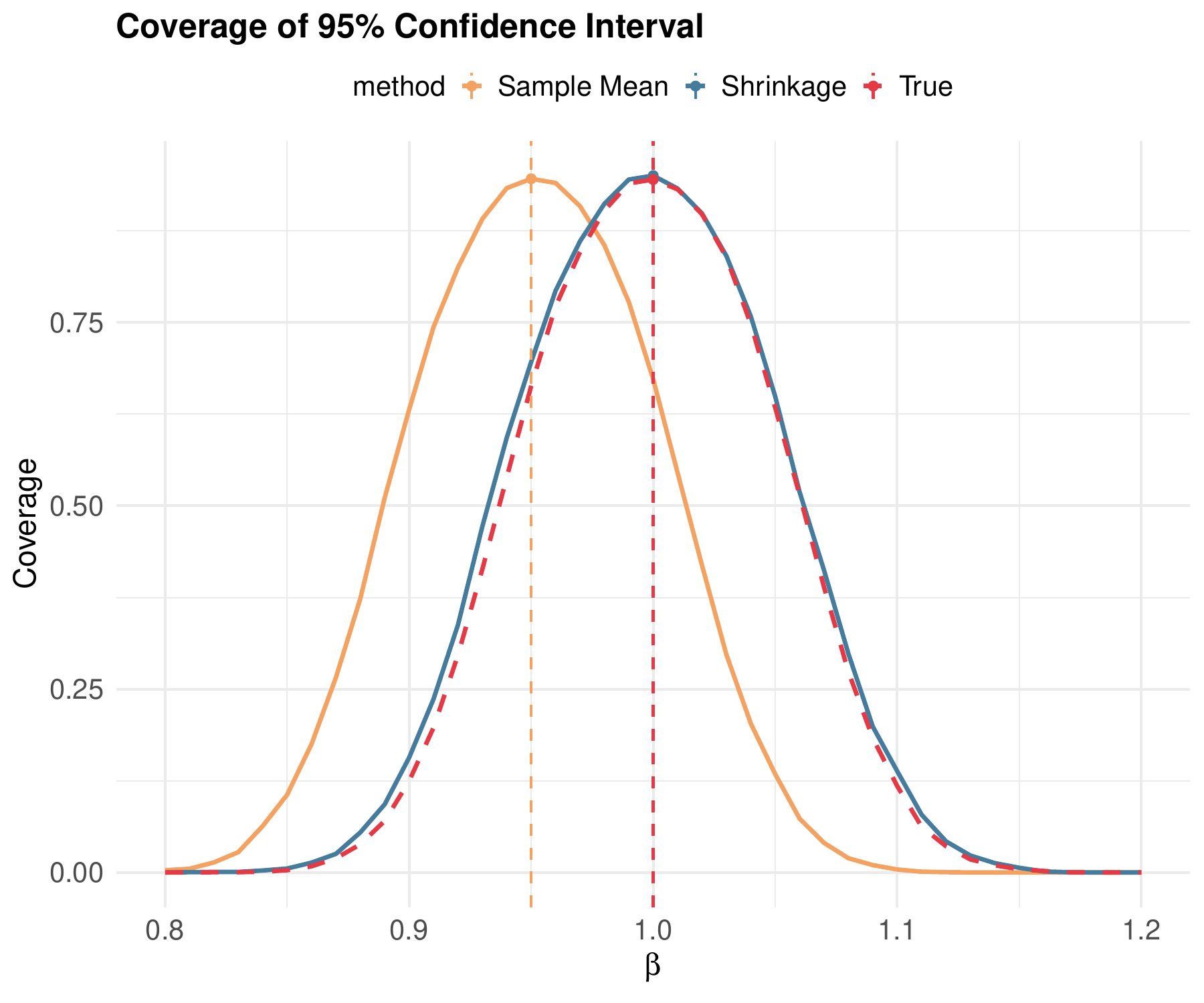}
  }
  \caption*{\parbox{0.9\textwidth}{\textbf{Note}: Each curve represents the proportion of simulations in which the value of $\beta$ on the x-axis falls within the 95\% confidence interval. The dashed red line represents the infeasible case of regressing on the true $\theta_i$. The blue line represents regressing on the shrinkage estimator $\hat\theta_{i,\HE}$ (HE). The orange line represents regressing on the sample mean $\bar X_i$ (FE).}}

  \caption{Coverage Rates Under Different Noise Distributions}
  \label{fig:normal_chi}
  
\end{figure}

The second set of results is reported in \Cref{tab:HE_CW}. The results indicate that regressing $Y_i$ on $\hat\theta_{i,\HE}$ has a smaller MSE of $\beta$, higher coverage rate, and smaller bias. Also, $\hat\theta_{i,\HE}$ has smaller MSE of $\bm\theta$ than $\hat\theta_{i,\CW}$ when the ratio $\sqrt{n}\E\left[J_i^{-1}\right]$ is reasonable. When $n$ and $J$ grow large, the performance of $\hat\theta_{i,\CW}$ and $\hat\theta_{i,\HE}$ becomes close, but $\hat\theta_{i,\HE}$ always dominates and is even more preferable when the sample size is relatively small. As discussed in \Cref{sec:CW}, we know that $\hat\beta_{\CW}$ is first-order asymptotically equivalent to $\hat\beta_\HE$. Combined with the simulation results, we can see that $\hat\beta_\HE$ can achieve better performance in finite samples.

\begin{table}[t] 
  \centering
  \caption{Method Comparison Across Scenarios}
  \label{tab:HE_CW}
  \resizebox{\textwidth}{!}{
  \begin{tabular}{llcccc} 
  \toprule
  \textbf{Scenario} & \textbf{Method} & $\sqrt{n \times \textbf{MSE($\beta$)}}$ & \textbf{Coverage Rate (\%)} & \textbf{Bias} &\textbf{MSE($\bm\theta$)} \\
  \midrule
  n=50, J=10 & True $\theta_i$ & 1.030 & 94.73 & 0.116 & 0\\
  $\sqrt{n}/J\approx$ 0.71          & HE & 1.497 & 91.90 & 0.166 & 0.322 \\
             & CW & 2.335 & 87.83 & 0.225 & 0.374\\
             & FE & 2.690 & 27.67 & 0.354 & 0.552\\
  \midrule
  n=225, J=10 & True $\theta_i$ & 1.016 & 94.17 & 0.054 & 0\\
  $\sqrt{n}/J\approx$ 1.5     & HE & 1.550 & 90.27 & 0.083 & 0.313\\
             & CW & 1.856 & 86.13 & 0.099 & 0.359\\
             & FE & 5.411 & 0.03 & 0.354 & 0.550\\
  \midrule
  n=1000, J=20 & True $\theta_i$ & 1.009 & 95.13 & 0.025 & 0\\
  $\sqrt{n}/J\approx$ 1.58           & HE & 1.298 & 93.13 & 0.033 & 0.195\\
             & CW & 1.392 & 91.33 & 0.035 & 0.216\\
             & FE & 6.930 & 0.00 & 0.217 & 0.275 \\
  \bottomrule
  \end{tabular}
  }
  \vspace{0.5cm}
  \caption*{\parbox{0.9\textwidth}{\textbf{Note}: This table compares the performance of different estimation methods across various scenarios, focusing on mean squared error (MSE), coverage rate, and bias. "True $\theta_i$" represents the benchmark case where the true values of $\theta_i$ are known. HE denotes the heteroskedastic estimator, while CW refers to the common-weight estimator. FE represents the fixed-effects estimator. The square root of scaled MSE of $\beta$ is computed as $\sqrt{n \times \text{MSE}(\beta)}$. Coverage rates are reported as percentages, and bias refers to the absolute deviation of the estimated $\beta$ from its true value. The MSE of $\bm\theta$ measures the estimation error for individual effects. A larger $\sqrt{n}/J$ indicates a higher ratio of measurement error to sampling error.}}
\end{table}

The coverage results for random $J_i$ are presented in \Cref{fig:J_sigma}. We mainly focus on the curves for $\hat\beta_\HE$ (the blue line) and $\hat\beta_\HO$ (the purple line). We can see that the $\hat\beta_\HE$ delivers valid coverage in both settings, close to the infeasible case (the dashed red line). Instead, $\hat\beta_\HO$ only delivers valid coverage in the first setting. When $J_i$ and $\sigma_i^2$ are correlated, the asymptotic distribution of $\hat\beta_{\text{HO}}$ is centered away from the true beta, reflected by a shift to the right in the coverage plot for $\hat\beta_{\text{HO}}$. Similar to before, regressing on the sample mean $\bar X_i$ results in a leftward biased coverage curve. As discussed in \Cref{sec:HO}, the asymptotic bias in $\hat\beta_\HO$ is distinct from classical EIV attenuation bias, as it shifts the estimate in the opposite direction and causes amplification.

\begin{figure}[t]
  \centering
  \subfigure[$J_i \ind \sigma_i^2$]{
      \includegraphics[width=0.45\textwidth]{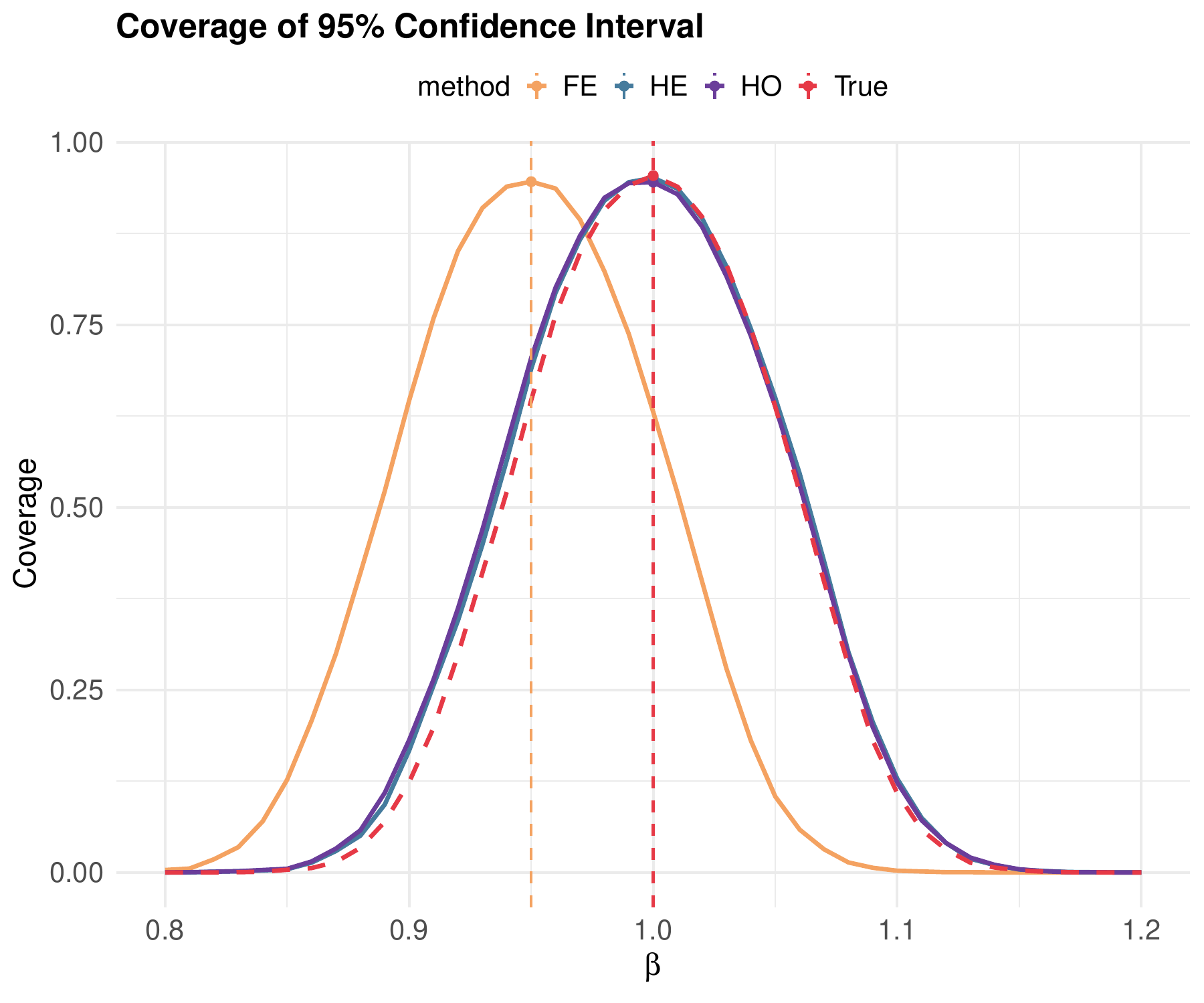}
  }
  \hfill
  \subfigure[$J_i$ correlated with $\sigma_i^2$]{
      \includegraphics[width=0.45\textwidth]{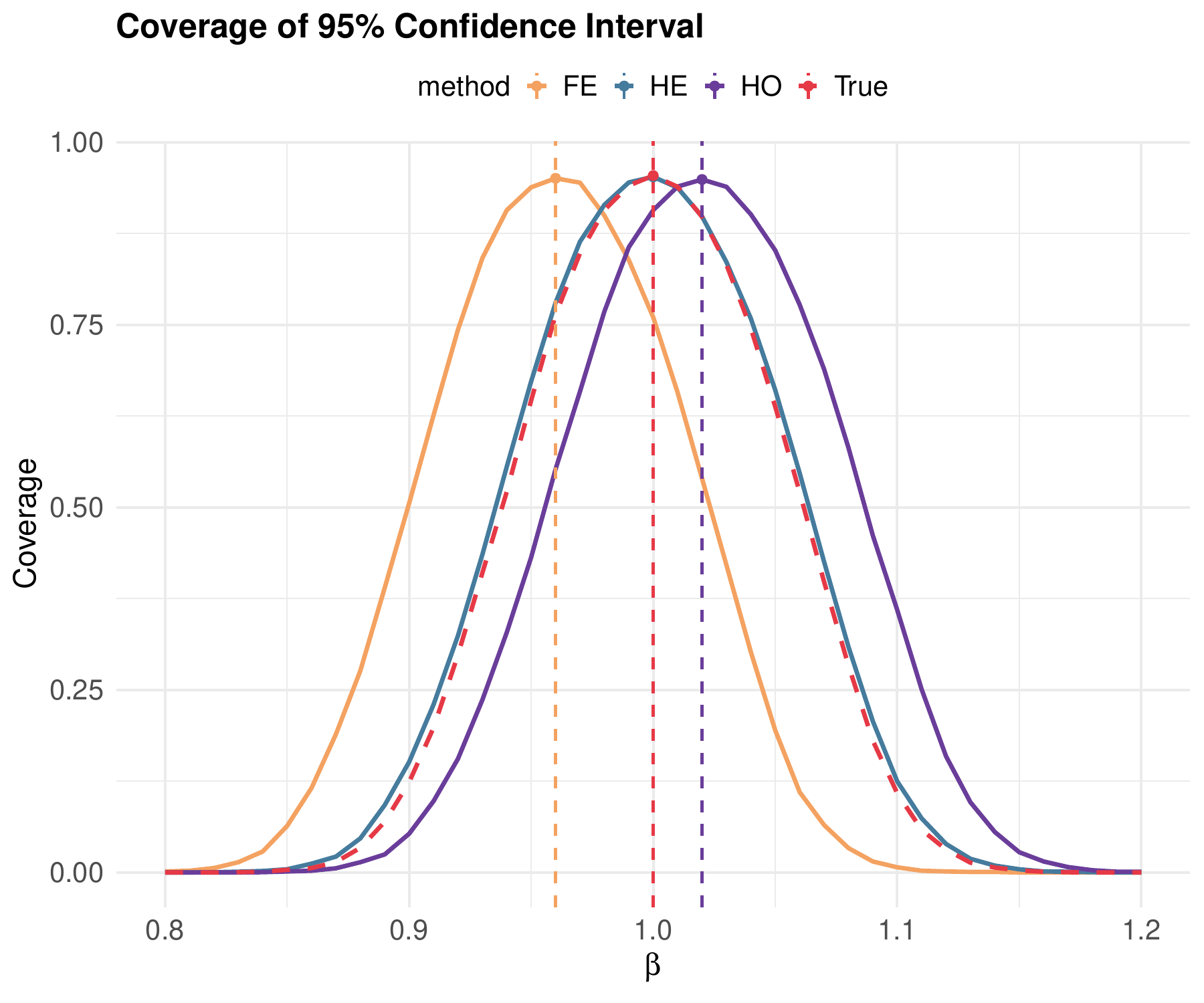}
  }
  \caption*{\parbox{0.9\textwidth}{\textbf{Note}: Each curve represents the proportion of simulations in which the value of $\beta$ on the x-axis falls within the 95\% confidence interval. The dashed red line represents the infeasible case of regressing on the true $\theta_i$. The blue line represents the method of regressing on the heteroskedastic individual-weight shrinkage estimator $\hat\theta_{i,\HE}$ (HE). The purple line represents the regressing on the homoskedastic individual-weight shrinkage estimator $\hat\theta_{i,\text{HO}}$ (HO). The orange line represents regressing on the sample mean $\bar X_i$ (FE).}}
  \caption{Coverage Rates Under Different Dependence}
  \label{fig:J_sigma}
\end{figure}

Overall, the simulation results confirm the theoretical findings in \Cref{sec:setup} and support \Cref{assum:nonnormal} regarding the nonnormality of $\epsilon_{i,j}$. Additionally, they validate the theoretical framework in \Cref{sec:setup} concerning the dependence of $J_i$ and $\sigma_i^2$. The estimator $\hat\beta_\HE$ works well in both normal and non-normal settings, and is robust to the correlation between $J_i$ and $\sigma_i^2$. It outperforms $\hat\beta_\FE$, $\hat\beta_\CW$ and $\hat\beta_\HO$, and is close to the infeasible best case.

\section{Empirical Application: Firm Discrimination} \label{sec:empirical}
In this section, we illustrate the use of regression on heteroskedastic individual-weight shrinkage estimators in the context of \citet{kline2022systemic}, which studies the
extent to which large U.S. employers systemically discriminate job applicants based on race. Their study utilizes correspondence audits, where fictitious resumes with randomized racial identifiers are sent to employers to measure differences in callback rates. The racial contact gap, defined as the difference in callback probabilities between racial groups, serves as the primary latent variable, analogous to value-added in our setting.

\subsection{Data}
We use the panel dataset in \citet{kline2022systemic} from an experiment that sends
fictitious applications to jobs posted by 108 of the largest U.S. employers. For each firm in each wave, about 25 entry-level vacancies were sampled and, for each vacancy, 8 job applications with randomly assigned
characteristics were sent to the employer. Sampling was organized in 5 waves. Focusing on firms sampled in all waves yields a balanced panel of $n = 70$
firms over 5 waves. Applications were sent in pairs, one randomly assigned a distinctively White name and the other a distinctively Black name. The primary outcome is whether the employer attempted to contact
the applicant within 30 days of applying. The racial contact gap is defined as the firm-level
difference between the contact rate (the ratio of the number of contacts to the number of received
applications) for White and that for Black applications. We follow the model similar to \citet[p.~41]{kline2022systemic}, who assumes that the racial contact gap is given by
\begin{align*}
  \bar{X}_i = \theta_i + \bar{\epsilon}_i, \quad \bar{\epsilon}_i \sim N\left(0, \frac{\sigma_i^2}{J_i} \right).
\end{align*}  
They justify the normality of $\bar \epsilon_{i}$ from the central limit theorem approximation with large numbers of measurements.

\subsection{Estimation}
In our analysis, we estimate the predictive effect of callback probability in wave \( t \), for \( t = 1,2,3 \), on callback probability in wave \( t+2 \). Given the model setup, we expect the regression coefficient to be close to 1. For each wave \( t \), we first apply shrinkage estimation, and subsequently regress on these estimates. The shrinkage estimators constructed by \citet{kline2020leave} are nonlinear shrinkage based on nonparametric empirical Bayes posterior means, while we focus on the performance of linear shrinkage estimators as regressors. We compare the performance of heteroskedastic individual-weight shrinkage estimator $\hat\beta_\HE$ with the fixed effect estimator $\hat\beta_\FE$ across the three waves. Since the discrimination gap is computed from pairs of job applications, we have \( J = 100 \), resulting in a small ratio of measurement error to sampling error, \( \sqrt{n}/J = 0.08 \). 

The results presented in \Cref{tab:regression_results} indicate that regression on the shrinkage estimates (\( \hat{\theta}_{i,\HE} \)) yields coefficients closer to 1. Despite a small \( \sqrt{n}/J \), $\hat\beta_\FE$ still exhibits attenuation bias. In terms of inference, $\hat\beta_\HE$ robustly rejects the null hypothesis of no predictive effect, whereas $\hat\beta_\FE$ fails to reject this null hypothesis for waves 1 and 2. These findings demonstrate that $\hat\beta_\HE$ improves the validity of inference for the regression coefficient.

\begin{table}[t]
  \centering
  \caption{Regression Results for Different Waves Predicting $t+2$}
  \resizebox{\textwidth}{!}{
  \begin{tabular}{lcccc|cccc}
    \toprule
    & \multicolumn{4}{c|}{$\hat\beta_\FE$} & \multicolumn{4}{c}{$\hat\beta_\HE$} \\
    \cmidrule(lr){2-5} \cmidrule(lr){6-9}
    Wave & \(\beta\) & SE & CI & \(p\)-value & \(\beta\) & SE & CI & \(p\)-value \\
    \midrule
    Wave 1 & 0.150 & 0.100 & [-0.047, 0.347] & 0.140 & 0.987 & 0.385 & [0.232, 1.742] & 0.013 \\
    Wave 2 & 0.092 & 0.110 & [-0.124, 0.308] & 0.406 & 0.883 & 0.363 & [0.171, 1.594] & 0.017 \\
    Wave 3 & 0.415 & 0.125 & [0.170, 0.659] & 0.001 & 2.186 & 0.963 & [0.299, 4.073] & 0.026 \\
 
    \bottomrule
  \end{tabular}
  }
  \vspace{0.5cm}
  \caption*{\parbox{0.9\textwidth}{\textbf{Note}: This table presents regression results for firm discrimination in wave $t$ predicting firm discrimination in wave $t+2$. Columns labeled \(\bar{X}_i\) represent regressions using the raw sample mean, while columns labeled \(\hat{\theta}_{i,\HE}\) correspond to regressions using the shrinkage estimates (HE). The coefficients (\(\beta\)) indicate the estimated effect of firm discrimination in wave $t$ on wave $t+2$. The standard error (SE), 95\% confidence interval (CI), and \(p\)-value are also reported, based on the Eicker--Huber--White variance estimator.}}

  \label{tab:regression_results}
\end{table}

\section{Application: School Value-Added in Pakistan} \label{sec:school}

In this section, we illustrate the use of regression on heteroskedastic individual-weight shrinkage estimators in the context of \citet{andrabi2025heterogeneity} about the school value-added in Pakistan. Their study investigates the relation between private school fees and school value-added (SVA), and finds evidence that parents are able to identify and reward school quality.

\subsection{Data and Empirical Setting}

The analysis utilizes the LEAPS project dataset, a rich longitudinal panel of student test scores from rural Punjab, Pakistan. This setting is characterized by the rapid emergence of a private school market, making the estimation of school quality (value-added) and its perceived return (school fees) a central question of interest. We use the school-year level sample from their analysis, which consists of $n=1158$ observations. We also have $\E_n \left[J_i^{-1}\right] = 0.120$. Therefore $\hat\kappa = \sqrt{n} \E_n \left[J_i^{-1}\right] \approx 4.08$.

\subsection{Estimation and Results}

We replicate and extend the primary downstream analysis in \citet{andrabi2025heterogeneity}, which is a regression of private school fees on estimated SVA. The latent variable $\theta_i$ represents the true SVA of a school, and the downstream regression investigates whether schools with higher SVA command higher fees.

We compare two individualized shrinkage estimators. The first is HO, as used in \citet{andrabi2025heterogeneity}. The second is HE, which is fully individualized by allowing for heterogeneity in both $J_i$ and the noise variance $\sigma_i^2$.

The results are presented in \Cref{tab:school_fees} (full sample) and \Cref{tab:school_fees_selected} (a selected sample restricting $20 \leq J_i \leq 80$, as a robustness check under less heterogeneity). The tables compare the HE estimator (right column) and the HO estimator (left column) in both bivariate specifications and specifications including household-level controls (parental education and asset index).

Across all specifications, our results confirm the central economic finding of \citet{andrabi2025heterogeneity}: SVA is a large, positive, and statistically significant predictor of private school fees. In the full-sample specification with controls (\Cref{tab:school_fees}, Row 2), the HE point estimate is 793.01 (significant at the 1\% level), which is comparable to the HO estimate. This comparison provides a valuable robustness check. Since one cannot know a priori if the assumptions required for valid inference with $\hat\beta_{\HO}$ hold, the stability of the coefficients under the more general $\hat\beta_{\HE}$ estimator strengthens the credibility of the original economic conclusion.

\begin{table}[t]
  \centering
  \caption{Regression Results of  Private School Fees on School Value-Added}
  \label{tab:school_fees}
  \resizebox{\textwidth}{!}{
  \begin{tabular}{lcccc|cccc}
  \toprule
  & \multicolumn{4}{c|}{\(\hat\beta_{\HO}\)} & \multicolumn{4}{c}{\(\hat\beta_{\HE}\)} \\
  \cmidrule(lr){2-5} \cmidrule(lr){6-9}
  Controls & \(\beta\) & SE & CI & \(p\)-value & \(\beta\) & SE & CI & \(p\)-value \\
  \midrule
  No 
& 991.37 
& 320.96 
& [362.69, 1{,}620.06] 
& 0.003 
& 1054.13 
& 298.15 
& [469.01, 1{,}639.25] 
& 0.001 \\
Yes 
& 719.86 
& 280.88 
& [169.34, 1{,}270.38] 
& 0.012 
& 793.01 
& 264.05 
& [275.48, 1{,}310.55] 
& 0.004 \\
  \bottomrule
  \end{tabular}
  }
  \vspace{0.5cm}
  \caption*{\parbox{0.9\textwidth}{\textbf{Note}: This table reports coefficients from regressions of private school fees on estimated school value-added using two estimators. Columns labeled \(\hat{\beta}_{\HO}\) use the homoskedastic individual-weight estimator following \citet{andrabi2025heterogeneity}, while columns labeled \(\hat{\beta}_{\HE}\) use the heteroskedastic individual-weight estimator. The rows correspond to specifications without controls and with controls for mean mother education, mean father education, and mean household asset index. For each specification, the table reports the coefficient (\(\beta\)), standard error (SE), 95\% confidence interval (CI), and associated \(p\)-value. Standard errors are clustered at the village level.}}
  
\end{table}

\begin{table}[t]
  \centering
  \caption{Regression Results of Private School Fees on School Value-Added\\ \centering (Selected Sample)}
  \label{tab:school_fees_selected}
  \resizebox{\textwidth}{!}{
  \begin{tabular}{lcccc|cccc}
  \toprule
  & \multicolumn{4}{c|}{\(\hat\beta_{\HO}\)} & \multicolumn{4}{c}{\(\hat\beta_{\HE}\)} \\
  \cmidrule(lr){2-5} \cmidrule(lr){6-9}
  Controls & \(\beta\) & SE & CI & \(p\)-value &
             \(\beta\) & SE & CI & \(p\)-value \\
  \midrule
  No 
  & 1569.70 
  & 451.27 
  & [685.22, 2454.18]
  & 0.000
  & 1622.13
  & 454.32
  & [733.66, 2510.60]
  & 0.000
  \\
  Yes 
  & 1320.18 
  & 384.74 
  & [566.90, 2073.46]
  & 0.000
  & 1358.51
  & 391.16
  & [592.06, 2124.96]
  & 0.000
  \\
  \bottomrule
  \end{tabular}
  }
  
  \vspace{0.3em}
  \caption*{\parbox{0.9\textwidth}{\textbf{Note}: This table reports coefficients from regressions of private school fees on estimated school value-added using two estimators. The sample is restricted to schools with between 20 and 80 students. Columns labeled \(\hat{\beta}_{\HO}\) use the homoskedastic individual-weight estimator following \citet{andrabi2025heterogeneity}, while columns labeled \(\hat{\beta}_{\HE}\) use the heteroskedastic individual-weight estimator. The rows correspond to specifications without controls and with controls for mean mother education, mean father education, and mean household asset index. For each specification, the table reports the coefficient (\(\beta\)), standard error (SE), 95\% confidence interval (CI), and associated \(p\)-value. Standard errors are clustered at the village level.}}
\end{table}

\section{Conclusion} \label{sec:conclusion}
This paper investigates the inferential properties of using individualized shrinkage estimators as covariates in a regression, a widely used but not fully understood empirical practice. We formalize the conditions under which this approach yields valid inference in the downstream regression model. Our central finding is that a correctly specified and fully individualized shrinkage estimator yields an asymptotically efficient estimator of the downstream regression coefficient. Crucially, we show that conventional OLS inference is asymptotically valid, justifying standard empirical practice. This result stands in contrast to the biased inference from regressing on the unshrunk fixed effect estimates (individual means). Similarly, the simpler individual-weight (HO) estimators  which are widely used in practice can fail to deliver valid inference when noise variance is correlated with the number of measurements. Our simulations confirm the predictions of the theory, including the finding that regression on individualized shrinkage estimators that do not properly account for noise heteroskedasticity suffers from amplification bias. We apply our method to data on firm discrimination and school value-added and show that it improves the estimation and inference.

Our analysis provides a formal bridge between the empirical Bayes literature, where shrinkage estimators were developed to improve estimation accuracy of individual effects, and the common empirical practice of using these estimates for downstream inference. The key takeaway for practitioners is that while the plug-in approach can be valid and efficient, its robustness depends critically on the specification of the shrinkage estimator. For applied researchers, our results provide a clear theoretical foundation and a practical guide for obtaining valid inference when using shrinkage estimates as regressors in linear models. Extending this method to nonlinear settings is a natural yet nontrivial direction, which we are studying in ongoing work.

\appendix

\section{Proofs of Main Results}

\subsection{Proofs for FE}

\begin{proof}[Proof of \Cref{prop:mean}]
  Firstly, for simplicity we abbreviate notations and denote $\hat\beta_{\FE}$ as $\hat\beta$, and $\Var\left(\theta_i\right)$ as $V$.
  \begin{align*}
    \sqrt{n}\left(\hat\beta - \beta\right) &=   \left(\frac{1}{n} \sum_{i=1}^{n} \left(\bar X_i - \bar X\right)^2\right)^{-1}\\
    & \quad \left( \beta \frac{1}{\sqrt{n}} \sum_{i=1}^{n} \left(\bar X_i - \bar X\right) \left(\theta_i - \bar \theta\right) - \beta \frac{1}{\sqrt{n}}\sum_{i=1}^{n} \left(\bar X_i - \bar X\right)^2 + \frac{1}{\sqrt{n}} \sum_{i=1}^{n} \left( \bar X_i - \bar X \right)\left(u_i - \bar u\right) \right)\\
    &= \frac{\beta \sqrt{n} T_{1,n} - \beta \sqrt{n}T_{2.n} + \sqrt{n}T_{3,n}}{T_{2,n}}.
  \end{align*}
  Firstly, from \Cref{lemma:properties_nonnormal} and \Cref{lemma:V_nonnormal}, we have for the denominator,
  \begin{align*}
    T_{2,n} &= \hat V + \frac{n-1}{n^2}\sum_{i=1}^{n} \frac{1}{J_i }\hat\sigma_i^2\\
    &= \hat V + O_p(n^{-1/2})\\
    &= V + O_p(n^{-1/2}).
  \end{align*}
  For the numerator terms, by properties from \Cref{lemma:properties_nonnormal},
  \begin{align*}
    \sqrt{n} T_{1,n} &= \frac{1}{\sqrt{n}}\sum_{i=1}^{n} \left(\theta_i - \E \left[\theta_i\right]\right)^2 + \frac{1}{\sqrt{n}} \sum_{i=1}^{n} \bar\epsilon_i \left(\theta_i - \E \left[\theta_i\right]\right) \\
    &\quad - \sqrt{n} \left(\bar\theta - \E \left[\theta_i\right]\right)^2 - \sqrt{n} \left(\bar\theta - \E \left[\theta_i\right]\right)\bar\epsilon\\
    &= \frac{1}{\sqrt{n}}\sum_{i=1}^{n} \left(\theta_i - \E \left[\theta_i\right]\right)^2 + o_p(1),
  \end{align*}
  \begin{align*}
    \sqrt{n} T_{3,n} &= \frac{1}{\sqrt{n}} \sum_{i=1}^{n} \left(\theta_i - \E\left[\theta_i\right]\right) u_i + \frac{1}{\sqrt{n}} \sum_{i=1}^{n} \bar\epsilon_i u_i\\
    &\quad - \sqrt{n} \left(\bar\theta - \E \left[\theta_i\right]\right) \bar u - \sqrt{n} \bar \epsilon \bar u \\
    &= \frac{1}{\sqrt{n}} \sum_{i=1}^{n} \left(\theta_i - \E\left[\theta_i\right]\right) u_i + o_p(1).
  \end{align*}
  Combined with the proof of \Cref{lemma:V_nonnormal}, we have
  \begin{align*}
    \sqrt{n}T_{2,n} &= \frac{1}{\sqrt{n}} \sum_{i=1}^{n}\left[ \left(\theta_i - \E\left[\theta_i\right]\right)^2 + \bar \epsilon_i^2\right] - \sqrt{n}\left( \bar\theta - \E\left[\theta_i\right] \right)^2 - \sqrt{n}\bar\epsilon^2 + \frac{2}{\sqrt{n}} \sum_{i=1}^{n} \left( \theta_i - \E \left[\theta_i\right]\right)\bar\epsilon_i  \\
    &\quad - 2 \sqrt{n}\left(\bar\theta-\E \left[\theta_i\right]\right)\bar\epsilon \\
    &= \frac{1}{\sqrt{n}} \sum_{i=1}^{n}\left[ \left(\theta_i - \E\left[\theta_i\right]\right)^2 + \bar \epsilon_i^2\right] + o_p(1).
  \end{align*}
  Therefore, the numerator is
  \begin{align*}
    &\quad \beta\sqrt{n} T_{1,n} - \beta\sqrt{n} T_{2,n} + \sqrt{n} T_{3,n} \\
    &= -\beta\frac{1}{\sqrt{n}}\sum_{i=1}^{n} \bar \epsilon_i^2 + \frac{1}{\sqrt{n}} \sum_{i=1}^{n} \left(\theta_i - \E\left[\theta_i\right]\right) u_i + o_p(1)\\
    &= \frac{1}{\sqrt{n}} \sum_{i=1}^{n}\left[ \left(\theta_i -\E \left[\theta_i\right]\right) u_i + \beta \frac{1}{J_i}\sigma_i^2-\beta\bar\epsilon_i^2 \right] - \beta \frac{1}{\sqrt{n}} \sum_{i=1}^{n} \frac{1}{J_i} \sigma_i^2 + o_p(1)\\
    &\coloneqq \frac{1}{\sqrt{n}} \sum_{i=1}^{n} \xi_i - \beta \frac{1}{\sqrt{n}} \sum_{i=1}^{n} \frac{1}{J_i} \sigma_i^2 + o_p(1).
  \end{align*}
  Here for the second term, by Chebyshev's inequality, for any $s>0$ we have
  \begin{align*}
    \Pr \left( \sqrt{n} \abs{ \frac{1}{n} \sum_{i=1}^{n} \frac{1}{J_i}\sigma_i^2 - \E \left[\frac{1}{J_i}\sigma_i^2\right] } >s \right) \leq \frac{\E \left[\frac{1}{J_i^2}\right]\E \left[\sigma_i^4\right]}{s^2} \to 0,
  \end{align*}
  and also,
  \begin{align*}
    \sqrt{n}\E \left[ \frac{1}{J_i}\sigma_i^2 \right] = \sqrt{n} \E \left[\frac{1}{J_i}\right]\E \left[\sigma_i^2\right] \to \kappa \E \left[\sigma_i^2\right].
  \end{align*}
  Therefore, the second term is $-\kappa\beta\E \left[ \sigma_i^2\right] + o_p(1)$.

  For $\xi_i$, since
  \begin{align*}
    \E \left[\xi_i\right] = 0,
  \end{align*}
  \begin{align}
    \E \left[ \xi_i^2 \right] &= \E \left[ u_i^2 \left(\theta_i-\E \left[\theta_i\right]\right)^2 \right] + \beta^2\E\left[\left(\frac{1}{J_i}\sigma_i^2\right)^2\right] + \beta^2\E \left[ \bar\epsilon_i^4 \right] - 2 \beta^2\E \left[\frac{1}{J_i}\sigma_i^2 \bar\epsilon_i^2 \right]\nonumber \\
    &= \E \left[ u_i^2 \left(\theta_i-\E \left[\theta_i\right]\right)^2 \right] + o(1) - 2 \beta^2\E\left[ \frac{1}{J_i^2} \sigma_i^4 \right] \nonumber\\
    &= \E \left[ u_i^2 \left(\theta_i-\E \left[\theta_i\right]\right)^2 \right] + o(1), \label{eq:xi_sq_mean}
  \end{align}
  where the last line follows from \Cref{lemma:properties_nonnormal}.

  Therefore, for $\left(n \E \left[ \xi_i^2\right]\right)^{-1/2} \sum_{i=1}^{n} \xi_i$ in the triangular array, by the Lindeberg-Feller theorem (See \citet{ferguson2017course} p.27), because the Lindeberg condition holds below
\begin{align}
   \frac{1}{n\E\left[\xi_i^2\right]}\sum_{i=1}^{n} \E \left\{ \xi_i^2 \indicator\left( \abs{\xi_i}  > s \sqrt{n \E \left[ \xi_i^2 \right]} \right) \right\} &= \frac{1}{\E \left[\xi_i^2\right]} \E \left\{ \xi_i^2 \indicator\left( \abs{\xi_i}  > s \sqrt{n \E \left[ \xi_i^2 \right]} \right) \right\} \nonumber\\
   &\to 0, \; \forall s > 0, \nonumber
\end{align}
which is derived from \eqref{eq:xi_sq_mean} and the dominated convergence theorem, then we have
\begin{align}
  \frac{1}{\sqrt{n \E \left[ \xi_i^2 \right]} } \sum_{i=1}^{n} \xi_i \to_d N(0,1). \nonumber
\end{align}
By \eqref{eq:xi_sq_mean} and Slutsky's theorem, 
\begin{align}
  \frac{1}{\sqrt{n}} \sum_{i=1}^{n} \xi_i - \frac{1}{\sqrt{n}} \sum_{i=1}^{n} \frac{1}{J_i} \sigma_i^2 \to_d N\left(-\kappa\beta\E\left[\sigma_i^2\right],\E \left[ u_i^2 \left(\theta_i-\E \left[\theta_i\right]\right)^2 \right]\right). \nonumber
\end{align}
Then combined with the denominator, by Slutsky's theorem,
\begin{align*}
\sqrt{n} \left( \hat\beta - \beta \right)= \frac{\beta \sqrt{n} T_{1,n} - \beta \sqrt{n} T_{2,n} + \sqrt{n} T_{3,n}}{T_{2,n}}\to_d N\left(-\kappa\beta\frac{\E \left[\sigma_i^2\right]}{V}, \frac{\E \left[ u_i^2 \left(\theta_i-\E \left[\theta_i\right]\right)^2 \right]}{V^2}\right).
\end{align*}

\end{proof}

\subsection{Proofs for Asymptotic Normality and Inference}

\begin{proof}[Proof of \Cref{lemma:beta_c}]
  Firstly, for simplicity we abbreviate notations and denote $\hat\beta_{c,\HE}$ as $\hat\beta_c$, $\hat\theta_{i,c,\HE}$ as $\hat\theta_{i,c}$, and $\Var\left(\theta_i\right)$ as $V$.

  We have   
  \begin{align}
    &\quad \sqrt{n} \left( \hat \beta_{c} - \beta\right) \nonumber \\
    &= \left(\frac{1}{n} \sum_{i=1}^{n} \left( \hat \theta_{i,c} - \bar{\hat \theta}_c\right)^2 \right)^{-1} \nonumber\\
    &\quad \left( \beta \frac{1}{\sqrt{n}} \sum_{i=1}^{n} \left(\hat \theta_{i,c} - \bar{\hat\theta}_c\right)\left( \theta_i - \bar\theta \right) - \beta \frac{1}{\sqrt{n}} \sum_{i=1}^{n} \left(\hat\theta_{i,c} - \bar{\hat\theta}_c\right)^2 + \frac{1}{\sqrt{n}} \sum_{i=1}^{n} \left(\hat\theta_{i,c} - \bar{\hat\theta}_c\right)\left(u_i - \bar{u}\right) \right)  \nonumber\\
    &\coloneqq \frac{\beta \sqrt{n} T_{1,n} -\beta \sqrt{n}T_{2,n} + \sqrt{n}T_{3,n}}{T_{2,n}}.\nonumber
  \end{align}
  Respectively, by properties from \Cref{lemma:properties_nonnormal},
  \begin{align*}
    \sqrt{n}T_{1,n} &=  \frac{1}{\sqrt{n}} \sum_{i=1}^{n} c_i \left[ \left( \theta_i - \bar\theta  \right)^2 + \left( \bar\epsilon_i -\bar \epsilon \right)\left( \theta_i - \bar\theta \right) \right]\\
    &= \frac{1}{\sqrt{n}}\sum_{i=1}^{n} c_i  \left(\theta_i - \E \left[\theta_i\right]\right)^2 +  \frac{1}{\sqrt{n}}\sum_{i=1}^{n}c_i\bar\epsilon_i \left(\theta_i - \E \left[\theta_i\right]\right)\\
    &\quad  - 2\sqrt{n} \left(\bar\theta - \E \left[\theta_i\right]\right)\frac{1}{n} \sum_{i=1}^{n} c_i \left(\theta_i - \E \left[\theta_i\right]\right) +  \sqrt{n}\left( \bar\theta - \E \left[\theta_i\right] \right)^2 \frac{1}{n} \sum_{i=1}^{n} c_i\\
    &\quad - \sqrt{n}\left( \bar\theta - \E \left[\theta_i\right] \right)\frac{1}{n}\sum_{i=1}^{n}c_i \bar\epsilon_i -  \sqrt{n}\bar\epsilon\frac{1}{n}\sum_{i=1}^{n} c_i \left(\theta_i - \E \left[\theta_i\right]\right) +  \sqrt{n} \bar\epsilon \left(\bar\theta - \E \left[\theta_i\right]\right)\frac{1}{n} \sum_{i=1}^{n} c_i \\
    &=  \frac{1}{\sqrt{n}}\sum_{i=1}^{n} c_i  \left(\theta_i - \E \left[\theta_i\right]\right)^2 + o_p(1),
  \end{align*}
  \begin{align*}
    T_{2,n} &= \frac{1}{n} \sum_{i=1}^{n}  \left( c_i \left(\theta_i -\bar\theta\right) + c_i \left(\bar\epsilon_i -  \bar\epsilon\right) \right)^2 +  \left[ \frac{1}{n} \sum_{i=1}^{n} \left(c_i \left(\theta_i - \bar\theta + \bar\epsilon_i - \bar\epsilon\right)\right) \right]^2\\
    &= \frac{1}{n} \sum_{i=1}^{n} c_i^2 \left[ \left(\theta_i - \bar \theta\right)^2 + \left(\bar\epsilon_i - \bar\epsilon\right)^2 + 2 \left(\theta_i - \bar \theta\right)\left(\bar\epsilon_i-\bar\epsilon \right)\right]\\
    &\quad + \left[ \frac{1}{n} \sum_{i=1}^{n}c_i \theta_i - \bar\theta \bar{c} + \frac{1}{n}\sum_{i=1}^{n} c_i \bar\epsilon_i - \bar\epsilon\bar c \right]^2\\
    &= \frac{1}{n}\sum_{i=1}^{n} c_i^2 \left[ \left(\theta_i - \E \left[\theta_i\right]\right)^2 + \bar\epsilon_i^2 + 2 \left( \theta_i - \E \left[\theta_i\right] \right)\bar\epsilon_i \right]\\
    &\quad - 2\left(\bar\theta - \E \left[\theta_i\right]\right)\frac{1}{n} \sum_{i=1}^{n} c_i^2 \left(\theta_i - \E \left[\theta_i\right]\right) +  \left( \bar\theta - \E \left[\theta_i\right] \right)^2 \frac{1}{n} \sum_{i=1}^{n} c_i^2\\
    &\quad - 2\bar\epsilon\frac{1}{n}\sum_{i=1}^{n}c_i^2 \bar\epsilon_i + \bar\epsilon^2\frac{1}{n}\sum_{i=1}^{n}c_i^2  \\
    &\quad - 2\left( \bar\theta - \E \left[\theta_i\right] \right)\frac{1}{n}\sum_{i=1}^{n}c_i^2 \bar\epsilon_i - 2 \bar\epsilon\frac{1}{n}\sum_{i=1}^{n} c_i^2 \left(\theta_i - \E \left[\theta_i\right]\right) +  2\bar\epsilon \left(\bar\theta - \E \left[\theta_i\right]\right)\frac{1}{n} \sum_{i=1}^{n} c_i^2 \\
    &\quad + \left[ \E \left[\theta_i\right] + O_p(n^{-1/2}) - \left(\E \left[\theta_i\right] + O_p(n^{-1/2})\right) \left( 1 + O_p(n^{-1/2}) \right) + o_p(n^{-1/2})\right]^2\\
    &= \frac{1}{n}\sum_{i=1}^{n} c_i^2 \left[\left(\theta_i - \E \left[\theta_i\right]\right)^2 + \bar\epsilon_i^2\right]  + o_p(n^{-1/2})\\
    &\quad - 2 O_p(n^{-1/2}) O_p(n^{-1/2}) + O_p(n^{-1})\left[1 + O_p(n^{-1/2})\right]\\
    &\quad - 2 o_p(n^{-1/2})o_p(n^{-1/2}) + o_p(n^{-1}) \left[1 + O_p(n^{-1/2})\right]\\
    &\quad - 2 O_p(n^{-1/2}) o_p(n^{-1/2}) - 2 o_p(n^{-1/2}) O_p(n^{-1/2}) + 2 o_p(n^{-1/2}) O_p(n^{-1/2}) \left[1 + O_p(n^{-1/2})\right]\\
    &\quad + O_p(n^{-1})\\
    &= \frac{1}{n}\sum_{i=1}^{n} c_i^2 \left[\left(\theta_i - \E \left[\theta_i\right]\right)^2 +\bar\epsilon_i^2\right]+ o_p(n^{-1/2}).
  \end{align*}
  \begin{align*}
    \sqrt{n}T_{3,n} &= \frac{1}{\sqrt{n}}\sum_{i=1}^{n}c_i \left[ \left(\theta_i - \bar\theta \right) \left(u_i -\bar u\right) + \left(\bar\epsilon_i -\bar\epsilon\right)\left(u_i - \bar u \right) \right]\\
    &\quad = \frac{1}{\sqrt{n}}\sum_{i=1}^{n}c_i \left[ \left(\theta_i - \E\left[\theta_i\right] \right) u_i\right] + \frac{1}{\sqrt{n}} \sum_{i=1}^{n} c_i \bar\epsilon_i u_i\\
    &\quad - \sqrt{n} \left(\bar\theta - \E \left[\theta_i\right]\right)\frac{1}{n}\sum_{i=1}^{n} c_i u_i - \sqrt{n} \bar u \frac{1}{n} \sum_{i=1}^{n}c_i \left(\theta_i - \E \left[\theta_i\right]\right) + \sqrt{n} \bar u \left(\bar\theta - \E \left[\theta_i\right]\right)\frac{1}{n} \sum_{i=1}^{n} c_i\\
    &\quad - \sqrt{n} \bar\epsilon \frac{1}{n} \sum_{i=1}^{n} c_i u_i - \sqrt{n} \bar u \frac{1}{n} \sum_{i=1}^{n} c_i \bar\epsilon_i + \sqrt{n} \bar u \bar\epsilon \frac{1}{n} \sum_{i=1}^{n} c_i\\
    &= \frac{1}{\sqrt{n}}\sum_{i=1}^{n}c_i \left[ \left(\theta_i - \E\left[\theta_i\right] \right) u_i\right] + o_p(1).
  \end{align*}
  Therefore, for the denominator we have
  \begin{align*}
    T_{2,n} &= \frac{1}{n}\sum_{i=1}^{n} c_i^2 \left[\left(\theta_i - \E \left[\theta_i\right]\right)^2 +\bar\epsilon_i^2\right] + o_p(n^{-1/2})\\
    &= \frac{1}{n}\sum_{i=1}^{n} c_i^2 \left(\theta_i - \E \left[\theta_i\right]\right)^2  + O_p(n^{-1/2})\\
    &= V + o_p(1).
  \end{align*}

  Also, combined with \Cref{lemma:sigma_nonnormal}, the numerator is
  \begin{align*}
    &\quad \beta\sqrt{n} T_{1,n} - \beta \sqrt{n}T_{2,n} + \sqrt{n}T_{3,n}\\
    &=  \frac{1}{\sqrt{n}}  \sum_{i=1}^{n}\left[\beta c_i\left(\theta_i - \E \left[\theta_i\right]\right)^2 - \beta c_i^2\left[ \left(\theta_i - \E \left[\theta_i\right]\right)^2 + \bar\epsilon_i^2\right] +  c_i \left(\theta_i - \E \left[\theta_i\right]\right)u_i\right] + o_p(1)\\
    &=  \frac{1}{\sqrt{n}}  \sum_{i=1}^{n}\left[\beta c_i\left(\theta_i - \E \left[\theta_i\right]\right)^2 - \beta c_i^2\left[ \left(\theta_i - \E \left[\theta_i\right]\right)^2 + \frac{1}{J_i}\hat\sigma_i^2\right] +  c_i \left(\theta_i - \E \left[\theta_i\right]\right)u_i\right] + o_p(1)\\
    &\coloneqq \frac{1}{\sqrt{n}} \sum_{i=1}^{n} \xi_i + o_p(1).
  \end{align*}
  For $\xi_i$, since
  \begin{align}
    \E \left[\xi_i\right] &= \beta \E\left[ c_i \left( \theta_i - \E \left[\theta_i\right]\right)^2 \right] - \beta \E \left[ c_i^2 \left( \left(\theta_i - \E \left[\theta_i\right]\right)^2 + \frac{1}{J_i}\hat\sigma_i^2\right) \right] \nonumber\\
    &= \beta V \E \left( c_i\right) - \beta \E \left[ c_i^2 \left(V + \frac{1}{J_i}\hat\sigma_i^2  \right)  \right] \nonumber\\
    &= \beta V \E \left( c_i\right) - \beta \E \left[ V c_i \right] \nonumber\\
    &= 0, \nonumber
  \end{align} 
  \begin{align}
    \E \left[ \xi_i^2 \right]  &= \beta^2  \E \left[\left( c_i - c_i^2 \right)^2 \right] \E \left[ \left(\theta_i - \E \left(\theta_i\right)\right)^4 \right] + \beta^2  \E \left[ \frac{1}{J_i^2}c_i^4 \hat\sigma_i^4\right] \nonumber \\
    &\quad  + \E \left[ c_i^2\right] \E\left[ u_i^2\left( \theta - \E \left[\theta_i\right] \right)^2 \right] \nonumber\\
    &\quad - 2\beta^2 V   \E \left[\frac{1}{J_i} \hat\sigma_i^2 c_i^2 \left( c_i - c_i^2\right)  \right] \nonumber\\
    &= o(1) +  \beta^2  \E \left[ c_i^4 \left(\frac{1}{J_i} \hat\sigma_i^2 \right)^2\right] \nonumber\\
    &\quad +  \E \left[ c_i^2\right] \E\left[ u_i^2\left( \theta - \E \left[\theta_i\right] \right)^2 \right] - 2\beta^2 V   \E \left[\frac{1}{J_i} \hat\sigma_i^2 c_i^2 \left( c_i - c_i^2\right)  \right] \nonumber \\
    &= \E\left[ u_i^2\left( \theta - \E \left[\theta_i\right] \right)^2 \right]  + o(1), \label{eq:xi_sq_2}
  \end{align}
  where the last line follows from \Cref{lemma:properties_nonnormal} and also $0 < c_i < 1$.

Therefore, for $\left(n \E \left[ \xi_i^2\right]\right)^{-1/2} \sum_{i=1}^{n} \xi_i$ in the triangular array, by the Lindeberg-Feller theorem (See \citet{ferguson2017course} p.27), because the Lindeberg condition holds below
\begin{align}
   \frac{1}{n\E\left[\xi_i^2\right]}\sum_{i=1}^{n} \E \left\{ \xi_i^2 \indicator\left( \left|\xi_i \right| > s \sqrt{n \E \left[ \xi_i^2 \right]} \right) \right\} &= \frac{1}{\E \left[\xi_i^2\right]} \E \left\{ \xi_i^2 \indicator\left( \left|\xi_i \right| > s \sqrt{n \E \left[ \xi_i^2 \right]} \right) \right\} \nonumber\\
   &\to 0, \; \forall s > 0, \nonumber
\end{align}
which is derived from \eqref{eq:xi_sq_2} and the dominated convergence theorem, then we have
\begin{align}
  \frac{1}{\sqrt{n \E \left[ \xi_i^2 \right]} } \sum_{i=1}^{n} \xi_i \to_d N(0,1). \nonumber
\end{align}
By \eqref{eq:xi_sq_2} and Slutsky's theorem, 
\begin{align}
  \frac{1}{\sqrt{n \sigma^2_u V}} \sum_{i=1}^{n} \xi_i \to_d N(0,1). \nonumber
\end{align}
Therefore,
\begin{align}
  \beta \sqrt{n}T_{1,n} - \beta\sqrt{n} T_{2,n} + \sqrt{n} T_{3,n} \to_d N\left(0, \E\left[ u_i^2\left( \theta - \E \left[\theta_i\right] \right)^2 \right]\right). \nonumber
\end{align}
Then combined with the denominator, by Slutsky's theorem,
\begin{align*}
\sqrt{n} \left( \hat \beta_c - \beta \right)= \frac{\beta \sqrt{n} T_{1,n} - \beta \sqrt{n} T_{2,n} + \sqrt{n} T_{3,n}}{T_{2,n}}\to_d N\left(0, \frac{\E\left[ u_i^2\left( \theta - \E \left[\theta_i\right] \right)^2 \right]}{V^2}\right).
\end{align*}
\end{proof}

\begin{proof}[Proof of \Cref{prop:beta}]
  Firstly, for simplicity we abbreviate notations and denote $\hat\beta_{c,\HE}$ as $\hat\beta_c$, $\hat\beta_\HE$ as $\hat\beta$, $\hat\theta_{i,c,\HE}$ as $\hat\theta_{i,c}$, $\hat\theta_{i,\HE}$ as $\hat\theta_i$ and $\Var\left(\theta_i\right)$ as $V$.
  We then prove the asymptotics by showing
\begin{align*}
  \sqrt{n} \left( \hat\beta - \beta\right) \to _p \sqrt{n}\left(\hat\beta_c - \beta\right).
\end{align*}
By taking the difference, it's equivalent to
\begin{align*}
  \sqrt{n}\hat\beta - \sqrt{n}\hat\beta_c =o_p(1).
\end{align*}
It suffices to show for the numerator part
\begin{align*}
  &\quad\frac{1}{\sqrt{n}} \sum_{i=1}^{n} \left( \hat\theta_i - \bar{\hat\theta}\right) \left( Y_i - \bar Y\right) \to_p \frac{1}{\sqrt{n}} \sum_{i=1}^{n} \left( \hat\theta_{i,c} - \bar{\hat\theta}_c\right) \left( Y_i - \bar Y\right),
\end{align*}
and then show for the denominator part
\begin{align*}
  &\quad \frac{1}{n} \sum_{i=1}^{n} \left( \hat\theta_i - \bar{\hat\theta}\right)^2 \to_p \frac{1}{n} \sum_{i=1}^{n} \left( \hat\theta_{i,c} - \bar{\hat\theta}_c\right)^2.
\end{align*}
First we show the numerator part.
Since $\hat\theta_i$ and $\hat\theta_{i,c}$ can be viewed as function values given $t = \hat V$ and $t=V$ for the function of $t$:
\begin{align*}
  \bar{X} + \frac{t}{\frac{1}{J_i} \hat\sigma_i^2 + t}\left(\bar X_i - \bar X\right),
\end{align*}
the first order derivative of which is denoted as
\begin{align*}
  A_i(t) &\coloneqq \frac{\frac{1}{J_i}\hat\sigma_i^2}{\left(\frac{1}{J_i} \hat\sigma_i^2 + t\right)^2} \left(\bar X_i - \bar X\right),
\end{align*}
for simplicity of notations, we denote the whole numerator part as a function $f_n(t)$, and what we want to show is
\begin{align*}
  f_n(\hat{V}) -f_n(V) = o_p(1).
\end{align*}
By the mean value expansion theorem, for $\tilde{V}$ such that $\left|\tilde{V} -V \right| \leq \left|\hat V - V \right|$,
\begin{align*}
   \left|  f_n\left(\hat V\right) - f_n\left(V\right) \right| &= \left|\hat V - V\right| \left| f^\prime_n \left(\tilde V\right) \right|\\
   &\leq \sqrt{n} \left| \hat V - V \right| \frac{1}{n} \sum_{i=1}^{n} \frac{\frac{1}{J_i} \hat\sigma_i^2}{\left(\frac{1}{J_i}\hat\sigma_i^2 + \tilde{V}\right)^2}  \left|  \left(\bar X_i - \bar{X}\right) \left(Y_i -\bar Y\right) \right|
\end{align*}

For any $s >0$, there exists $\delta > 0$ such that $V - \delta > 0$, and we have
\begin{align*}
  &\quad \Pr \left( \left|  f_n\left(\hat V\right) - f_n\left(V\right) \right| > s  \right)\\
  &\leq \Pr \left( \left|  f_n\left(\hat V\right) - f_n\left(V\right) \right| > s, \left|\hat V - V \right| \leq \delta \right) + \Pr \left( \left| \hat V - V \right| > \delta  \right)\\
  &\leq \Pr \left( \sqrt{n} \left| \hat V - V \right| \frac{1}{n} \sum_{i=1}^{n} \frac{\frac{1}{J_i} \hat\sigma_i^2}{\left(\frac{1}{J_i}\hat\sigma_i^2 + \tilde{V}\right)^2}  \left|  \left(\bar X_i - \bar{X}\right) \left(Y_i -\bar Y\right) \right| > s, \left|\hat V - V \right| \leq \delta  \right)\\
  &\quad + \Pr \left( \left| \hat V - V \right| > \delta  \right)\\
  &\leq \Pr \left( \sqrt{n} \left| \hat V - V \right| \frac{1}{n} \sum_{i=1}^{n} \frac{\frac{1}{J_i} \hat\sigma_i^2}{\left(\frac{1}{J_i}\hat\sigma_i^2 + V - \delta\right)^2}  \left|  \left(\bar X_i - \bar{X}\right) \left(Y_i -\bar Y\right) \right| > s, \left|\hat V - V \right| \leq \delta   \right)\\
  &\quad + \Pr \left( \left| \hat V - V \right| > \delta  \right)\\
  &\leq \Pr \left( \sqrt{n} \left| \hat V - V \right| \frac{1}{n} \sum_{i=1}^{n} \frac{\frac{1}{J_i} \hat\sigma_i^2}{\left(\frac{1}{J_i}\hat\sigma_i^2 + V - \delta\right)^2}  \left|  \left(\bar X_i - \bar{X}\right) \left(Y_i -\bar Y\right) \right| > s \right)\\
  &\quad + \Pr \left( \left| \hat V - V \right| > \delta  \right)\\
\end{align*}
Then we only need to show
\begin{align*}
  \frac{1}{n} \sum_{i=1}^{n} \frac{\frac{1}{J_i} \hat\sigma_i^2}{\left(\frac{1}{J_i}\hat\sigma_i^2 + V - \delta\right)^2}  \left|  \left(\bar X_i - \bar{X}\right) \left(Y_i -\bar Y\right) \right| = o_p(1).
\end{align*}
This can be proved if for any integers $k_1, k_2 \in \left\{ 0,1\right\}$,
\begin{align*}
  \frac{1}{n} \sum_{i=1}^{n} \frac{\frac{1}{J_i} \hat\sigma_i^2}{\left(\frac{1}{J_i}\hat\sigma_i^2 + V - \delta\right)^2}  \left| \bar X_i\right|^{k_1} \left|Y_i \right|^{k_2} = o_p(1).
\end{align*}
For any $s >0$, by Chebyshev's inequality,
\begin{align*}
  &\quad \Pr \left(\left|  \frac{1}{n} \sum_{i=1}^{n} \frac{\frac{1}{J_i} \hat\sigma_i^2}{\left(\frac{1}{J_i}\hat\sigma_i^2 + V - \delta\right)^2}  \left|  \bar X_i\right|^{k_1} \left|Y_i \right|^{k_2} - \E \left[ \frac{\frac{1}{J_i}\hat\sigma_i^2}{\left( \frac{1}{J_i}\hat\sigma_i^2 + V - \delta \right)^2} \left| \bar X_i\right|^{k_1} \left|Y_i \right|^{k_2}  \right]  \right| > s\right)\\
  &\leq \frac{\E \left[ \left( \frac{\frac{1}{J_i}\hat\sigma_i^2}{\left( \frac{1}{J_i}\hat\sigma_i^2 + V - \delta \right)^2} \right)^2 \left| \bar X_i\right|^{2k_1} \left|Y_i \right|^{2k_2}  \right]}{ns^2}\\
  &\leq \frac{\frac{1}{\left(V-\delta\right)^2} \E \left[ \left|\bar{X_i} \right|^{2k_1} \left(\alpha + \beta \theta_i + u_i\right)^{2k_2} \right]}{ns^2} \to 0.
\end{align*}
where the last limit follows from the independence and finite moments assumptions, as well as the fact from \Cref{lemma:properties_nonnormal} that $\E \left[ \left|\bar{X_i} \right|^{2k_1} \right] \to \E \left[\left| \theta_i\right|^{2k_1}\right]$.
Also, by the Cauchy-Schwarz inequality, from the property of $\left(\frac{1}{J_i}\hat\sigma_i^2\right)^2$ in \Cref{lemma:properties_nonnormal}, we also have
\begin{align*}
  \E \left[\frac{\frac{1}{J_i}\hat\sigma_i^2}{\left( \frac{1}{J_i}\hat\sigma_i^2 + V - \delta \right)^2} \left| \bar X_i\right|^{k_1} \left|Y_i \right|^{k_2} \right] = \E \left[ \frac{\frac{1}{J_i}\hat\sigma_i^2}{\left( \frac{1}{J_i}\hat\sigma_i^2 + V - \delta \right)^2} \right] \E \left[ \left| \bar X_i\right|^{k_1} \left|Y_i \right|^{k_2} \right]\\ \leq \frac{1}{\left(V - \delta\right)^2} \sqrt{\E \left[ \left( \frac{1}{J_i} \hat\sigma_i^2 \right)^2\right] \E \left[ \left| \bar X_i\right|^{2k_1} \left|Y_i \right|^{2k_2} \right]}\to 0.
\end{align*}
Therefore, 
\begin{align*}
  &\quad\frac{1}{n} \sum_{i=1}^{n} \frac{\frac{1}{J_i} \hat\sigma_i^2}{\left(\frac{1}{J_i}\hat\sigma_i^2 + V - \delta\right)^2}  \left|  \left(\bar X_i - \bar{X}\right) \left(Y_i -\bar Y\right) \right|\\
   &\leq \frac{1}{n} \sum_{i=1}^{n} \frac{\frac{1}{J_i} \hat\sigma_i^2}{\left(\frac{1}{J_i}\hat\sigma_i^2 + V - \delta\right)^2}  \left|\bar X_i Y_i \right| + \left| \bar{X}\right| \frac{1}{n} \sum_{i=1}^{n} \frac{\frac{1}{J_i} \hat\sigma_i^2}{\left(\frac{1}{J_i}\hat\sigma_i^2 + V - \delta\right)^2}  \left|Y_i \right|\\
   &\quad + \left| \bar{Y}\right| \frac{1}{n} \sum_{i=1}^{n} \frac{\frac{1}{J_i} \hat\sigma_i^2}{\left(\frac{1}{J_i}\hat\sigma_i^2 + V - \delta\right)^2}  \left|\bar X_i \right| + \left| \bar{X}\right| \left| \bar{Y}\right| \frac{1}{n} \sum_{i=1}^{n} \frac{\frac{1}{J_i} \hat\sigma_i^2}{\left(\frac{1}{J_i}\hat\sigma_i^2 + V - \delta\right)^2} \\
   &= o_p(1) + O_p(1) \cdot o_p(1) + O_p(1) \cdot o_p(1) + O_p(1) \cdot O_p(1) \cdot o_p(1) = o_p(1).
\end{align*}
Then the numerator part is shown.

Lastly, we show for the denominator part
\begin{align*}
  &\quad \frac{1}{n} \sum_{i=1}^{n} \left( \hat\theta_i - \bar{\hat\theta}\right)^2 \to_p \frac{1}{n} \sum_{i=1}^{n} \left( \hat\theta_{i,c} - \bar{\hat\theta}_c\right)^2.
\end{align*}
Similarly, for simplicity of notations, we denote the whole denominator part as a function $g_n \left( t\right)$. Similarly, we have
\begin{align*}
  g_n \left( \hat{V}\right) - g_n \left( V\right) &= \left(\hat{V} - V\right) g_n^\prime \left( \tilde V\right).
\end{align*}
Here with a slight abuse of notation, $\tilde{V}$ satisfies $\left|\tilde{V} - V \right| \leq \left| \hat V - V\right|$. We finish the proof by showing
\begin{align*}
  g_n^\prime \left( \tilde V\right) = O_p(1).
\end{align*}
Let
\begin{align*}
  D_i \coloneqq \frac{\frac{1}{J_i}\hat\sigma_i^2}{\left(\frac{1}{J_i}\hat\sigma_i^2 + \tilde V\right)^2}.
\end{align*}
\begin{align*}
  G_i \coloneqq \frac{\tilde V}{\frac{1}{J_i}\hat\sigma_i^2 + \tilde V}.
\end{align*}
\begin{align*}
  g_n^\prime \left( \tilde V\right) &= \frac{2}{n} \sum_{i=1}^{n} \left[ G_i \left(\bar X_i - \bar X\right) - \frac{1}{n}\sum_{k=1}^{n}G_k \left(\bar X_k - \bar X  \right) \right] D_i\left(\bar X_i - \bar X\right)\\
  &= \frac{2}{n} \sum_{i=1}^{n} G_i D_i \left(\bar X_i - \bar X\right)^2 - 2 \left[ \frac{1}{n} \sum_{i=1}^{n} G_i \left(\bar X_i - \bar X\right) \right] \left[\frac{1}{n} \sum_{i=1}^{n} D_i \left(\bar X_i - \bar X\right)  \right].
\end{align*}
Similarly,  to show it's $O_p(1)$, it suffices to show for $k_1, k_2 \in \left\{0,1\right\}$, $k_3 \in \left\{0,1,2\right\}$,
\begin{align*}
  \frac{1}{n} \sum_{i=1}^{n} G_i^{k_1} D_i^{k_2} \bar X_i^{k_3} = O_p(1).
\end{align*}
For any $M > 0$, there exists $0<\delta <V $, and
\begin{align*}
  &\quad \Pr \left(  \left| \frac{1}{n} \sum_{i=1}^{n} G_i^{k_1} D_i^{k_2} \bar X_i^{k_3}  \right| >M\right)\\
  &\leq \Pr \left(   \frac{1}{n} \sum_{i=1}^{n} G_i^{k_1} D_i^{k_2} \left| \bar X_i \right|^{k_3}  >M, \left|\tilde V - V \right| \leq \delta\right) + \Pr \left(\left|\tilde{V} - V \right| > \delta\right)\\
  &\leq \Pr \left( \frac{1}{n} \sum_{i=1}^{n}\left( \frac{V-\delta}{\frac{1}{J_i}\hat \sigma_i^2 + V-\delta}\right)^{k_1}  \left(\frac{\frac{1}{J_i}\hat\sigma_i^2}{\left(\frac{1}{J_i} \hat\sigma_i^2+ V - \delta\right)^2}\right)^{k_2} \left|\bar X_i \right|^{k_3} > M\right) + \Pr \left(\left|\tilde{V} - V \right| > \delta\right)\\
  &\leq \Pr \left( \frac{1}{\left(V-\delta\right)^{k_2}} \frac{1}{n} \sum_{i=1}^{n}\left|\bar X_i \right|^{k_3} > M   \right) + \Pr \left(\left|\tilde{V} - V \right| > \delta\right).
\end{align*}
Because
\begin{align*}
  \frac{1}{n} \sum_{i=1}^{n}\left|\bar X_i \right|^{k_3} = O_p(1),
\end{align*}
therefore it is proved.
\end{proof}

\begin{proof}[Proof of \Cref{prop:inference}]
  Firstly, for simplicity we abbreviate notations and denote $\hat\beta_{\HE}$ as $\hat\beta$, $\hat\theta_{i,\HE}$ as $\hat\theta_{i}$, and $\Var\left(\theta_i\right)$ as $V$.
  For the numerator, we have
  \begin{align*}
    &\quad \frac{1}{n} \sum_{i=1}^{n} \left(\hat\theta_i - \bar{\hat\theta}\right)^2 \hat u_i^2\\
    &= \frac{1}{n} \sum_{i=1}^{n} \left(\hat \theta_i - \bar{\hat\theta}\right)^2 \left(Y_i - \bar Y - \hat\beta\left(\hat\theta_i - \bar{\hat\theta}\right) \right)^2\\
    &= \frac{1}{n} \sum_{i=1}^{n} \left(\hat \theta_i - \bar{\hat\theta}\right)^2 \left(Y_i - \bar Y\right)^2 - 2\hat\beta \frac{1}{n} \sum_{i=1}^{n} \left(\hat \theta_i - \bar{\hat\theta}\right)^3 \left(Y_i - \bar Y\right)\\
    &\quad + \hat\beta^2 \frac{1}{n} \sum_{i=1}^{n} \left(\hat \theta_i - \bar{\hat\theta}\right)^4.\\
    & 
  \end{align*}
  Then by \Cref{lemma:inference_c}, in order to show
  \begin{align*}
    \frac{1}{n} \sum_{i=1}^{n} \left(\hat\theta_i - \bar{\hat\theta}\right)^2 \hat u_i^2 \to \E \left[ \left(\theta_i - \E \left[\theta_i\right]\right)^2 u_i^2 \right],
  \end{align*}
  it suffices to show for integers $2\leq k\leq4$,
  \begin{align*}
    \frac{1}{n} \sum_{i=1}^{n} \left(\hat\theta_i - \bar{\hat\theta}\right)^{k} \left( Y_i - \bar Y\right)^{4-k} \to_p \frac{1}{n} \sum_{i=1}^{n} \left(\hat\theta_{i,c} - \bar{\hat\theta}_c\right)^{k} \left( Y_i - \bar Y\right)^{4-k}.
  \end{align*}
  Because of the proof shown in \Cref{prop:beta}, the above can be shown if we have $2\leq k\leq4$,
  \begin{align*}
    \frac{1}{n} \sum_{i=1}^{n} \left(\hat\theta_i - \bar{\hat\theta}\right)^{k} Y_i ^{4-k} \to_p \frac{1}{n} \sum_{i=1}^{n} \left(\hat\theta_{i,c} - \bar{\hat\theta}_c\right)^{k} Y_i ^{4-k}
  \end{align*}
  For simplicity of notations, we denote the left hand side as the function value of $f_n(t)$ at $t = \hat V$, and the right hand side as the function value at $t = V$. Then we have
  \begin{align*}
    f_n\left(\hat V\right) - f_n\left(V\right) = \left(\hat V - V\right) f_n^\prime \left(\tilde V\right).
  \end{align*}
  Here $\tilde V$ satisfies $\left|\tilde V - V\right| \leq \left|\hat V - V\right|$. Then we only need to show
  \begin{align*}
    f_n^\prime \left(\tilde V\right) = O_p(1).
  \end{align*}
  Let
  \begin{align*}
    D_i \coloneqq \frac{\frac{1}{J_i}\hat\sigma_i^2}{\left(\frac{1}{J_i}\hat\sigma_i^2 + \tilde V\right)^2}.
  \end{align*}
  \begin{align*}
    G_i \coloneqq \frac{\tilde V}{\frac{1}{J_i}\hat\sigma_i^2 + \tilde V}.
  \end{align*}
  \begin{align*}
    &\quad f_n^\prime \left(\tilde V\right)\\
     &= \frac{k}{n} \sum_{i=1}^{n} \left[ G_i \left(\bar X_i - \bar X\right) - \frac{1}{n}\sum_{k=1}^{n}G_k \left(\bar X_k - \bar X  \right) \right]^{k-1} \left(D_i\left(\bar X_i - \bar X\right) - \frac{1}{n} \sum_{k=1}^{n} D_k \left(\bar X_k - \bar X\right) \right) Y_i^{4-k}\\
  \end{align*}
  It suffices to show for $k_1 \in \left\{ 1,2,3 \right\}$, $k_2 \in \left\{ 0,1\right\}$, $k_3 \in 1,2,3$,
  \begin{align*}
    \frac{1}{n} \sum_{i=1}^{n} G_i^{k_1} D_i^{k_2} \bar X_i^{k_3} Y_i^{4-k} = O_p(1).
  \end{align*}
  For any $M > 0$, there exists $0<\delta <V $, and
  \begin{align*}
    &\quad \Pr \left(  \left| \frac{1}{n} \sum_{i=1}^{n} G_i^{k_1} D_i^{k_2} \bar X_i^{k_3} Y_i^{4-k}  \right| >M\right) \\
    &\leq \Pr \left( \frac{1}{n} \sum_{i=1}^{n}\left( \frac{V-\delta}{\frac{1}{J_i}\hat \sigma_i^2 + V-\delta}\right)^{k_1}  \left(\frac{\frac{1}{J_i}\hat\sigma_i^2}{\left(\frac{1}{J_i} \hat\sigma_i^2+ V - \delta\right)^2}\right)^{k_2} \left|\bar X_i \right|^{k_3} Y_i^{4-k} > M\right) + \Pr \left(\left|\tilde{V} - V \right| > \delta\right).\\
    &\leq \Pr \left( \frac{1}{\left(V-\delta\right)^{k_2}} \frac{1}{n} \sum_{i=1}^{n}\left|\bar X_i \right|^{k_3} Y_i^{4-k} > M   \right) + \Pr \left(\left|\tilde{V} - V \right| > \delta\right).
  \end{align*}
  Because 
  \begin{align*}
    \frac{1}{n} \sum_{i=1}^{n}\left|\bar X_i \right|^{k_3} Y_i^{4-k} = O_p(1),
  \end{align*}
  therefore it is proved.

  For the denominator, from the proof of \Cref{prop:beta} and \Cref{lemma:beta_c} , we have
  \begin{align*}
    \left( \frac{1}{n} \sum_{i=1}^{n} \left(\hat\theta_i -\bar{\hat\theta}\right)^2 \right)\to_p V^2.
  \end{align*}
\end{proof}

\setlength\bibsep{0pt}
\bibliography{reference}

\newpage

\clearpage

\setcounter{page}{1}

\section*{Supplemental Appendix}

\section{Lemmas and Proofs for Asymptotic Normality}
\label{sec:appendix-lemmas-nonnormal}

Firstly, for simplicity we abbreviate notations and denote $\Var\left(\theta_i\right)$ as $V$.

\begin{remark}
  \label{remark:sigma_hat}
  Recall that the variance estimator is defined as \[\hat\sigma_i^2 \coloneqq \frac{1}{J_i - 1} \sum_{j=1}^{J_i} \left(X_{i,j} - \bar X_i\right)^2.\] Then $\hat\sigma_i^2$ is unbiased for $\sigma_i^2$ and we have for its variance that
  \begin{align}
    \Var\left( \hat\sigma_i^2 \mid \sigma_i^2 ,J_i \right) &= \frac{1}{J_i} \E\left[ \epsilon_{i,j}^4 \mid \sigma_i^2, J_i \right] - \frac{\sigma^2_i\left(J_i-3\right)}{J_i \left(J_i -1\right)} \nonumber \\
    &\leq \frac{1}{J_i}K\sigma_i^4 - \frac{\sigma^2_i\left(J_i-3\right)}{J_i \left(J_i -1\right)} &\text{(\Cref{assum:nonnormal}.3)} \nonumber\\
    &\leq \frac{1}{J_i}K\sigma_i^4, &\text{(\Cref{assum:nonnormal}.1)}  \label{eq:var_nonnormal}
  \end{align}
  where the equality follows from e.g. \citet{o2014some} (Result 3, p. 284).
\end{remark}

\begin{lemma}
  \label{lemma:properties_nonnormal}
  Under \Cref{assum:nonnormal}, we have:
  \begin{enumerate}
    \item Properties of $\bar \epsilon_i$, $\bar X_i$, $\hat\sigma_i^2$, $c_i$:
    
    For any integer $k_1 \in \left\{1,2\right\}$, $k_2 \geq 1$,  we have
    \begin{enumerate}
      \item $\bar\epsilon_i = O_p(n^{-1/4})$, and $\E\left[\left|\bar\epsilon_i \right|^{k_1} \right] \to 0$. \label{item:epsilon}
      \item $\bar X_i = \theta_i + O_p(n^{-1/4})$, and $\E \left[ \left| \bar X_i \right|^{k_1} \right] \to \E \left[ \left| \theta_i \right|^{k_1} \right]$.
      \item $\hat\sigma_i^2 = \sigma_i^2 + O_p(n^{-1/4})$, and $\E \left[ \left| \hat\sigma_i^2 - \sigma_i^2 \right|^{k_1} \right] \to 0$.
      \item $\E \left[ \frac{1}{J_i} \hat\sigma_i^2 \right] \to 0$, and $\E \left[ \left( \frac{1}{J_i} \hat\sigma_i^2 \right)^2 \right] \to 0$.
      \item $\E\left[ \left(\frac{1}{J_i} \hat\sigma_i^2 \right)^4\right] = o(n^{-1/2})$.  \label{item:angelova}
      \item $c_i^{k_2} = 1 + O_p(n^{-1/2})$, and $\E \left[ c_i^{k_2} \right] = 1 + O(n^{-1/2})$. \label{item:c} 
    \end{enumerate}

    \item Properties of sample moments of $\bar\epsilon_i$, $c_i$, $\theta_i$: \label{item:moment}
    
    For any integer $k_1\geq 0$\footnote{For statements like this, if $k=0$ is on the exponent, it means that the term is $1$.}, $0 \leq k_2 \leq 2$, and $k_3, k_4 \in \left\{0,1 \right\}$, we have 
    \begin{align*}
      &\frac{1}{n}\sum_{i=1}^{n} c_i^{k_1} \left(\theta_i - k_3 \E \left[\theta_i\right]\right)^{k_2} = \E \left[\left(\theta_i - k_3 \E \left[\theta_i\right]\right)^{k_2}\right] +O_p(n^{-1/2}), \\
      &\frac{1}{n}\sum_{i=1}^{n} c_i^{k_1} \left( \theta_i - \E \left[\theta_i\right]\right)^{k_4}\bar\epsilon_i = o_p(n^{-1/2}), \\
      &\frac{1}{n}\sum_{i=1}^{n} c_i^{k_1} \left( \theta_i- k_3 \E \left[\theta_i\right]\right)^{k_2}\bar\epsilon_i^2 = O_p(n^{-1/2}).
    \end{align*}

    \item Properties of sample means of $\theta_i$, $\bar X_i$, $\hat\sigma_i^2$: 
    \begin{enumerate}
      \item $\bar\theta = \E \left[\theta_i\right] + O_p(n^{-1/2})$.
      \item $\bar X = \E \left[\theta_i\right] + O_p(n^{-1/2})$.
      \item $\frac{1}{n}\sum_{i=1}^{n}\frac{1}{J_i}\hat\sigma_i^2 = O_p(n^{-1/2})$
    \end{enumerate}

    \item Properties of sample moments of $\epsilon_i$, $c_i$, $u_i$:
    
    For any integer $k\geq 0$, we have
    \begin{enumerate}
      \item $\frac{1}{n} \sum_{i=1}^{n} c_i^{k} u_i = O_p(n^{-1/2})$.
      \item $\frac{1}{n}\sum_{i=1}^{n} c_i^{k} u_i \bar \epsilon_i = o_p(n^{-1/2})$.
    \end{enumerate}
    
  \end{enumerate}
\end{lemma}

\begin{proof}[Proof of \Cref{lemma:properties_nonnormal}]
  Below we take any $s >0$, 
  \begin{enumerate}
    \item Properties of $\bar \epsilon_i$, $\bar X_i$, $\hat\sigma_i^2$, $c_i$:
    \begin{enumerate}
      \item  By Markov's inequality,
      \begin{align*}
        \Pr \left( n^{1/4}\left|\bar \epsilon_i \right| > s  \right) \leq \frac{\sqrt{n} \E \left[\bar\epsilon_i^2\right]}{s^2} = \frac{\sqrt{n}\E \left[\frac{1}{J_i}\right]\E \left[\sigma_i^2\right]}{s^2}.
      \end{align*}
      Since $\sqrt{n} \E \left[\frac{1}{J_i}\right] \to \kappa$, we have $\bar\epsilon_i = O_p(n^{-1/4})$.

      Since 
      \begin{align*}
        \E \left[\bar\epsilon_i^2\right] = \E \left[\frac{1}{J_i}\right] \E \left[\sigma_i^2\right] \to 0,
      \end{align*}
      combined with the Cauchy-Schwarz inequality, we have the rest of the results.
      \item This follows from \ref{item:epsilon}.
      \item By Markov's inequality and \eqref{eq:var_nonnormal},
      \begin{align*}
        \Pr \left( n^{1/4}\left|\hat\sigma_i^2 - \sigma_i^2 \right| > s   \right) &\leq \frac{\sqrt{n}\E\left[\left(\hat\sigma_i^2 - \sigma_i^2\right)^2\right]}{s^2}\\
         &\leq \frac{K \sqrt{n}\E \left[\frac{1}{J_i} \right]\E \left[\sigma_i^4\right]}{s^2}.
      \end{align*}
      Due to \Cref{assum:nonnormal}.7, we have $\hat\sigma_i^2 = \sigma_i^2 + O_p(n^{-1/4})$.
      
      Since by \eqref{eq:var_nonnormal},
      \begin{align*}
        \E \left[ \left(\hat\sigma_i^2 - \sigma_i^2\right)^2 \right] \leq K \E\left[\frac{1}{J_i}\right] \E\left[\sigma_i^4\right]  \to 0,
      \end{align*}
      combined with the Cauchy-Schwarz inequality, we have the rest of the results.
      \item We have
      \begin{align*}
        \E \left[ \frac{1}{J_i} \hat\sigma_i^2 \right] = \E \left[ \frac{1}{J_i} \right] \E \left[ \sigma_i^2 \right]\to0,
      \end{align*}
      \begin{align*}
        \E \left[ \left( \frac{1}{J_i} \hat\sigma_i^2 \right)^2 \right] \leq \E \left[ \frac{1}{J_i^2} \left( \frac{K\sigma_i^4}{J_i} + \sigma_i^4  \right)  \right] \to 0. & \quad \text{(By \eqref{eq:var_nonnormal})}
      \end{align*}
      \item From Theorem 2 of \citet{angelova2012moments}, the fourth moment of $\hat\sigma_i^2$ is
      \begin{align*}
        \E\left[ \left(\hat\sigma_i^2\right)^4 \mid J_i,\sigma_i^2\right] &= \mu_2^4 
        + \frac{6\mu_2^2(\mu_4 - \mu_2^2)}{J_i}
        + \frac{4\mu_6\mu_2 + 3\mu_4^2 - 18\mu_4\mu_2^2 - 24\mu_3^2\mu_2 + 23\mu_2^4}{J_i^2}\\
        &\quad + \frac{\mu_8 - 4\mu_6\mu_2 - 24\mu_5\mu_3 - 3\mu_4^2 + 72\mu_4\mu_2^2 + 96\mu_3^2\mu_2 - 86\mu_2^4}{J_i^3}\\
        &\quad + \frac{4(6\mu_6\mu_2 + 6\mu_4^2 - 39\mu_4\mu_2^2 - 40\mu_3^2\mu_2 + 45\mu_2^4)}{J_i^3(J_i-1)}\\
        &\quad + \frac{4(36\mu_4\mu_2^2 - 8\mu_5\mu_3 + 52\mu_3^2\mu_2 - 61\mu_2^4)}{J_i^3(J_i-1)^2}\\
        &\quad + \frac{8(\mu_4^2 - 6\mu_4\mu_2^2 - 12\mu_3^2\mu_2 + 15\mu_2^4)}{J_i^3(J_i-1)^3},
      \end{align*}
      where $\mu_k, k=1,...,8$ are the $k$-th raw moments of $\epsilon_{i,j} \mid \sigma_i^2$. By \Cref{assum:nonnormal}.3, we have $\mu_k \leq K \sigma_i^k$ for some constant $K$. Thus, by independence of \Cref{assum:nonnormal}.1 and \Cref{assum:nonnormal}.8, we have that
      \begin{align*}
        \sqrt{n} \E\left[\left(\frac{1}{J_i}\hat\sigma_i^2\right)^4 \right] \to 0.
      \end{align*}

      \item
       By the mean value theorem,
      \begin{align*}
        c_i^{k_2} &= \left( \frac{V}{V + \frac{1}{J_i} \hat\sigma_i^2}\right)^{k_2}\\
        &= 1-k_2 \frac{1}{VJ_i} \hat\sigma_i^2 + \frac{k_2 \left(k_2 +1\right)}{2} \frac{1}{\left(1 + \omega_i\right)^{k_2 + 2}} \frac{1}{V^2 J_i^2} \left(\hat\sigma_i^2\right)^2,
      \end{align*}
      where $\omega_i$ is between $0$ and $\frac{1}{VJ_i} \hat\sigma_i^2$. Therefore,
      \begin{align*}
        \abs{\sqrt{n}\E\left[ 1- c_i^{k_2} \right]} &\leq \frac{k_2}{V}\sqrt{n}\E\left[\frac{\sigma^2_i}{J_i} \right] \\
        &\quad + \frac{k_2\left(k_2 +1\right)}{2V^2} \sqrt{n} \E\left[ \frac{1}{J_i^2} \left( \sigma_i^4 + \frac{K}{J_i}\sigma_i^4 \right) \right] &\text{(By \eqref{eq:var_nonnormal})}\\
        &\leq \frac{k_2}{V} \E\left[\sigma_i^2\right] \sqrt{n}\E\left[\frac{1}{J_i}\right] \\
        &\quad + \frac{k_2\left(k_2 +1\right)}{2V^2} \E\left[\sigma_i^4\right] \sqrt{n} \E\left[\frac{1}{J_i^2}\right] \\
        &\quad + \frac{k_2\left(k_2 +1\right)K}{2V^2} \E\left[\sigma_i^4\right] \sqrt{n} \E\left[\frac{1}{J_i^3}\right] & \text{(By \Cref{assum:nonnormal}.1)}\\
        &\to \frac{k_2}{V} \E\left[\sigma_i^2\right] \kappa + 0 + 0. &\text{(By \Cref{assum:nonnormal}.7 and  \ref{assum:nonnormal}.8)}
      \end{align*}
       Then we have the second result. By Markov's inequality, we have the first result.
    \end{enumerate}
    \item 
      By Chebyshev's inequality ($k_2$ can be replaced by $k_4$),
\begin{align*}
  &\quad \Pr\left( \sqrt{n}\left| \frac{1}{n} \sum_{i=1}^{n} c_i^{k_1} \left(\theta_i - k_3 \E\left[\theta_i\right]\right)^{k_2} - \E \left[ c_i^{k_1} \left(\theta_i - k_3 \E\left[\theta_i\right]\right)^{k_2}  \right] \right| > s \right)\\
   &\leq \frac{\E \left[ c_i^{2k_1} \left(\theta_i - k_3 \E\left[\theta_i\right]\right)^{2k_2} \right]}{s^2} \leq \frac{\E \left[\left(\theta_i - k_3 \E\left[\theta_i\right]\right)^{2k_2}\right]}{s^2}
\end{align*}
\begin{align*}
  &\quad \Pr \left( \sqrt{n}\left| \frac{1}{n} \sum_{i=1}^{n} c_i^{k_1} \left(\theta_i - k_3 \E\left[\theta_i\right]\right)^{k_2} \bar\epsilon_i - \E \left[ c_i^{k_1} \left(\theta_i - k_3 \E\left[\theta_i\right]\right)^{k_2} \bar\epsilon_i \right] \right| > s \right)\\
  & \leq \frac{\E \left[ c_i^{2k_1} \left(\theta_i - k_3 \E\left[\theta_i\right]\right)^{2k_2} \bar\epsilon_i^{2}  \right]}{s^2} \leq \frac{\E \left[ \left(\theta_i - k_3 \E\left[\theta_i\right]\right)^{2k_2} \right] \E\left[ \frac{1}{J_i} \right] \E \left[\sigma_i^2\right]}{s^2} \to 0,
\end{align*}
\begin{align*}
  &\quad \Pr \left( \sqrt{n}\left| \frac{1}{n} \sum_{i=1}^{n} c_i^{k_1} \left(\theta_i - k_3 \E\left[\theta_i\right]\right)^{k_2} \bar\epsilon_i^2 - \E \left[ c_i^{k_1} \left(\theta_i - k_3 \E\left[\theta_i\right]\right)^{k_2} \bar\epsilon_i^2 \right] \right| > s \right) \\
  &\leq \frac{\E \left[ c_i^{2k_1} \left(\theta_i - k_3 \E\left[\theta_i\right]\right)^{2k_2} \bar\epsilon_i^{4}  \right]}{s^2}\\
  &\leq \frac{\E \left[ \left(\theta_i - k_3 \E\left[\theta_i\right]\right)^{2k_2} \right] \left( \E\left[ \frac{1}{J_i^3} \right] \E \left[K\sigma_i^4\right] + 3\E\left[\frac{J_i-1}{J_i^3} \right] \E\left[\sigma_i^4\right] \right)}{s^2} \to 0.
\end{align*}
Meanwhile, by independence from \Cref{assum:nonnormal}.3,
\begin{align*}
  &\quad \E \left[ c_i^{k_1}\left(\theta_i - k_3 \E\left[\theta_i\right]\right)^{k_2} \right] =\E \left[ c_i^{k_1} \right] \E \left[ \left(\theta_i - k_3 \E\left[\theta_i\right]\right)^{k_2} \right]\\
   &= \left(1 + O(n^{-1/2}) \right) \E \left[ \left(\theta_i - k_3 \E\left[\theta_i\right]\right)^{k_2} \right] &\text{(By \ref{item:c})}\\
   &= \E \left[ \left(\theta_i - k_3 \E\left[\theta_i\right]\right)^{k_2} \right] + O(n^{-1/2}),
\end{align*}
by the mean value theorem ($k_4 = 0$) and \eqref{eq:var_nonnormal},
\begin{align*}
  \abs{\sqrt{n}\E\left[c_i^{k_1} \bar\epsilon_i\right]} &\leq 0 + \frac{k_1}{V}\sqrt{n}\E\left[ \frac{1}{J_i}\hat\sigma_i^2 \bar\epsilon_i \right] \\
  &\quad + \frac{k_1\left(k_1 +1\right)}{2V^2} \sqrt{n} \E\left[ \frac{1}{J_i^2} \left(\hat\sigma_i^2\right)^2 \bar \epsilon_i \right] \\
  &\leq \frac{k_1}{V} \sqrt{ \sqrt{n} \E\left[\frac{1}{J_i^2} \left( \sigma_i^4 + \frac{K}{J_i}\sigma_i^4 \right) \right] \sqrt{n} \E\left[ \frac{1}{J_i}\sigma_i^2 \right] }\\
  &\quad + \frac{k_1\left(k_1 +1\right)}{2V^2} \sqrt{ \sqrt{n}\E\left[ \frac{1}{J_i^4}\left(\hat\sigma_i^2 \right)^4 \right] \sqrt{n}\E\left[ \frac{1}{J_i}\sigma_i^2 \right]  } &\text{(Cauchy-Schwarz)}\\
  &\to 0,
\end{align*}
where the last step follows from independence of \Cref{assum:nonnormal}.1, \Cref{assum:nonnormal}.7, \Cref{assum:nonnormal}.8 and \ref{item:angelova},

by the independence in \Cref{assum:nonnormal}.3 ($k_4 = 1$),
\begin{align*}
  &\E \left[ c_i^{k_1} \left(\theta_i - \E\left[\theta_i\right]\right) \bar\epsilon_i \right]  = \E \left[\E\left( c_i^{k_1}\left(\theta_i -\E\left[\theta_i\right]\right)\bar{\epsilon_i}  \bigg| \epsilon \right) \right] = 0,
\end{align*}
and also by \Cref{assum:nonnormal}.3,
\begin{align*}
  \left|\sqrt{n} \E \left[ c_i^{k_1} \left(\theta_i - k_3 \E\left[\theta_i\right]\right)^{k_2}\bar\epsilon_i^2 \right] \right| \leq \E \left[ \left| \left(\theta_i - k_3 \E\left[\theta_i\right]\right)^{k_2}  \right| \right] \sqrt{n} \E \left[ \frac{1}{J_i} \right] \E \left[ \sigma_i^2\right],
\end{align*}
where $\sqrt{n} \E \left[\frac{1}{J_i}\right] \to \kappa$. Therefore we have the results.
   \item  \begin{enumerate}
    \item Follows from the finite fourth moment and the inequality below for any $s>0$, 
    \begin{align*}
      &\Pr \left( \left| \sqrt{n} \left( \bar\theta - \E\left[ \theta_i\right] \right) \right| > s  \right) \leq \frac{\E\left[\theta_i^2\right]}{s^2}.  &\text{(Chebyshev's inequality)}
    \end{align*}
    \item Follows from the property of $\bar\theta$ and the second property of \ref{item:moment}:
    \begin{align*}
      \bar X = \bar\theta + \bar\epsilon = \E\left[\theta_i\right] + O_p(n^{-1/2}) + o_p(n^{-1/2}) = \E\left[\theta_i\right] + O_p(n^{-1/2}).
    \end{align*}
    \item By Chebyshev's inequality and \eqref{eq:var_nonnormal}, for any $s>0$,
    \begin{align*}
      \Pr \left(\sqrt{n}\left| \frac{1}{n} \sum_{i=1}^{n} \frac{1}{J_i} \hat\sigma_i^2 - \E \left[ \frac{1}{J_i} \hat\sigma_i^2\right]  \right| > s \right) \leq \frac{\E \left[\frac{1}{J_i^2} \left(  \frac{K\sigma_i^4}{J_i}+ \sigma_i^4 \right)\right]}{s^2} \to 0,
    \end{align*}
    and also
    \begin{align*}
      \sqrt{n}\E \left[\frac{1}{J_i} \hat\sigma_i^2 \right] = \sqrt{n}\E \left[ \frac{1}{J_i} \right] \E \left[\sigma_i^2\right] \to \kappa \E \left[\sigma_i^2\right].
    \end{align*}
    Thus
    \begin{align*}
      \frac{1}{n}\sum_{i=1}^{n} \frac{1}{J_i} \hat \sigma_i^2 = O_p(n^{-1/2}).
    \end{align*}
   \end{enumerate}
   \item \begin{enumerate}
    \item Since $u_i \ind \epsilon_{i,j} \mid \theta_i$, by Chebyshev's inequality,
    \begin{align*}
      \Pr \left( \sqrt{n} \left| \frac{1}{n} \sum_{i=1}^{n} c_i^k u_i \right|>s \right) \leq \frac{\E \left[ c_i^{2k} u_i^2 \right]}{s^2} \leq \frac{\sigma^2_u}{s^2}
    \end{align*}
    Thus 
    \begin{align*}
      \frac{1}{n} \sum_{i=1}^{n} c_i^k u_i = O_p(n^{-1/2}).
    \end{align*}
    \item By Chebyshev's inequality,
    \begin{align*}
      \Pr \left( \sqrt{n} \left| \frac{1}{n} \sum_{i=1}^{n} c_i^k u_i \bar\epsilon_i \right|>s \right) \leq \frac{\E \left[ c_i^{2k} u_i^2 \bar\epsilon_i^2\right]}{s^2} \leq \frac{\sigma^2_u}{s^2} \E \left[\frac{1}{J_i }\right] \E \left[\sigma_i^2\right] \to 0.
    \end{align*}
    Thus 
    \begin{align*}
      \frac{1}{n} \sum_{i=1}^{n} c_i^k u_i\bar\epsilon_i = o_p(n^{-1/2}).
    \end{align*}
   \end{enumerate}
  \end{enumerate}
\end{proof}

\begin{lemma}
  \label{lemma:sigma_nonnormal}
  Under \Cref{assum:nonnormal}, we have
  \begin{align*}
    \frac{1}{n} \sum_{i=1}^{n} c_i^2 \bar\epsilon_i^2 = \frac{1}{n} \sum_{i=1}^{n} c_i^2 \frac{1}{J_i}\hat\sigma_i^2 + o_p(n^{-1/2}).
  \end{align*}
\end{lemma}
\begin{proof}[Proof of \Cref{lemma:sigma_nonnormal}]
  First, notice that the difference
  \begin{align*}
    \bar \epsilon_i^2  - \frac{1}{J_i} \hat\sigma_i^2&= \bar\epsilon_i^2 - \frac{1}{J_i\left(J_i-1\right)} \sum_{j=1}^{J_i} \epsilon_{i,j}^2 + \frac{1}{J_i-1}\bar\epsilon_i^2 \\
    &= \frac{1}{J_i\left(J_i -1\right)} \left(\sum_{j=1}^{J_i}\epsilon_{i,j} \right)^2 - \frac{1}{J_i\left(J_i -1\right)} \sum_{j=1}^{J_i} \epsilon_{i,j}^2 \\
    &= \frac{2}{J_i \left(J_i -1\right)} \sum_{k\leq j} \epsilon_{i,k}\epsilon_{i,j}.
  \end{align*}
  Then the properties of the difference are
  \begin{align*}
    &\E \left[ \bar\epsilon_i^2 - \frac{1}{J_i} \hat\sigma_i^2 \right] = 0,\\
    &\E \left[ \left(\bar\epsilon_i^2 - \frac{1}{J_i} \hat\sigma_i^2\right)^2 \mid J_i, \sigma_i^2\right] = \frac{2\sigma_i^4}{J_i\left(J_i -1\right)}.
  \end{align*}
  For $c_i^2$, by the mean value theorem,
  \begin{align*}
    c_i^2 &= \left( \frac{V}{V + \frac{1}{J_i}\hat\sigma_i^2}\right)^2\\
    &= 1 - 2\frac{1}{VJ_i} \hat\sigma_i^2 + 3\frac{1}{\left(1 + \omega_i\right)^4} \frac{1}{V^2J_i^2}\left(\hat\sigma_i^2\right)^2,
  \end{align*}
  where $\omega_i$ is between $0$ and $\frac{1}{VJ_i}\hat\sigma_i^2$. Therefore,

  \begin{align*}
    &\quad \abs{\sqrt{n}\E\left[ c_i^2 \left( \bar\epsilon_i^2 - \frac{1}{J_i}\hat\sigma_i^2\right) \right]}\\
    &\leq \abs{\sqrt{n} \E\left[  \left( \bar\epsilon_i^2 - \frac{1}{J_i}\hat\sigma_i^2\right) \right]} + \frac{2}{V} \abs{\sqrt{n} \E \left[ \frac{1}{J_i} \hat\sigma_i^2 \left( \bar\epsilon_i^2 - \frac{1}{J_i}\hat\sigma_i^2\right) \right]} \\
    &\quad + \frac{3}{V^2}\abs{\sqrt{n} \E \left[ \frac{1}{\left(1 + \omega_i\right)^4} \frac{1}{J_i^2}\left(\hat\sigma_i^2\right)^2 \left( \bar\epsilon_i^2 - \frac{1}{J_i}\hat\sigma_i^2\right) \right]} \\
    &\leq 0 + \frac{2}{V} \sqrt{n \E\left[\frac{1}{J_i^2} \left(\hat\sigma_i^2 \right)^2 \right] \E\left[\left(\bar\epsilon_i^2 - \frac{1}{J_i}\hat\sigma_i^2  \right)^2  \right] }\\
    &\quad + \frac{3}{V^2} \sqrt{n \E\left[\frac{1}{J_i^4} \left(\hat\sigma_i^2\right)^4 \right] \E\left[\left(\bar\epsilon_i^2 - \frac{1}{J_i}\hat\sigma_i^2  \right)^2  \right] } &\text{(Cauchy-Schwarz)}\\
    &\leq \frac{2}{V} \sqrt{\sqrt{n} \E\left[\frac{1}{J_i^2}\left( \frac{K}{J_i} \sigma_i^4 + \sigma_i^4  \right)   \right]\sqrt{n} \E\left[ \frac{2\sigma_i^4}{J_i \left(J_i -1\right)}  \right] }\\
    &\quad + \frac{3}{V^2} \sqrt{\sqrt{n} \E\left[\frac{1}{J_i^4}\left(  \hat\sigma_i^2\right)^4   \right]\sqrt{n} \E\left[ \frac{2\sigma_i^4}{J_i \left(J_i -1\right)}  \right] }. &\text{(By properties of the difference)}\\
  \end{align*}
  Since we have independence of $\sigma_i^2 \ind J_i$ from \Cref{assum:nonnormal}.1, \Cref{assum:nonnormal}.7, \Cref{assum:nonnormal}.8, and that
  \begin{align*}
    \sqrt{n} \E \left[\frac{1}{J_i^4} \left(\hat\sigma_i^2\right)^4 \right] \to 0,
  \end{align*}
  from \Cref{lemma:properties_nonnormal}.\ref{item:angelova}, the above converges to $0$.
  Thus, the expectation
  \begin{align*}
    \sqrt{n} \E\left[ c_i^2 \left( \bar\epsilon_i^2 - \frac{1}{J_i}\hat\sigma_i^2\right) \right] \to 0.
  \end{align*}
  Next, we show the convergence to expectations by Chebyshev's inequality:
  \begin{align*}
    &\quad \Pr \left( \sqrt{n}\left| \frac{1}{n} \sum_{i=1}^{n} c_i^2 \left( \bar\epsilon_i^2 - \frac{1}{J_i}\hat\sigma_i^2\right) - \E \left[ c_i^2 \left( \bar\epsilon_i^2 - \frac{1}{J_i}\hat\sigma_i^2\right)  \right] \right| > s  \right)\\
    &\leq \frac{\E \left[ c_i^4 \left( \bar\epsilon_i^2 - \frac{1}{J_i}\hat\sigma_i^2\right)^2  \right]}{s^2} \leq \frac{\E \left[  \left( \bar\epsilon_i^2 - \frac{1}{J_i}\hat\sigma_i^2\right)^2  \right]}{s^2} = \frac{\E\left[ \frac{2\sigma_i^4}{J_i \left(J_i -1\right)}  \right]}{s^2}  \to 0.
  \end{align*}
  And the proof is complete.
\end{proof}

Recall that we define in the beginning of this section $V \coloneqq \Var\left(\theta_i\right)$ and in \Cref{sec:HE}
\begin{align*}
  \hat V \coloneqq \frac{1}{n} \sum_{i=1}^{n} \left( \bar X_i - \bar X \right)^2 - \frac{n-1}{n^2} \sum_{i=1}^{n} \frac{1}{J_i} \hat \sigma^2_i.
\end{align*}
The following lemma shows that $\hat V$ converges to $V$ at $\sqrt{n}$-rate.
\begin{lemma}
  \label{lemma:V_nonnormal}
  Under \Cref{assum:nonnormal}, we have
  \begin{align*}
    \hat V = V + O_p(n^{-1/2}).
  \end{align*}
\end{lemma}
\begin{proof}[Proof of \Cref{lemma:V_nonnormal}]
  By \Cref{lemma:properties_nonnormal},
  \begin{align}
    \hat{V} &= \frac{1}{n} \sum_{i=1}^n \left( \bar{X}_i- \bar{X} \right)^2 - \frac{n-1}{n^2} \sum_{i=1}^{n} \frac{1}{J_i} \hat \sigma^2_i \nonumber\\
    &= \frac{1}{n} \sum_{i=1}^{n} \left(\theta_i - \E\left[\theta_i\right]\right)^2 + \frac{1}{n} \sum_{i=1}^{n} \bar \epsilon_i^2 - \left( \bar\theta - \E\left[\theta_i\right] \right)^2 - \bar\epsilon^2 + \frac{2}{n} \sum_{i=1}^{n} \left( \theta_i - \E \left[\theta_i\right]\right)\bar\epsilon_i  \nonumber\\
    &\quad - 2 \left(\bar\theta-\E \left[\theta_i\right]\right)\bar\epsilon  - \frac{n-1}{n^2} \sum_{i=1}^{n} \frac{1}{J_i} \hat\sigma_i^2 \nonumber\\
    &= V + O_p(n^{-1/2}) + O_p(n^{-1/2}) - O_p(n^{-1}) - o_p(n^{-1}) + o_p(n^{-1/2}) \nonumber \\
    &\quad - o_p(n^{-1}) - O_p(n^{-1/2}) \nonumber\\
    &= V + O_p(n^{-1/2}). \nonumber
  \end{align}
\end{proof}

We prove the common-weight result here.
\begin{proof}[Proof of \Cref{prop:CW}]
  Firstly, for simplicity we abbreviate notations and denote $\hat\beta_{\CW}$ as $\hat\beta$, and $\Var\left(\theta_i\right)$ as $V$.
  \begin{align*}
    \sqrt{n}\left(\hat\beta - \beta\right) &=   \hat V ^{-1}\left( \beta \frac{1}{\sqrt{n}} \sum_{i=1}^{n} \left(\bar X_i - \bar X\right) \left(\theta_i - \bar \theta\right) - \beta \sqrt{n}\hat V + \frac{1}{\sqrt{n}} \sum_{i=1}^{n} \left( \bar X_i - \bar X \right)\left(u_i - \bar u\right) \right)\\
    &= \frac{\beta \sqrt{n} T_{1,n} - \beta \sqrt{n}T_{2.n} + \sqrt{n}T_{3,n}}{T_{2,n}}.
  \end{align*}
  Firstly, from \Cref{lemma:properties_nonnormal} and from the proof of \Cref{lemma:sigma_nonnormal}, we have for the denominator,
  \begin{align*}
    T_{2,n} &= \frac{1}{n} \sum_{i=1}^{n}  \left(\bar X_i - \bar X\right)^2 - \frac{1}{n} \sum_{i=1}^{n}\frac{1}{J_i} \hat\sigma_i^2  + o_p(n^{-1/2})\\
    &= \frac{1}{n} \sum_{i=1}^{n}  \left(\bar X_i - \bar X\right)^2 - \frac{1}{n} \sum_{i=1}^{n} \bar\epsilon_i^2 + o_p(n^{-1/2}).
  \end{align*}
  From \Cref{lemma:V_nonnormal}, 
  \begin{align*}
    T_{2,n} &= V + O_p(n^{-1/2}).
  \end{align*}
  For the numerator terms, by properties from \Cref{lemma:properties_nonnormal},
  \begin{align*}
    \sqrt{n} T_{1,n} &= \frac{1}{\sqrt{n}}\sum_{i=1}^{n} \left(\theta_i - \E \left[\theta_i\right]\right)^2 + \frac{1}{\sqrt{n}} \sum_{i=1}^{n} \bar\epsilon_i \left(\theta_i - \E \left[\theta_i\right]\right) \\
    &\quad - \sqrt{n} \left(\bar\theta - \E \left[\theta_i\right]\right)^2 - \sqrt{n} \left(\bar\theta - \E \left[\theta_i\right]\right)\bar\epsilon\\
    &= \frac{1}{\sqrt{n}}\sum_{i=1}^{n} \left(\theta_i - \E \left[\theta_i\right]\right)^2 + o_p(1),
  \end{align*}
  \begin{align*}
    \sqrt{n} T_{3,n} &= \frac{1}{\sqrt{n}} \sum_{i=1}^{n} \left(\theta_i - \E\left[\theta_i\right]\right) u_i + \frac{1}{\sqrt{n}} \sum_{i=1}^{n} \bar\epsilon_i u_i\\
    &\quad - \sqrt{n} \left(\bar\theta - \E \left[\theta_i\right]\right) \bar u - \sqrt{n} \bar \epsilon \bar u \\
    &= \frac{1}{\sqrt{n}} \sum_{i=1}^{n} \left(\theta_i - \E\left[\theta_i\right]\right) u_i + o_p(1).
  \end{align*}
  Combined with the proof of \Cref{lemma:V_nonnormal}, we have
  \begin{align*}
    \sqrt{n}T_{2,n} &= \frac{1}{\sqrt{n}} \sum_{i=1}^{n}\left[ \left(\theta_i - \E\left[\theta_i\right]\right)^2 + \bar \epsilon_i^2\right] - \sqrt{n}\left( \bar\theta - \E\left[\theta_i\right] \right)^2 - \sqrt{n}\bar\epsilon^2 + \frac{2}{\sqrt{n}} \sum_{i=1}^{n} \left( \theta_i - \E \left[\theta_i\right]\right)\bar\epsilon_i \\
    &\quad - 2 \sqrt{n}\left(\bar\theta-\E \left[\theta_i\right]\right)\bar\epsilon  - \frac{1}{\sqrt{n}}  \sum_{i=1}^{n} \bar\epsilon_i^2 \\
    &= \frac{1}{\sqrt{n}} \sum_{i=1}^{n} \left(\theta_i - \E\left[\theta_i\right]\right)^2  + o_p(1).
  \end{align*}
  Therefore, the numerator is
  \begin{align*}
    &\quad \beta\sqrt{n} T_{1,n} - \beta\sqrt{n} T_{2,n} + \sqrt{n} T_{3,n} \\
    &=  \frac{1}{\sqrt{n}} \sum_{i=1}^{n} \left(\theta_i - \E\left[\theta_i\right]\right) u_i + o_p(1).
  \end{align*}
Then applying the central limit theorem combined with the denominator, by Slutsky's theorem,
\begin{align*}
\sqrt{n} \left( \hat\beta - \beta \right)= \frac{\beta \sqrt{n} T_{1,n} - \beta \sqrt{n} T_{2,n} + \sqrt{n} T_{3,n}}{T_{2,n}}\to_d N\left(0, \frac{\E \left[ u_i^2 \left(\theta_i-\E \left[\theta_i\right]\right)^2 \right]}{V^2}\right).
\end{align*}

\end{proof}

\section{Lemmas and Proofs for Inference}
Firstly, for simplicity we abbreviate notations and denote $\hat\beta_{c,\HE}$ as $\hat\beta_c$, $\hat\beta_\HE$ as $\hat\beta$, $\hat\theta_{i,c,\HE}$ as $\hat\theta_{i,c}$, $\hat\theta_{i,\HE}$ as $\hat\theta_i$ and $\Var\left(\theta_i\right)$ as $V$.

\begin{lemma}
  \label{lemma:inference_c}
  Under \Cref{assum:nonnormal} and \Cref{assum:inference},
  \begin{align*}
    \frac{1}{n} \sum_{i-1}^{n} \left(\hat\theta_{i,c} - \bar{\hat\theta}_c\right)^2\hat u_{i,c}^2 \to _p \E \left[ u_i^2 \left(\theta_i - \E\left[\theta_i\right]\right)^2  \right].
  \end{align*}
\end{lemma}
\begin{proof}[Proof of \Cref{lemma:inference_c}]
  Combining \Cref{lemma:theta_4}, \ref{lemma:theta_3}, \ref{lemma:theta_2}, \ref{lemma:theta_u}, we have
  \begin{align*}
    &\quad \frac{1}{n} \sum_{i=1}^{n} \left(\hat\theta_{i,c} - \bar{\hat\theta}_c\right)^2\hat u_{i,c}^2\\
    &= \frac{1}{n} \sum_{i=1}^{n} \left(\hat \theta_{i,c} - \bar{\hat\theta}_c\right)^2 \left(Y_i - \bar Y - \hat\beta_c\left(\hat\theta_{i,c} - \bar{\hat\theta}_c\right) \right)^2\\
    &= \frac{1}{n} \sum_{i=1}^{n} \left(\hat \theta_{i,c} - \bar{\hat\theta}_c\right)^2 \left(Y_i - \bar Y\right)^2 - 2\hat\beta_c \frac{1}{n} \sum_{i=1}^{n} \left(\hat \theta_{i,c} - \bar{\hat\theta}_c\right)^3 \left(Y_i - \bar Y\right)\\
    &\quad + \hat\beta_c^2 \frac{1}{n} \sum_{i=1}^{n} \left(\hat \theta_{i,c} - \bar{\hat\theta}_c\right)^4.\\
    &= \beta^2 \frac{1}{n} \sum_{i=1}^{n} \left(\hat \theta_{i,c} - \bar{\hat\theta}_c\right)^2\left(\theta_i - \bar\theta\right)^2 - 2\beta \frac{1}{n} \sum_{i=1}^{n} \left(\hat \theta_{i,c} - \bar{\hat\theta}_c\right)^2 \left(\theta_i - \bar\theta\right)\left(u_i - \bar u\right)\\
    &\quad + \frac{1}{n} \sum_{i=1}^{n} \left(\hat \theta_{i,c} - \bar{\hat\theta}_c\right)^2 \left(u_i - \bar u\right)^2\\
    &\quad - 2\hat\beta_{c}\beta \frac{1}{n} \sum_{i=1}^{n} \left(\hat \theta_{i,c} - \bar{\hat\theta}_c\right)^3 \left(\theta_i - \bar\theta\right) + 2\hat\beta_{c} \frac{1}{n} \sum_{i=1}^{n} \left(\hat \theta_{i,c} - \bar{\hat\theta}_c\right)^3 \left(u_i - \bar u\right)\\
    &\quad + \hat\beta_c^2 \frac{1}{n} \sum_{i=1}^{n} \left(\hat \theta_{i,c} - \bar{\hat\theta}_c\right)^4\\
    &= \beta^2 \E \left[ \left(\theta_i - \E \left[\theta_i\right]\right)^4 \right] - 2\beta \E\left[\left(\theta_i - \E\left[\theta_i\right]\right)^3 u_i \right] + \E \left[\left(\theta_i - \E \left[\theta_i\right]\right)^2 u_i^2\right]\\
    &\quad - 2\beta^2 \E \left[\left(\theta_i - \E \left[\theta_i\right]\right)^4\right] + 2\beta  \E\left[\left(\theta_i - \E\left[\theta_i\right]\right)^3 u_i \right]+ \beta^2 \E \left[ \left(\theta_i - \E \left[\theta_i\right]\right)^4 \right] + o_p(1)\\
    &= \E \left[ u_i^2 \left(\theta_i - \E\left[\theta_i\right]\right)^2  \right] + o_p(1).
  \end{align*}
  
\end{proof}

\begin{lemma}
  \label{lemma:theta_4}
  Under  \Cref{assum:nonnormal} and \Cref{assum:inference},
  \begin{align*}
    \frac{1}{n} \sum_{i=1}^{n} \left(\hat \theta_{i,c} - \bar{\hat\theta}_c\right)^4 \to_p \E \left[ \left(\theta_i - \E\left[\theta_i\right]\right)^4\right].
  \end{align*}
\end{lemma}
\begin{proof}[Proof of \Cref{lemma:theta_4}]
  For any $s  >0$ and integers $k_1\geq0$, $0\leq k_2, k_3 \leq 4$, 
  \begin{align*}
    \Pr \left( \left| \frac{1}{n}\sum_{i=1}^{n} c_i^{k_1} \bar X_i ^{k_2} - \E \left[ c_i^{k_1} \bar X_i^{k_2} \right] \right| >s \right) \leq \frac{\E \left[ c_i^{2k_1} \bar X_i^{2k_2} \right]}{ns^2} \leq \frac{\E \left[ \bar X_i^{2k_2} \right]}{ns^2} \to 0.
  \end{align*}
   \begin{align*}
    &\E \left[ c_i^{k_1} \theta_i ^{k_3} \right] = \E \left[ c_i^{k_1} \right] \E \left[ \theta_i ^{k_3} \right] \to \E \left[ \theta_i^{k_3}\right]. &\text{(By \Cref{lemma:properties_nonnormal} \ref{item:c})}
   \end{align*}
   For $k_4 = 4$, by independence and \eqref{eq:var_nonnormal},
   \begin{align*}
    \left| \E \left[ c_i^{k_1} \theta_i^{k_3}\bar\epsilon_i^{k_4} \right]  \right| \leq \E \left[\left|\theta_i^{k_3} \right| \right]\E \left[ \bar\epsilon_i^{k_4} \right] &= \E \left[\left|\theta_i^{k_3} \right| \right]\E \left[ \frac{1}{J_i^3}\E\left[\epsilon_{i,j}^4 \mid J_i,\sigma_i^2 \right] + \frac{3\left(J_i -1\right)}{J_i^3}\sigma_i^2\right]\\
    &\leq \E \left[\left|\theta_i^{k_3} \right| \right]\E \left[ \frac{1}{J_i^3}K\sigma_i^4 + \frac{3\left(J_i -1\right)}{J_i^3}\sigma_i^2\right]\\
    &= \E \left[\left|\theta_i^{k_3} \right| \right]\left( \E \left[ \frac{1}{J_i^3}\right]K\E\left[\sigma_i^4\right] + \E\left[\frac{3\left(J_i -1\right)}{J_i^3}\right]\E\left[\sigma_i^2\right]  \right)\to 0.
   \end{align*}
   Then by Jensen's inequality, for $1\leq k_4 \leq 4$,
   \begin{align*}
    \left| \E \left[ c_i^{k_1} \theta_i^{k_3}\bar\epsilon_i^{k_4} \right]  \right| \to 0.
   \end{align*}
   Therefore for $0\leq k_2\leq 4$,
   \begin{align*}
    &\quad \E \left[ c_i^{k_1} \bar X_i^{k_2} \right] \\
    &= \sum_{k_4=0}^{k_2}  \binom{k_2}{k_4}\E \left[c_i^{k_1}  \theta_i^{k_2 - k_4} \bar\epsilon_i^{k_4}  \right]\\
    &\to \E \left[ \theta_i^{k_2} \right].
   \end{align*}
   So for $0\leq k_2\leq 4$,
   \begin{align*}
    \frac{1}{n} \sum_{i=1}^{n} c_i^{k_1} \bar X_i ^{k_2} \to_p \E \left[ \theta_i^{k_2} \right].
   \end{align*}
   Therefore for $1 \leq k \leq 4$, $k_2 \geq0$,
   \begin{align*}
    &\quad \frac{1}{n}\sum_{i=1}^{n} \left( c_i \bar X_i - \frac{1}{n}\sum_{i=1}^{n} c_i \bar X_i \right)^{k}\\
    &= \sum_{k_1=0}^{k}\binom{k}{k_1} \left( - \frac{1}{n} \sum_{i=1}^{n}c_i \bar X_i \right)^{k-k_1} \frac{1}{n}\sum_{i=1}^{n} c_i^{k_1} \bar X_i^{k_1}\\
    &= \sum_{k_1 = 0}^{k} \binom{k}{k_1} \left(- \E \left[\theta_i\right] \right)^{k-k_1} \E \left[ \theta_i^{k_1} \right] + o_p(1)\\
    &= \E \left[ \left(\theta_i - \E \left[\theta_i\right]\right)^k \right] + o_p(1).
   \end{align*}
   \begin{align*}
    &\quad \frac{1}{n}\sum_{i=1}^{n} \left( c_i \bar X_i - \frac{1}{n}\sum_{i=1}^{n} c_i \bar X_i \right)^{k} c_i^{k_2}\\
    &= \sum_{k_1=0}^{k}\binom{k}{k_1} \left( - \frac{1}{n} \sum_{i=1}^{n}c_i \bar X_i \right)^{k-k_1} \frac{1}{n}\sum_{i=1}^{n} c_i^{k_1+k_2} \bar X_i^{k_1}\\
    &= \E \left[ \left(\theta_i - \E \left[\theta_i\right]\right)^k \right] + o_p(1).
   \end{align*}
   \begin{align*}
    &\quad \frac{1}{n}\sum_{i=1}^{n} \left( c_i \bar X_i - \frac{1}{n}\sum_{i=1}^{n} c_i \bar X_i \right)^{k} \left(c_i - \bar c\right)^{k_2}\\
    &= \sum_{k_1=0}^{k}\binom{k}{k_1} \left( - \frac{1}{n} \sum_{i=1}^{n}c_i \bar X_i \right)^{k-k_1} \frac{1}{n}\sum_{i=1}^{n} c_i^{k_1+k_2} \bar X_i^{k_1}\left(c_i - \bar c\right)^{k_2}\\
    &=  \sum_{k_1=0}^{k}\binom{k}{k_1} \left( - \frac{1}{n} \sum_{i=1}^{n}c_i \bar X_i \right)^{k-k_1} \sum_{t=0}^{k_2}\binom{k_2}{t} \left( -\bar c \right)^{k_2 - t}\frac{1}{n}\sum_{i=1}^{n} c_i^{k_1+k_2+t} \bar X_i^{k_1}\\
    &=  \sum_{k_1=0}^{k}\binom{k}{k_1} \left( - \E \left[\theta_i\right] \right)^{k-k_1} \sum_{t=0}^{k_2}\binom{k_2}{t} \left( -1 \right)^{k_2 - t}\E \left[\theta_i^{k_1}\right] + o_p(1)\\
    &= o_p(1).
   \end{align*}
   The above also applies to $k=0$.

   Therefore,
   \begin{align*}
    &\quad \frac{1}{n} \sum_{i=1}^{n} \left(\hat \theta_{i,c} - \bar{\hat\theta}_c\right)^4\\
    &= \frac{1}{n} \sum_{i=1}^{n} \left( c_i \left( \bar X_i - \bar X \right) - \frac{1}{n} \sum_{i=1}^{n} \left( c_i \left(\bar X_i - \bar X\right) \right) \right)^4\\
    &= \frac{1}{n} \sum_{i=1}^{n} \left( c_i \bar X_i  - \frac{1}{n} \sum_{i=1}^{n} c_i \bar X_i - \bar X \left(c_i - \bar c  \right) \right)^4\\
    &= \frac{1}{n} \sum_{i=1}^{n} \left(c_i \bar X_i - \frac{1}{n}\sum_{i=1}^{n} c_i \bar X_i \right)^4 - 4 \bar X \frac{1}{n} \sum_{i=1}^{n} \left(c_i \bar X_i - \frac{1}{n}\sum_{i=1}^{n} c_i \bar X_i \right)^3 \left(c_i - \bar c \right)\\
    & \quad + 6 \bar X^2 \frac{1}{n} \sum_{i=1}^{n} \left(c_i \bar X_i - \frac{1}{n}\sum_{i=1}^{n} c_i \bar X_i \right)^2 \left(c_i - \bar c \right)^2 - 4 \bar X^3 \frac{1}{n} \sum_{i=1}^{n} \left(c_i \bar X_i - \frac{1}{n}\sum_{i=1}^{n} c_i \bar X_i \right) \left(c_i - \bar c \right)^3\\
    &\quad + \bar X^4 \frac{1}{n}\sum_{i=1}^{n} \left(c_i - \bar c \right)^4. \\
    &= \E \left[ \left(\theta_i - \E \left[\theta_i\right]\right)^4 \right] + o_p(1).
  \end{align*}
\end{proof}

\begin{lemma}
  \label{lemma:theta_3}
  Under \Cref{assum:nonnormal} and \Cref{assum:inference},
  \begin{align*}
    \frac{1}{n} \sum_{i=1}^{n} \left(\hat\theta_{i,c} - \bar{\hat\theta}_c\right)^{k} \left(\theta_i - \bar\theta\right)^{4-k} \to_p \E \left[ \left(\theta_i - \E\left[\theta_i\right]\right)^4 \right].
  \end{align*}
\end{lemma}
\begin{proof}[Proof of \Cref{lemma:theta_3}]
  For any $s>0$ and integers $k_1 \geq0$, $0 \leq k_2 \leq 4$, $k_2 + k_3 \leq 4$, by Chebyshev's inequality,
  \begin{align*}
    \Pr \left( \left| \frac{1}{n} \sum_{i=1}^{n} c_i^{k_1} \bar X_i^{k_2} \theta_i^{k_3} - \E \left[ c_i^{k_1} \bar X_i^{k_2} \theta_i^{k_3} \right]  \right| >s\right) \leq \frac{\E \left[ c_i^{2k_1} \bar X_i^{2k_2} \theta_i^{2k_3} \right]}{ns^2} \leq \frac{\E \left[ \bar X_i^{2k_2} \theta_i^{2k_3} \right]}{ns^2} \to 0.
  \end{align*}
  For $0 \leq k_2 \leq 4$, $k_2 + k_3 \leq 4$,
  \begin{align*}
    &\quad \E \left[ c_i^{k_1} \bar X_i^{k_2} \theta_i^{k_3} \right]\\
    &= \sum_{k_4 = 0}^{k_2} \binom{k_2}{k_4} \E \left[ c_i^{k_1}  \theta_i^{k_2 - k_4 + k_3} \bar \epsilon_i^{k_4} \right]\\
    &\to \E \left[ \theta_i^{k_2 + k_3} \right]. &\text{(By \Cref{lemma:properties_nonnormal} \ref{item:c})}
  \end{align*}
  Therefore, similar to the proof of \Cref{lemma:theta_4},
  \begin{align*}
    &\quad \frac{1}{n} \sum_{i=1}^{n} \left(\hat\theta_{i,c} - \bar{\hat\theta}_c\right)^{k} \left(\theta_i - \bar\theta\right)^{4-k}\\
    &= \sum_{k_1=0}^{4-k} \binom{4-k}{k_1} \bar\theta^{4-k-k_1} \frac{1}{n} \sum_{i=1}^{n} \left(\hat\theta_{i,c} - \bar{\hat\theta}_c\right)^{k} \theta_i^{k_1}\\
    &= \sum_{k_1=0}^{4-k} \binom{4-k}{k_1} \left( \E \left[\theta_i\right]\right)^{4-k-k_1} \E \left[ \left(\theta_i -\E \left[\theta_i\right]\right)^k \theta_i^{k_1}  \right]  + o_p(1)\\
    &= \E \left[ \left(\theta_i - \E \left[\theta_i\right]\right)^4 \right] + o_p(1).
  \end{align*}
\end{proof}

\begin{lemma}
  \label{lemma:theta_2}
  Under \Cref{assum:nonnormal} and \Cref{assum:inference},
  \begin{align*}
    \frac{1}{n} \sum_{i=1}^{n} \left(\hat\theta_{i,c} - \bar{\hat\theta}_c\right)^3 \left(u_i - \bar u \right)\to_p \E \left[ \left(\theta_i - \E\left[\theta_i\right]\right)^3 u_i  \right],
  \end{align*}
  \begin{align*}
    \frac{1}{n} \sum_{i=1}^{n} \left(\hat\theta_{i,c} - \bar{\hat\theta}_c\right)^2 \left(\theta_i -\bar \theta\right)\left(u_i - \bar u \right) \to_p \E\left[ \left(\theta_i - \E\left[\theta_i\right]\right)^3 u_i \right].
  \end{align*}
\end{lemma}

\begin{proof}
  For any $s >0$ and integers $k_1 \geq0$, $0 \leq k_2 \leq 3$, by Chebyshev's inequality,
  \begin{align}
    \Pr \left( \left| \frac{1}{n} \sum_{i=1}^{n} c_i^{k_1} \bar X_i^{k_2}u_i  - \E\left[ c_i^{k_1} \bar X_i^{k_2} u_i\right] \right| >s \right) \leq \frac{\E \left[\bar X_i^{2k_2} u_i^2 \right]}{ns^2} \to 0.
  \end{align}
  For $0\leq k_3 \leq 3$,
  \begin{align*}
   &\E \left[ c_i^{k_1}\theta_i^{k_3} u_i\right] = \E \left[c_i^{k_1}\right] \E \left[ \theta_i^{k_3}u_i \right] \to \E \left[ \theta_i^{k_3} u_i \right]. &\text{(By \Cref{lemma:properties_nonnormal} \ref{item:c})}
  \end{align*}
  For $k_4 = 4$, by independence and \eqref{eq:var_nonnormal},
  \begin{align*}
    \left| \E \left[ c_i^{k_1} \theta_i^{k_3} \bar\epsilon_i^{k_4} u_i^2 \right] \right|\leq \E \left[ \left| \theta_i^{k_3} u_i^2 \right| \right] \E \left[ \bar\epsilon_i^{k_4} \right] &= \E \left[\left|\theta_i^{k_3} u_i^2\right| \right]\E \left[ \frac{1}{J_i^3}\E\left[\epsilon_{i,j}^4 \mid J_i,\sigma_i^2 \right] + \frac{3\left(J_i -1\right)}{J_i^3}\sigma_i^2\right]\\
    &\leq \E \left[\left|\theta_i^{k_3} u_i^2\right| \right]\E \left[ \frac{1}{J_i^3}K\sigma_i^4 + \frac{3\left(J_i -1\right)}{J_i^3}\sigma_i^2\right]\\
    &= \E \left[\left|\theta_i^{k_3}u_i^2 \right| \right]\left( \E \left[ \frac{1}{J_i^3}\right]K\E\left[\sigma_i^4\right] + \E\left[\frac{3\left(J_i -1\right)}{J_i^3}\right]\E\left[\sigma_i^2\right]  \right)\to 0.
  \end{align*}
  Then by Jensen's inequality, for $1\leq k_4 \leq 4$,
  \begin{align*}
    \left| \E \left[ c_i^{k_1} \theta_i^{k_3} \bar\epsilon_i^{k_4} u_i^2 \right] \right| \to 0.
  \end{align*}
  Therefore for $1\leq k_2\leq 4$,
  \begin{align*}
    &\quad \E \left[ c_i^{k_1} \bar X_i^{k_2} u_i \right] \\
    &= \sum_{k_4=0}^{k_2}  \binom{k_2}{k_4}\E \left[c_i^{k_1}  \theta_i^{k_2 - k_4} \bar\epsilon_i^{k_4} u_i  \right]\\
    &\to \E \left[ \theta_i^{k_2} u_i \right].
  \end{align*}
  Therefore,
  \begin{align*}
    &\quad \frac{1}{n} \sum_{i=1}^{n} \left(\hat\theta_{i,c} - \bar{\hat\theta}_c\right)^3 \left(u_i - \bar u\right)\\
    &=\frac{1}{n} \sum_{i=1}^{n} \left(\hat\theta_{i,c} - \bar{\hat\theta}_c\right)^3 u_i + o_p(1)\\
    &= \E \left[ \left(\theta_i - \E \left[\theta_i\right]\right)^3 u_i \right] + o_p(1).
  \end{align*}
  The second result is shown by combining the proof procedures of \Cref{lemma:theta_3} and the above.
\end{proof}

\begin{lemma}
  \label{lemma:theta_u}
  Under \Cref{assum:nonnormal} and \Cref{assum:inference},
  \begin{align*}
    \frac{1}{n} \sum_{i=1}^{n} \left(\hat\theta_{i,c} - \bar{\hat\theta}_c\right)^2 \left(u_i - \bar u \right)^2\to_p \E \left[ u_i^2 \left(\theta_i - \E \left[\theta_i\right]\right)^2 \right]
  \end{align*}
\end{lemma}

\begin{proof}
  For any $s >0$ and integers $k_1 \geq0$, $0 \leq k_2 \leq 3$, by Chebyshev's inequality,
  \begin{align*}
    \Pr \left( \left| \frac{1}{n} \sum_{i=1}^{n} c_i^{k_1} \bar X_i^{k_2}u_i^2 - \E \left[ c_i^{k_1}\bar X_i^{k_2} u_i^2 \right]   \right| >s \right) \leq \frac{\E \left[\bar X_i^{2k_2} u_i^4 \right]}{ns^2} \to 0.
  \end{align*}
  For $0\leq k_3 \leq 3$,
  \begin{align*}
   &\E \left[ c_i^{k_1}\theta_i^{k_3} u_i^2\right] = \E \left[c_i^{k_1}\right] \E \left[ \theta_i^{k_3}u_i^2 \right] \to \E \left[ \theta_i^{k_3} u_i^2 \right]. &&\text{(By \Cref{lemma:properties_nonnormal} \ref{item:c})}
  \end{align*}
  For $k_4 = 4$, by independence and \eqref{eq:var_nonnormal}, 
  \begin{align*}
    \left| \E \left[ c_i^{k_1} \theta_i^{k_3} \bar\epsilon_i^{k_4} u_i^2 \right] \right|\leq \E \left[ \left| \theta_i^{k_3} u_i^2 \right| \right] \E \left[ \bar\epsilon_i^{k_4} \right] 
    &= \E \left[\left|\theta_i^{k_3} u_i^2\right| \right]\E \left[ \frac{1}{J_i^3}\E\left[\epsilon_{i,j}^4 \mid J_i,\sigma_i^2 \right] + \frac{3\left(J_i -1\right)}{J_i^3}\sigma_i^2\right]\\
    &\leq \E \left[\left|\theta_i^{k_3} u_i^2\right| \right]\E \left[ \frac{1}{J_i^3}K\sigma_i^4 + \frac{3\left(J_i -1\right)}{J_i^3}\sigma_i^2\right]\\
    &= \E \left[\left|\theta_i^{k_3}u_i^2 \right| \right]\left( \E \left[ \frac{1}{J_i^3}\right]K\E\left[\sigma_i^4\right] + \E\left[\frac{3\left(J_i -1\right)}{J_i^3}\right]\E\left[\sigma_i^2\right]  \right)\to 0.
  \end{align*}
  Then by Jensen's inequality, for $1\leq k_4 \leq 4$,
  \begin{align*}
    \left| \E \left[ c_i^{k_1} \theta_i^{k_3} \bar\epsilon_i^{k_4} u_i^2 \right] \right| \to 0.
  \end{align*}
  Therefore for $1\leq k_2\leq 4$,
  \begin{align*}
    &\quad \E \left[ c_i^{k_1} \bar X_i^{k_2} u_i^2 \right] \\
    &= \sum_{k_4=0}^{k_2}  \binom{k_2}{k_4}\E \left[c_i^{k_1}  \theta_i^{k_2 - k_4} \bar\epsilon_i^{k_4} u_i^2  \right]\\
    &\to \E \left[ \theta_i^{k_2} u_i^2 \right].
  \end{align*}
  Therefore,
  \begin{align*}
    &\quad \frac{1}{n} \sum_{i=1}^{n} \left(\hat\theta_{i,c} - \bar{\hat\theta}_c\right)^2 \left(u_i - \bar u\right)^2\\
    &=\frac{1}{n} \sum_{i=1}^{n} \left(\hat\theta_{i,c} - \bar{\hat\theta}_c\right)^2 u_i^2 + o_p(1)\\
    &= \E \left[ \left(\theta_i - \E \left[\theta_i\right]\right)^2 u_i^2 \right] + o_p(1).
  \end{align*}

\end{proof}

\section{Lemmas and Proofs for Correlated $J_i$ and $\sigma_i^2$}
\label{sec:appendix-lemmas-correlated_J}

Firstly, for simplicity we abbreviate notations and denote $\hat\beta_{c,\HE}$ as $\hat\beta_c$, $\hat\beta_\HE$ as $\hat\beta$, $\hat\theta_{i,c,\HE}$ as $\hat\theta_{i,c}$, $\hat\theta_{i,\HE}$ as $\hat\theta_i$ and $\Var\left(\theta_i\right)$ as $V$. We keep the subscripts for those related to HO.

We next define the shrinkage weight
\begin{align*}
  c_{i,\text{HO}} \coloneqq \frac{V}{\frac{\E\left[\sigma_i^2\right]}{J_i} + V},
\end{align*}
where $V = \Var\left(\theta_i\right)$.
And without affecting the results, we define
\begin{align*}
  \hat \theta_{i, \text{HO}} \coloneqq \frac{\hat\sigma_\theta^2}{\frac{1}{J_i}\hat\sigma^2 + \hat\sigma_\theta^2}\bar X_i + \frac{\frac{1}{J_i}\hat\sigma^2}{\frac{1}{J_i}\hat\sigma^2 + \hat\sigma_\theta^2}\bar X.
\end{align*}

\begin{assumption}
  \label{assum:corr_J}
  \begin{enumerate}
    \item $J_i$ is independent of $\theta_i$ and $u_i$. $J_i \ind \epsilon_{i,j} \mid \sigma_i^2$ . $J_i \geq 3, a.s$. 
    \item $\E \left[u_i\right] = 0$, $\E(u_i  \theta_i) = 0$. $\E\left[Y_i^4\right] < \infty$.
    \item $\E\left[\epsilon_{i,j} \mid  \sigma^2_i\right] =0$, $\E\left[ \epsilon_{i,j}^2 \mid \sigma_i^2  \right] = \sigma_i^2$, $\E\left[\abs{\epsilon_{i,j}}^L \mid \sigma_i^2\right] \leq K\sigma_i^L, L\geq3$. Also, $\epsilon_{i,j} \ind \theta_i$.
    \item $ u_i \ind \epsilon_{i,j} \big| \theta_i$. 
    \item $\E\left[\theta_i^4\right] < \infty$. 
    \item $\E\left[\sigma^{32}_i\right] < \infty$. 
  \end{enumerate}
\end{assumption}

\begin{assumption}
  \label{assum:corr_J_large_J}
  \begin{align*}
    n^{\frac{3}{2}} \E \left[ \frac{1}{J_i^3} \right] \to \kappa^3
  \end{align*}
\end{assumption}

\begin{remark}
  \label{remark:corr_J}
  Since we assume $J_i \geq 3$, \Cref{assum:corr_J_large_J} also implies that
  \begin{align*}
    &n\E\left[ \frac{1}{J_i^2} \right] \leq  \left(n^{\frac{3}{2}} \E\left[\frac{1}{J_i^3} \right] \right)^{2/3} = O(1). &\text{(Jensen's inequality)}\\
  \end{align*}
\end{remark}

\begin{lemma}
  \label{lemma:properties_J}
  Under \Cref{assum:corr_J}, \Cref{assum:corr_J_large_J}, we have:
  \begin{enumerate}
    \item Properties of $\bar \epsilon_i$, $\bar X_i$, $\hat\sigma_i^2$, $c_i$, $c_{i,\text{HO}}$:
    
    For any integer $k_1 \in \left\{1,2\right\}$, $k_2 \geq 1$,  we have
    \begin{enumerate}
      \item $\bar\epsilon_i = O_p(n^{-1/4})$, and $\E\left[\left|\bar\epsilon_i \right|^{k_1} \right] \to 0$. \label{item:epsilon}
      \item $\bar X_i = \theta_i + O_p(n^{-1/4})$, and $\E \left[ \left| \bar X_i \right|^{k_1} \right] \to \E \left[ \left| \theta_i \right|^{k_1} \right]$.
      \item $\hat\sigma_i^2 = \sigma_i^2 + O_p(n^{-1/4})$, and $\E \left[ \left| \hat\sigma_i^2 - \sigma_i^2 \right|^{k_1} \right] \to 0$.
      \item $\E \left[ \frac{1}{J_i} \hat\sigma_i^2 \right] \to 0$, and $\E \left[ \left( \frac{1}{J_i} \hat\sigma_i^2 \right)^2 \right] \to 0$.
      \item $\E\left[ \left(\frac{1}{J_i} \hat\sigma_i^2 \right)^4\right] = o(n^{-1/2})$.  \label{item:angelova}
      \item $c_i^{k_2} = 1 + O_p(n^{-1/2})$, and $\E \left[ c_i^{k_2} \right] = 1 + O(n^{-1/2})$. \label{item:c}
      \item $c_{i,\HO}^{k_2} = 1 + O_p(n^{-1/2})$, and $\E \left[c_{i,\HO}^{k_2} \right] = 1 + O(n^{-1/2})$.
    \end{enumerate}

    \item Properties of sample moments of $\bar\epsilon_i$, $c_i$, $\theta_i$, $c_{i,\HO}$: \label{item:moment}
    
    For any integer $k_1\geq 0$\footnote{For statements like this, if $k=0$ is on the exponent, it means that the term is $1$.}, $0 \leq k_2 \leq 2$, and $k_3, k_4 \in \left\{0,1 \right\}$, we have 
    \begin{align*}
      &\frac{1}{n}\sum_{i=1}^{n} c_i^{k_1} \left(\theta_i - k_3 \E \left[\theta_i\right]\right)^{k_2} = \E \left[\left(\theta_i - k_3 \E \left[\theta_i\right]\right)^{k_2}\right] +O_p(n^{-1/2}), \\
      &\frac{1}{n}\sum_{i=1}^{n} c_i^{k_1} \left( \theta_i - \E \left[\theta_i\right]\right)^{k_4}\bar\epsilon_i = o_p(n^{-1/2}), \\
      &\frac{1}{n}\sum_{i=1}^{n} c_i^{k_1} \left( \theta_i- k_3 \E \left[\theta_i\right]\right)^{k_2}\bar\epsilon_i^2 = O_p(n^{-1/2}).
    \end{align*}
    And similar results hold for $c_{i,\text{HO}}$.

    \item Properties of sample means of $\theta_i$, $\bar X_i$, $\hat\sigma_i^2$:
    \begin{enumerate}
      \item $\bar\theta = \E \left[\theta_i\right] + O_p(n^{-1/2})$.
      \item $\bar X = \E \left[\theta_i\right] + O_p(n^{-1/2})$.
      \item $\frac{1}{n}\sum_{i=1}^{n}\frac{1}{J_i}\hat\sigma_i^2 = O_p(n^{-1/2})$
    \end{enumerate}

    \item Properties of sample moments of $\epsilon_i$, $c_i$, $u_i$,$c_{i,\text{HO}}$:
    
    For any integer $k\geq 0$, we have
    \begin{enumerate}
      \item $\frac{1}{n} \sum_{i=1}^{n} c_i^{k} u_i = O_p(n^{-1/2})$.
      \item $\frac{1}{n}\sum_{i=1}^{n} c_i^{k} u_i \bar \epsilon_i = o_p(n^{-1/2})$.
    \end{enumerate}
    And similar results hold for $c_{i,\HO}$.
  \end{enumerate}
\end{lemma}

\begin{proof}[Proof of \Cref{lemma:properties_J}]
  Below we take any $s >0$, 
  \begin{enumerate}
    \item Properties of $\bar \epsilon_i$, $\bar X_i$, $\hat\sigma_i^2$, $c_i$:
    \begin{enumerate}
      \item By Markov's inequality,
      \begin{align*}
        \Pr \left( n^{1/4}\left|\bar \epsilon_i \right| > s  \right) \leq \frac{\sqrt{n} \E \left[\bar\epsilon_i^2\right]}{s^2} = \frac{\sqrt{n}\E \left[\frac{1}{J_i}\sigma_i^2\right]}{s^2} \leq \frac{\sqrt{n \E \left[\frac{1}{J_i^2}\right]\E \left[\sigma_i^4\right]}}{s^2}.
      \end{align*}
      Since $n \E \left[\frac{1}{J_i^2}\right] \to \kappa^2$, we have $\bar\epsilon_i = O_p(n^{-1/4})$.

      Since 
      \begin{align*}
        \E \left[\bar\epsilon_i^2\right] = \E \left[\frac{1}{J_i}\sigma_i^2\right] \leq \sqrt{ \E \left[ \frac{1}{J_i^2} \right] \E \left[\sigma_i^4\right] } \to 0,
      \end{align*}
      combined with the Cauchy-Schwarz inequality, we have the rest of the results.
      \item This follows from \ref{item:epsilon}.
      \item By Markov's inequality and \eqref{eq:var_nonnormal},
      \begin{align*}
        \Pr \left( n^{1/4}\left|\hat\sigma_i^2 - \sigma_i^2 \right| > s   \right) \leq \frac{\sqrt{n}\E\left[\left(\hat\sigma_i^2 - \sigma_i^2\right)^2\right]}{s^2} \leq \frac{2 \sqrt{n}\E \left[\frac{K}{J_i } \sigma_i^4\right]}{s^2} \leq \frac{2K \sqrt{n \E \left[\frac{1}{\left(J_i \right)^2}\right] \E \left[\sigma_i^8\right]}}{s^2}.
      \end{align*}
      Due to \Cref{assum:corr_J_large_J}, we have $\hat\sigma_i^2 = \sigma_i^2 + O_p(n^{-1/4})$.
      
      Also,
      \begin{align*}
        \E \left[ \left(\hat\sigma_i^2 - \sigma_i^2\right)^2 \right] \leq 2 \E\left[\frac{K}{J_i}\sigma_i^4\right] \leq 2K \sqrt{\E\left[\frac{1}{\left(J_i \right)^2}\right]\E\left[\sigma_i^8\right]} \to 0,
      \end{align*}
      combined with the Cauchy-Schwarz inequality, we have the rest of the results.
      \item We have
      \begin{align*}
        0 \leq \E \left[ \frac{1}{J_i} \hat\sigma_i^2 \right] = \E \left[ \frac{1}{J_i}  \sigma_i^2 \right] \leq \sqrt{ \E\left[\frac{1}{J_i^2}\right] \E\left[\sigma_i^4\right]}\to0,
      \end{align*}
      \begin{align*}
        \E \left[ \left( \frac{1}{J_i} \hat\sigma_i^2 \right)^2 \right] = \E \left[ \frac{1}{J_i^2} \left( \frac{K\sigma_i^4}{J_i } + \sigma_i^4  \right)  \right] \to 0.
      \end{align*}
      \item From Theorem 2 of \citet{angelova2012moments}, the fourth moment of $\hat\sigma_i^2$ is
      \begin{align*}
        \E\left[ \left(\hat\sigma_i^2\right)^4 \mid J_i,\sigma_i^2\right] &= \mu_2^4 
        + \frac{6\mu_2^2(\mu_4 - \mu_2^2)}{J_i}
        + \frac{4\mu_6\mu_2 + 3\mu_4^2 - 18\mu_4\mu_2^2 - 24\mu_3^2\mu_2 + 23\mu_2^4}{J_i^2}\\
        &\quad + \frac{\mu_8 - 4\mu_6\mu_2 - 24\mu_5\mu_3 - 3\mu_4^2 + 72\mu_4\mu_2^2 + 96\mu_3^2\mu_2 - 86\mu_2^4}{J_i^3}\\
        &\quad + \frac{4(6\mu_6\mu_2 + 6\mu_4^2 - 39\mu_4\mu_2^2 - 40\mu_3^2\mu_2 + 45\mu_2^4)}{J_i^3(J_i-1)}\\
        &\quad + \frac{4(36\mu_4\mu_2^2 - 8\mu_5\mu_3 + 52\mu_3^2\mu_2 - 61\mu_2^4)}{J_i^3(J_i-1)^2}\\
        &\quad + \frac{8(\mu_4^2 - 6\mu_4\mu_2^2 - 12\mu_3^2\mu_2 + 15\mu_2^4)}{J_i^3(J_i-1)^3},
      \end{align*}
      where $\mu_k, k=1,...,8$ are the $k$-th raw moments of $\epsilon_{i,j} \mid \sigma_i^2$. By \Cref{assum:corr_J}.3, we have $\mu_k \leq K \sigma_i^k$ for some constant $K$. By the law of iterated expectations and \Cref{assum:corr_J_large_J}, for the first term we have
      \begin{align*}
        \sqrt{n} \E\left[\frac{1}{J_i^4}\sigma_i^8 \right] \leq \sqrt{n \E\left[\frac{1}{J_i^8} \right] \E\left[\sigma_i^{16}\right]} \leq \sqrt{n \E\left[\frac{1}{J_i^3}\right] \E\left[ \sigma_i^{16}\right]}  \to 0.
      \end{align*}
      For the rest of the terms, we have the same result. Therefore,
      \begin{align*}
        \sqrt{n} \E\left[ \left(\frac{1}{J_i}\hat\sigma_i^2\right)^4 \right] = o(n^{-1/2}).
      \end{align*}
      \item By the mean value theorem,
      \begin{align*}
        c_i^{k_2} &= \left( \frac{V}{V + \frac{1}{J_i} \hat\sigma_i^2}\right)^{k_2}\\
        &= 1-k_2 \frac{1}{VJ_i} \hat\sigma_i^2 + \frac{k_2 \left(k_2 +1\right)}{2} \frac{1}{\left(1 + \omega_i\right)^{k_2 + 2}} \frac{1}{V^2 J_i^2} \left(\hat\sigma_i^2\right)^2,
      \end{align*}
      where $\omega_i$ is between $0$ and $\frac{1}{VJ_i} \hat\sigma_i^2$. Therefore,
      \begin{align*}
        \abs{\sqrt{n}\E\left[ 1- c_i^{k_2} \right]} &\leq \frac{k_2}{V}\sqrt{n}\E\left[\frac{\sigma^2_i}{J_i} \right] \\
        &\quad + \frac{k_2\left(k_2 +1\right)}{2V^2} \sqrt{n} \E\left[ \frac{1}{J_i^2} \left( \sigma_i^4 + \frac{K}{J_i}\sigma_i^4 \right) \right] &\text{(By \eqref{eq:var_nonnormal})}\\
        &\leq \frac{k_2}{V} \sqrt{\E\left[\sigma_i^4\right] n\E\left[\frac{1}{J_i^2}\right]} \\
        &\quad + \frac{k_2\left(k_2 +1\right)}{2V^2} \sqrt{\E\left[\sigma_i^8\right] n \E\left[\frac{1}{J_i^4}\right]} \\
        &\quad + \frac{k_2\left(k_2 +1\right)K}{2V^2}\sqrt{\E\left[\sigma_i^8\right] n \E\left[\frac{1}{J_i^6}\right]} & \text{(Cauchy-Schwarz)}\\
        &= O(1) + o(1) + o(1) &\text{(By \Cref{assum:corr_J_large_J})}
      \end{align*}
       Then we have the second result. By Markov's inequality, we have the first result.
       \item The proof is similar to \ref{item:c}. By the mean value theorem,
        \begin{align*}
          c_{i,\HO}^{k_2} &= \left( \frac{V}{V + \frac{\E\left[\sigma_i^2\right]}{J_i}}\right)^{k_2}\\
          &= 1-k_2 \frac{\E\left[\sigma_i^2\right]}{VJ_i} + \frac{k_2 \left(k_2 +1\right)}{2} \frac{1}{\left(1 + \omega_i\right)^{k_2 + 2}} \frac{\E\left[\sigma_i^2\right]^2}{V^2 J_i^2}.
        \end{align*}
        It is easy to see that the rest of the proof is similar to \ref{item:c}.
    \end{enumerate}
    \item 
      By Chebyshev's inequality ($k_2$ can be replaced by $k_4$),
\begin{align*}
  &\quad \Pr\left( \sqrt{n}\left| \frac{1}{n} \sum_{i=1}^{n} c_i^{k_1} \left(\theta_i - k_3 \E\left[\theta_i\right]\right)^{k_2} - \E \left[ c_i^{k_1} \left(\theta_i - k_3 \E\left[\theta_i\right]\right)^{k_2}  \right] \right| > s \right)\\
   &\leq \frac{\E \left[ c_i^{2k_1} \left(\theta_i - k_3 \E\left[\theta_i\right]\right)^{2k_2} \right]}{s^2} \leq \frac{\E \left[\left(\theta_i - k_3 \E\left[\theta_i\right]\right)^{2k_2}\right]}{s^2} 
\end{align*}
\begin{align*}
  &\quad \Pr \left( \sqrt{n}\left| \frac{1}{n} \sum_{i=1}^{n} c_i^{k_1} \left(\theta_i - k_3 \E\left[\theta_i\right]\right)^{k_2} \bar\epsilon_i - \E \left[ c_i^{k_1} \left(\theta_i - k_3 \E\left[\theta_i\right]\right)^{k_2} \bar\epsilon_i \right] \right| > s \right)\\
  & \leq \frac{\E \left[ c_i^{2k_1} \left(\theta_i - k_3 \E\left[\theta_i\right]\right)^{2k_2} \bar\epsilon_i^{2}  \right]}{s^2} \leq \frac{\E \left[ \left(\theta_i - k_3 \E\left[\theta_i\right]\right)^{2k_2} \right] \sqrt{\E\left[ \frac{1}{J_i^2} \right] \E \left[\sigma_i^4\right]}}{s^2} \to 0,
\end{align*}
\begin{align*}
  &\quad \Pr \left( \sqrt{n}\left| \frac{1}{n} \sum_{i=1}^{n} c_i^{k_1} \left(\theta_i - k_3 \E\left[\theta_i\right]\right)^{k_2} \bar\epsilon_i^2 - \E \left[ c_i^{k_1} \left(\theta_i - k_3 \E\left[\theta_i\right]\right)^{k_2} \bar\epsilon_i^2 \right] \right| > s \right) \\
  &\leq \frac{\E \left[ c_i^{2k_1} \left(\theta_i - k_3 \E\left[\theta_i\right]\right)^{2k_2} \bar\epsilon_i^{4}  \right]}{s^2}\\
   &\leq \frac{\E \left[ \left(\theta_i - k_3 \E\left[\theta_i\right]\right)^{2k_2} \right]\left(  \sqrt{\E\left[ \frac{1}{J_i^6} \right] \E \left[K^2\sigma_i^8\right]} + \sqrt{\E\left[ \frac{\left(J_i -1\right)^2}{J_i^6} \right] \E\left[\sigma_i^8\right]}\right) }{s^2} \to 0.
\end{align*}
Meanwhile,
\begin{align*}
  &\quad\sqrt{n}\E \left[ c_i^{k_1}\left(\theta_i - k_3 \E\left[\theta_i\right]\right)^{k_2} \right] = \sqrt{n}\E \left[ c_i^{k_1} \right] \E \left[ \left(\theta_i - k_3 \E\left[\theta_i\right]\right)^{k_2} \right] \\
  &=\sqrt{n} \E \left[ \left(\theta_i - k_3 \E\left[\theta_i\right]\right)^{k_2} \right] + O(1), &\text{(By \ref{item:c})}
\end{align*}

by the mean value theorem ($k_4 = 0$) and \eqref{eq:var_nonnormal},
\begin{align*}
  \abs{\sqrt{n}\E\left[c_i^{k_1} \bar\epsilon_i\right]} &\leq 0 + \frac{k_1}{V}\sqrt{n}\E\left[ \frac{1}{J_i}\hat\sigma_i^2 \bar\epsilon_i \right] \\
  &\quad + \frac{k_1\left(k_1 +1\right)}{2V^2} \sqrt{n} \E\left[ \frac{1}{J_i^2} \left(\hat\sigma_i^2\right)^2 \bar \epsilon_i \right] \\
  &\leq \frac{k_1}{V} \sqrt{ \sqrt{n} \E\left[\frac{1}{J_i^2} \left( \sigma_i^4 + \frac{K}{J_i}\sigma_i^4 \right) \right] \sqrt{n} \E\left[ \frac{1}{J_i}\sigma_i^2 \right] }\\
  &\quad + \frac{k_1\left(k_1 +1\right)}{2V^2} \sqrt{ \sqrt{n}\E\left[ \frac{1}{J_i^4}\left(\hat\sigma_i^2 \right)^4 \right] \sqrt{n}\E\left[ \frac{1}{J_i}\sigma_i^2 \right]  } &\text{(Cauchy-Schwarz)}\\
  &\leq  \frac{k_1}{V} \sqrt{ \sqrt{n \E\left[\frac{1}{J_i^4}\right] \E\left[\sigma_i^8 \right] + n \E\left[\frac{1}{J_i^6} \right] K^2 \E\left[ \sigma_i^8\right]} \sqrt{n \E\left[ \frac{1}{J_i^2}\right] \E\left[\sigma_i^4\right]} } \\
  &\quad + \frac{k_1\left(k_1 +1\right)}{2V^2} \sqrt{\sqrt{n}\E\left[ \frac{1}{J_i^4}\left(\hat\sigma_i^2 \right)^4 \right] \sqrt{n\E\left[\frac{1}{J_i^2}\right] \E\left[\sigma_i^4\right]} } \\
  &= o(1) \cdot O(1) + o(1) \cdot O(1)
\end{align*}
where the last step follows from \Cref{assum:corr_J_large_J} and \ref{item:angelova},

by the independence in \Cref{assum:nonnormal}.3 ($k_4 = 1$),
\begin{align*}
  &\E \left[ c_i^{k_1} \left(\theta_i - \E\left[\theta_i\right]\right) \bar\epsilon_i \right]  = \E \left[\E\left( c_i^{k_1}\left(\theta_i -\E\left[\theta_i\right]\right)\bar{\epsilon_i}  \bigg| \epsilon \right) \right] = 0,
\end{align*}

and also by \Cref{assum:nonnormal}.3,
\begin{align*}
  \left|\sqrt{n} \E \left[ c_i^{k_1} \left(\theta_i - k_3 \E\left[\theta_i\right]\right)^{k_2}\bar\epsilon_i^2 \right] \right| \leq \E \left[ \left| \left(\theta_i - k_3 \E\left[\theta_i\right]\right)^{k_2}  \right| \right] \sqrt{n \E \left[ \frac{1}{J_i^2} \right] \E \left[ \sigma_i^4\right]},
\end{align*}
where $n \E \left[\frac{1}{J_i^2}\right]=O(1)$. Therefore we have the results. For $c_{i,\text{HO}}$, the proof is similar.
   \item  \begin{enumerate}
    \item Follows from the finite fourth moment and the inequality below for any $s>0$, 
    \begin{align*}
      &\Pr \left( \left| \sqrt{n} \left( \bar\theta - \E\left[ \theta_i\right] \right) \right| > s  \right) \leq \frac{\E\left[\theta_i^2\right]}{s^2}.  &\text{(Chebyshev's inequality)}
    \end{align*}
    \item Follows from the property of $\bar\theta$ and the second property of \ref{item:moment}:
    \begin{align*}
      \bar X = \bar\theta + \bar\epsilon = \E\left[\theta_i\right] + O_p(n^{-1/2}) + o_p(n^{-1/2}) = \E\left[\theta_i\right] + O_p(n^{-1/2}).
    \end{align*}
    \item By Chebyshev's inequality and \eqref{eq:var_nonnormal},
    \begin{align*}
      \Pr \left(\sqrt{n}\left| \frac{1}{n} \sum_{i=1}^{n} \frac{1}{J_i} \hat\sigma_i^2 - \E \left[ \frac{1}{J_i} \hat\sigma_i^2\right]  \right| > s \right) \leq \frac{\E \left[\frac{1}{J_i^2} \left(  \frac{K\sigma_i^4}{J_i}+ \sigma_i^4 \right)\right]}{s^2} \to 0,
    \end{align*}
    and also
    \begin{align*}
      \sqrt{n}\E \left[\frac{1}{J_i} \hat\sigma_i^2 \right] \leq \sqrt{n\E \left[ \frac{1}{J_i^2} \right] \E \left[\sigma_i^4\right]} = O(1).
    \end{align*}
    Thus
    \begin{align*}
      \frac{1}{n}\sum_{i=1}^{n} \frac{1}{J_i} \hat \sigma_i^2 = O_p(n^{-1/2}).
    \end{align*}
   \end{enumerate}
   \item \begin{enumerate}
    \item Since $u_i \ind \epsilon_{i,j} \mid \theta_i$, by Chebyshev's inequality,
    \begin{align*}
      \Pr \left( \sqrt{n} \left| \frac{1}{n} \sum_{i=1}^{n} c_i^k u_i \right|>s \right) \leq \frac{\E \left[ c_i^{2k} u_i^2 \right]}{s^2} \leq \frac{\sigma^2_u}{s^2}
    \end{align*}
    Thus 
    \begin{align*}
      \frac{1}{n} \sum_{i=1}^{n} c_i^k u_i = O_p(n^{-1/2}).
    \end{align*}
    \item By Chebyshev's inequality,
    \begin{align*}
      \Pr \left( \sqrt{n} \left| \frac{1}{n} \sum_{i=1}^{n} c_i^k u_i \bar\epsilon_i \right|>s \right) \leq \frac{\E \left[ c_i^{2k} u_i^2 \bar\epsilon_i^2\right]}{s^2} \leq \frac{\sigma_u^2}{s^2} \sqrt{\E \left[\frac{1}{J_i^4 }\right] \E \left[\sigma_i^8\right]} \to 0.
    \end{align*}
    Thus 
    \begin{align*}
      \frac{1}{n} \sum_{i=1}^{n} c_i^k u_i\bar\epsilon_i = o_p(n^{-1/2}).
    \end{align*}
   \end{enumerate}
   For $c_{i,\text{HO}}$, the proof is similar.
  \end{enumerate}
\end{proof}

\begin{lemma}
  \label{lemma:sigma_corr_J}
  Under \Cref{assum:corr_J}, \Cref{assum:corr_J_large_J}, we have
  \begin{align*}
    \frac{1}{n} \sum_{i=1}^{n} c_i^2 \bar\epsilon_i^2 = \frac{1}{n} \sum_{i=1}^{n} c_i^2 \frac{1}{J_i}\hat\sigma_i^2 + o_p(n^{-1/2}).
  \end{align*}
  For $c_{i,\HO}$, the result is similar.
\end{lemma}
\begin{proof}[Proof of \Cref{lemma:sigma_corr_J}]
  First, notice that the difference
  \begin{align*}
    \bar \epsilon_i^2  - \frac{1}{J_i} \hat\sigma_i^2&= \bar\epsilon_i^2 - \frac{1}{J_i\left(J_i-1\right)} \sum_{j=1}^{J_i} \epsilon_{i,j}^2 + \frac{1}{J_i-1}\bar\epsilon_i^2 \\
    &= \frac{1}{J_i\left(J_i -1\right)} \left(\sum_{j=1}^{J_i}\epsilon_{i,j} \right)^2 - \frac{1}{J_i\left(J_i -1\right)} \sum_{j=1}^{J_i} \epsilon_{i,j}^2 \\
    &= \frac{2}{J_i \left(J_i -1\right)} \sum_{k\leq j} \epsilon_{i,k}\epsilon_{i,j}.
  \end{align*}
  Then the properties of the difference are
  \begin{align*}
    &\E \left[ \bar\epsilon_i^2 - \frac{1}{J_i} \hat\sigma_i^2 \right] = 0,\\
    &\E \left[ \left(\bar\epsilon_i^2 - \frac{1}{J_i} \hat\sigma_i^2\right)^2 \mid J_i, \sigma_i^2\right] = \frac{2\sigma_i^4}{J_i\left(J_i -1\right)}.
  \end{align*}
  For $c_i^2$, by the mean value theorem,
  \begin{align*}
    c_i^2 &= \left( \frac{V}{V + \frac{1}{J_i}\hat\sigma_i^2}\right)^2\\
    &= 1 - 2\frac{1}{VJ_i} \hat\sigma_i^2 + 3\frac{1}{\left(1 + \omega_i\right)^4} \frac{1}{V^2J_i^2}\left(\hat\sigma_i^2\right)^2,
  \end{align*}
  where $\omega_i$ is between $0$ and $\frac{1}{VJ_i}\hat\sigma_i^2$. Therefore,

  \begin{align*}
    &\quad \abs{\sqrt{n}\E\left[ c_i^2 \left( \bar\epsilon_i^2 - \frac{1}{J_i}\hat\sigma_i^2\right) \right]}\\
    &\leq \abs{\sqrt{n} \E\left[  \left( \bar\epsilon_i^2 - \frac{1}{J_i}\hat\sigma_i^2\right) \right]} + \frac{2}{V} \abs{\sqrt{n} \E \left[ \frac{1}{J_i} \hat\sigma_i^2 \left( \bar\epsilon_i^2 - \frac{1}{J_i}\hat\sigma_i^2\right) \right]} \\
    &\quad + \frac{3}{V^2}\abs{\sqrt{n} \E \left[ \frac{1}{\left(1 + \omega_i\right)^4} \frac{1}{J_i^2}\left(\hat\sigma_i^2\right)^2 \left( \bar\epsilon_i^2 - \frac{1}{J_i}\hat\sigma_i^2\right) \right]} \\
    &\leq 0 + \frac{2}{V} \sqrt{n \E\left[\frac{1}{J_i^2} \left(\hat\sigma_i^2 \right)^2 \right] \E\left[\left(\bar\epsilon_i^2 - \frac{1}{J_i}\hat\sigma_i^2  \right)^2  \right] }\\
    &\quad + \frac{3}{V^2} \sqrt{n \E\left[\frac{1}{J_i^4} \left(\hat\sigma_i^2\right)^4 \right] \E\left[\left(\bar\epsilon_i^2 - \frac{1}{J_i}\hat\sigma_i^2  \right)^2  \right] } &\text{(Cauchy-Schwarz)}\\
    &\leq \frac{2}{V} \sqrt{\sqrt{n} \E\left[\frac{1}{J_i^2}\left( \frac{K}{J_i} \sigma_i^4 + \sigma_i^4  \right)   \right]\sqrt{n} \E\left[ \frac{2\sigma_i^4}{J_i \left(J_i -1\right)}  \right] }\\
    &\quad + \frac{3}{V^2} \sqrt{\sqrt{n} \E\left[\frac{1}{J_i^4}\left(  \hat\sigma_i^2\right)^4   \right]\sqrt{n} \E\left[ \frac{2\sigma_i^4}{J_i \left(J_i -1\right)}  \right] }.\\
    &\leq \frac{2}{V} \sqrt{\sqrt{n \E\left[\frac{1}{J_i^6}  \right] \E\left[ K^2\sigma_i^8\right] + n\E\left[\frac{1}{J_i^4} \right]\E\left[\sigma_i^8\right]} \sqrt{n \E\left[\frac{1}{J_i^2\left(J_i -1\right)^2}\right] \E\left[4\sigma_i^8\right]} }\\
    &\quad + \frac{3}{V^2} \sqrt{\sqrt{n \E\left[\frac{1}{J_i^4}\left(  \hat\sigma_i^2\right)^4   \right]} \sqrt{n \E\left[\frac{1}{J_i^2\left(J_i -1\right)^2}\right] \E\left[4\sigma_i^8\right]} }. &\text{(Cauchy-Schwarz)}\\
  \end{align*}
  Since we have \Cref{assum:corr_J_large_J} and that
  \begin{align*}
    \sqrt{n} \E \left[\frac{1}{J_i^4} \left(\hat\sigma_i^2\right)^4 \right] \to 0,
  \end{align*}
  from \Cref{lemma:properties_J}.\ref{item:angelova}, the above converges to $0$.
  Thus, the expectation
  \begin{align*}
    \sqrt{n} \E\left[ c_i^2 \left( \bar\epsilon_i^2 - \frac{1}{J_i}\hat\sigma_i^2\right) \right] \to 0.
  \end{align*}
  Next, we show the convergence to expectations by Chebyshev's inequality:
  \begin{align*}
    &\quad \Pr \left( \sqrt{n}\left| \frac{1}{n} \sum_{i=1}^{n} c_i^2 \left( \bar\epsilon_i^2 - \frac{1}{J_i}\hat\sigma_i^2\right) - \E \left[ c_i^2 \left( \bar\epsilon_i^2 - \frac{1}{J_i}\hat\sigma_i^2\right)  \right] \right| > s  \right)\\
    &\leq \frac{\E \left[ c_i^4 \left( \bar\epsilon_i^2 - \frac{1}{J_i}\hat\sigma_i^2\right)^2  \right]}{s^2} \leq \frac{\E \left[  \left( \bar\epsilon_i^2 - \frac{1}{J_i}\hat\sigma_i^2\right)^2  \right]}{s^2} = \frac{\E\left[ \frac{2\sigma_i^4}{J_i \left(J_i -1\right)}  \right]}{s^2}  \to 0.
  \end{align*}
  And the proof for $c_i$ is complete.

  For $c_{i,\HO}$, the proof is similar.
\end{proof}

\begin{lemma}
  \label{lemma:V_corr_J}
  Under \Cref{assum:corr_J}, \Cref{assum:corr_J_large_J}, we have
  \begin{align*}
    \hat V = V + O_p(n^{-1/2}).
  \end{align*}
\end{lemma}
\begin{proof}[Proof of \Cref{lemma:V_corr_J}]
  By \Cref{lemma:properties_J},
  \begin{align}
    \hat{V} &= \frac{1}{n} \sum_{i=1}^n \left( \bar{X}_i- \bar{X} \right)^2 - \frac{n-1}{n^2} \sum_{i=1}^{n} \frac{1}{J_i} \hat \sigma^2_i \nonumber\\
    &= \frac{1}{n} \sum_{i=1}^{n} \left(\theta_i - \E\left[\theta_i\right]\right)^2 + \frac{1}{n} \sum_{i=1}^{n} \bar \epsilon_i^2 - \left( \bar\theta - \E\left[\theta_i\right] \right)^2 - \bar\epsilon^2 + \frac{2}{n} \sum_{i=1}^{n} \left( \theta_i - \E \left[\theta_i\right]\right)\bar\epsilon_i  \nonumber\\
    &\quad - 2 \left(\bar\theta-\E \left[\theta_i\right]\right)\bar\epsilon  - \frac{n-1}{n^2} \sum_{i=1}^{n} \frac{1}{J_i} \hat\sigma_i^2 \nonumber\\
    &= V + O_p(n^{-1/2}) + O_p(n^{-1/2}) - O_p(n^{-1}) - o_p(n^{-1}) + o_p(n^{-1/2}) \nonumber \\
    &\quad - o_p(n^{-1}) - O_p(n^{-1/2}) \nonumber\\
    &= V + O_p(n^{-1/2}). \nonumber
  \end{align}
\end{proof}

\begin{lemma}
  \label{lemma:V_corr_J_new_V}
  Under \Cref{assum:corr_J}, \Cref{assum:corr_J_large_J}, we have
  \begin{align*}
    \hat V_{\text{HO}} = V + O_p(n^{-1/2}).
  \end{align*}
\end{lemma}
\begin{proof}[Proof of \Cref{lemma:V_corr_J_new_V}]
  By conventional law of large number and central limit theorem arguments,
  \begin{align*}
    \hat V_{\text{HO}} &= \frac{1}{n}\sum_{i=1}^{n} \left( \bar X_{1,i}- \bar X_1\right) \left( \bar X_{2,i}- \bar X_2\right)\\
    &= V + O_p(n^{-1/2})
  \end{align*}
\end{proof}

\begin{proposition}
  \label{prop:beta_corr_J}
  Suppose the asymptotic framework satisfies \Cref{assum:corr_J_large_J}. Then under \Cref{assum:corr_J} we have
\begin{align*}
  \sqrt{n} \left(\hat{\beta} - \beta\right) \to_d N \left(0, \frac{\E\left[ u_i^2 \left(\theta_i - \E \left[\theta_i\right]\right)^2 \right]}{V^2}\right).
\end{align*}
\end{proposition}
\begin{proof}[Proof of \Cref{prop:beta_corr_J}]
  Based on a similar set of lemmas to \Cref{sec:appendix-lemmas-nonnormal}, the claim follows from the same arguments as in \Cref{prop:beta}.
\end{proof}

\begin{proposition}[\Cref{prop:HO_in_text}]
  \label{prop:beta_HO_corr_J}
  Suppose the asymptotic framework satisfies \Cref{assum:corr_J_large_J}. Then under \Cref{assum:corr_J}, there exist cases where
  \begin{align*}
    \sqrt{n} \left( \hat\beta_{\text{HO}} - \beta \right) \to_d N\left(\frac{\beta}{\Var\left(\theta_i\right)},\frac{\E\left[u_i^2 \left(\theta_i - \E\left[\theta_i\right]\right)^2\right]}{\left(\Var\left(\theta_i\right)\right)^2}\right).
  \end{align*}
\end{proposition}
\begin{proof}[Proof of \Cref{prop:beta_HO_corr_J}]
  For the shrinkage weight
  \begin{align*}
    c_{i,\text{HO}} = \frac{V}{\frac{1}{J_i}\E \left[\sigma_i^2\right] + V},
  \end{align*}
  we first show properties of the regression coefficients $\hat\beta_{\text{HO},c}$, and then show that $\sqrt{n}\left(\hat\beta_{\text{HO}}-\beta\right) \to_p \sqrt{n} \left(\hat\beta_{\text{HO},c}-\beta\right)$.

  For $\hat\beta_{\text{HO},c}$, based on a similar set of lemmas to \Cref{sec:appendix-lemmas-nonnormal}, from the same derivation as the proof of \Cref{lemma:beta_c}, we have
  \begin{align*}
    \sqrt{n}\left(\hat\beta_{\text{HO},c} - \beta\right) &=  \frac{\beta\sqrt{n}T_{1,n} - \beta\sqrt{n}T_{2,n} + \sqrt{n}T_{3,n}}{T_{2,n}},
  \end{align*}
  where the denominator
  \begin{align*}
    T_{2,n} = V + O_p(n^{-1/2}),
  \end{align*} 
  and the numerator
  \begin{align*}
    &\quad\beta\sqrt{n}T_{1,n} - \beta\sqrt{n}T_{2,n} + \sqrt{n}T_{3,n}\\
    &= \frac{1}{\sqrt{n}}\sum_{i=1}^{n} \left[ \beta c_{i,\text{HO}} \left(\theta_i - \E\left[\theta_i\right]\right)^2 - \beta c_{i,\text{HO}}^2 \left[\left( \theta_i - \E\left[\theta_i\right] \right)^2 + \frac{1}{J_i}\hat\sigma_i^2  \right] + c_{i,\text{HO}} \left(\theta_i - \E\left[\theta_i\right]\right)u_i \right] + o_p(1)\\
    &= \frac{1}{\sqrt{n}}\sum_{i=1}^{n} \left[ \beta c_{i,\text{HO}} \left(\theta_i - \E\left[\theta_i\right]\right)^2 - \beta c_{i,\text{HO}}^2 \left[\left( \theta_i - \E\left[\theta_i\right] \right)^2 + \frac{1}{J_i}\E\left[\sigma_i^2\right]  \right] + c_{i,\text{HO}} \left(\theta_i - \E\left[\theta_i\right]\right)u_i \right] \\
    &\quad - \beta \frac{1}{n} \sum_{i=1}^{n} c_{i,\text{HO}}^2 \frac{\sqrt{n}}{J_i}\left( \hat\sigma_i^2 - \E \left[\sigma_i^2\right]  \right) + o_p(1)\\
    &= \frac{1}{\sqrt{n}}\sum_{i=1}^{n}\xi_i - \beta \frac{1}{n} \sum_{i=1}^{n} c_{i,\text{HO}}^2 \frac{\sqrt{n}}{J_i}\left( \hat\sigma_i^2 - \E \left[\sigma_i^2\right]  \right) + o_p(1).
  \end{align*}
  For $\xi_i$, similar to \Cref{lemma:beta_c}, we have $\E \left[\xi_i\right]=0$, $\E \left[\xi_i^2\right] \to \E\left[u_i^2 \left(\theta_i - \E\left[\theta_i\right]\right)^2\right]$. However, the second term of the numerator might not disappear.

  Next, following from the same arguments as \Cref{prop:beta}, we have $\sqrt{n}\left(\hat\beta_{\text{HO}}-\beta\right) \to_p \sqrt{n} \left(\hat\beta_{\text{HO},c}-\beta\right)$.

  Now, focusing on the bias term, by Chebyshev's inequality and \eqref{eq:var_nonnormal}, for any $s >0$,
  \begin{align*}
    &\quad \Pr\left( \left| \frac{1}{n}\sum_{i=1}^{n} c_{i,\text{HO}}^2 \frac{\sqrt{n}}{J_i} \left( \hat\sigma_i^2 - \E\left[\sigma_i^2\right] \right) - \E\left[ c_{i,\text{HO}}^2 \frac{\sqrt{n}}{J_i} \left( \hat\sigma_i^2 - \E\left[\sigma_i^2\right] \right) \right] \right| >s    \right)\\
    &\leq \frac{\E\left[ \frac{1}{J_i^2} \left( \hat\sigma_i^2 - \E\left[\sigma_i^2\right] \right)^2 \right]}{s^2}\\
    &\leq \frac{ \E\left[ \frac{1}{J_i^2}\left( \frac{K\sigma_i^4}{J_i} + \left(\sigma_i^2 - \E\left[\sigma_i^2\right]\right)^2   \right) \right] }{s^2} \to 0.
  \end{align*}
  Suppose with equal probabilities, $\sigma_i^2 = 12 \gamma V, J_i = \lfloor 2 \sqrt{n} \rfloor$ or $\sigma_i^2 =8 \gamma V, J_i = \lfloor \frac{2}{3} \sqrt{n} \rfloor$, $\gamma>0$. Then we have
  \begin{align*}
    n\E\left[ \frac{1}{J_i^2}  \right] \sim \frac{1}{2}n * \frac{1}{\left(2\sqrt{n}\right)^2} + \frac{1}{2}n * \frac{1}{\left(\frac{2}{3}\sqrt{n}\right)^2} \to \frac{5}{4},
  \end{align*}
  \begin{align*}
    &\quad \E\left[ c_{i,\text{HO}}^2 \frac{\sqrt{n}}{J_i} \left( \hat\sigma_i^2 - \E\left[\sigma_i^2\right] \right) \right]\\
    &= \E\left[ c_{i,\text{HO}}^2 \frac{\sqrt{n}}{J_i} \left( \sigma_i^2 - \E\left[\sigma_i^2\right] \right)    \right]\\
    &\sim \frac{2}{2} \left( \frac{V}{\frac{10}{2\sqrt{n}} + V} \frac{\sqrt{n}}{2\sqrt{n}} -  \frac{V}{\frac{10}{\frac{2}{3}\sqrt{n}} + V} \frac{\sqrt{n}}{\frac{2}{3}\sqrt{n}} \right) \gamma V\\
    &\to -\gamma V.
  \end{align*}
  Then by the continuous mapping theorem, in this case we have
  \begin{align*}
    \sqrt{n} \left(\hat\beta_{\text{HO}} - \beta\right) \to_d N\left( \gamma \beta , \frac{\E\left[u_i^2 \left(\theta_i - \E\left[\theta_i\right]\right)^2\right]}{V^2} \right),
  \end{align*}
  which proves the first result of the theorem.

  To prove the second result of the theorem, it suffices to show that if $\sigma_i^2$ and $J_i$ are independent, then
  \begin{align*}
    \frac{1}{n} \sum_{i=1}^{n} c_{i,\text{HO}}^2 \frac{\sqrt{n}}{J_i}\left( \hat\sigma_i^2 - \E \left[\sigma_i^2\right]  \right) = o_p(1).
  \end{align*}
  Since $c_{i,\HO}$ only depends on $J_i$, by the independence, the expectation
  \begin{align*}
    \E\left[ c_{i,\text{HO}}^2 \frac{\sqrt{n}}{J_i}\left( \hat\sigma_i^2 - \E \left[\sigma_i^2\right]  \right) \right] =0.
  \end{align*}
  Then, by Chebyshev's inequality, 
  \begin{align*}
    &\quad \Pr\left( \left| \frac{1}{n}\sum_{i=1}^{n} c_{i,\text{HO}}^2 \frac{\sqrt{n}}{J_i} \left( \hat\sigma_i^2 - \E\left[\sigma_i^2\right] \right) - \E\left[ c_{i,\text{HO}}^2 \frac{\sqrt{n}}{J_i} \left( \hat\sigma_i^2 - \E\left[\sigma_i^2\right] \right) \right] \right| >s    \right)\\
    &\leq \frac{\E\left[ c_{i,\text{HO}}^4 \frac{1}{J_i^2} \left( \hat\sigma_i^2 - \E\left[\sigma_i^2\right] \right)^2 \right]}{s^2}\\
    &\leq \frac{\E\left[  \frac{1}{J_i^2}\left( \frac{K\sigma_i^4}{J_i} + \left(\sigma_i^2 - \E\left[\sigma_i^2\right]\right)^2   \right) \right] }{s^2} \to 0,
  \end{align*}
  which proves the second result of the theorem.
\end{proof}

\end{document}